\documentclass[a4paper,UKenglish,cleveref, autoref, thm-restate,numberwithinsect]{lipics-v2021}

\usepackage[T1]{fontenc}
\usepackage{graphicx}

\usepackage[utf8]{inputenc}

\usepackage{longtable}
\usepackage{rotating}
\usepackage{textcomp}
\usepackage{version}
\usepackage{hyperref}
\usepackage{amsmath}
\usepackage{amsfonts}
\usepackage{amssymb,dsfont}
\usepackage{mathtools}

\usepackage[inline, shortlabels]{enumitem}
\usepackage{appendix}
\usepackage[svgnames]{xcolor}
\usepackage{tikz}
\usetikzlibrary{automata,positioning}
\usetikzlibrary{arrows,automata}
\usetikzlibrary{shapes,snakes}
\usetikzlibrary{fit}
\usepackage{caption}
\usepackage{subcaption}
\usepackage{MnSymbol,wasysym}
\usepackage{ stmaryrd }
\usepackage{environ}
\usepackage{float}
\usepackage{comment}
\usepackage{framed}

\usepackage{tcolorbox}
\usetikzlibrary{shapes.geometric,arrows}
\usetikzlibrary{decorations.pathmorphing}
\usepackage{lineno}
\nolinenumbers

\usepackage[textsize=small]{todonotes}
\newcommand{\lug}[1]{}
\newcommand{\nas}[1]{}
\newcommand{\ars}[1]{}
\newcommand{\lugtext}[1]{}
\newcommand{\nastext}[1]{}
\newcommand{\arstext}[1]{}

\usepackage{stackengine}
\stackMath
\usepackage{ifthen}

\newcommand{\expspace}{\ensuremath{\textsc{ExpSpace}}}

\newenvironment{sketch-proof} {\textit{Sketch of proof.}}{\hfill$\square$\\ }

 \newcommand{\PP}{\ensuremath{P}}   
\newcommand{\CC}{\ensuremath{\mathcal{C}}}

\renewcommand{\SS}{\ensuremath{\mathcal{S}}}

\newcommand{\nat}{\mathbb{N}}
\newcommand{\set}[1]{\{#1\}}

\newcounter{sarrow}

\newcommand{\squig}{{\scriptstyle\sim\mkern-3.9mu}}

\newcommand{\rsquigend}{{\scriptstyle\rule{.1ex}{0ex}\rhd}}
\newcounter{sqindex}
\newcommand\squigs[1]{%
	\setcounter{sqindex}{0}%
	\whiledo {\value{sqindex}< #1}{\addtocounter{sqindex}{1}\squig}%
}
\newcommand\rsquigarrow[2]{%
	\mathbin{\stackon[2pt]{\squigs{#2}\rsquigend}{\scriptscriptstyle{#1\,}}}%
}

\newcommand{\trans}{\rightarrow}
\newcommand{\transup}[1]{\xrightarrow{#1}}

\newcommand{\abtransup}[1]{\xRightarrow{#1}}
\newcommand{\abconftrans}{\ensuremath{\rsquigarrow{}{3}}}
\newcommand{\abconftransup}[1]{\ensuremath{\rsquigarrow{\ensuremath{#1}}{3}}}

\newcommand{\Reachability}{\textsc{Reachability}}

\newcommand{\Target}{\textsc{Synchro}}
\newcommand{\RepeatedCover}{\textsc{Repeated Coverability}}

\newcommand{\pspace}{\textsc{PSpace}}

\newcommand{\qinit}{\ensuremath{q_{\textit{in}}}}

\renewcommand{\set}[1]{\ensuremath{\{#1\}}}

\newcommand{\Op}{\ensuremath{\mathsf{Op}}}
\newcommand{\implem}[1]{\ensuremath{\llbracket#1\rrbracket}}

\newcommand{\printof}[1]{\ensuremath{\mathbf{print}(#1)}}
\newcommand{\abconffrom}[2]{\ensuremath{\mathbf{repr}(#1,#2)}}
\newcommand{\Counters}{\ensuremath{\mathsf{X}}}
\newcommand{\counter}{\ensuremath{\mathsf{x}}}
\newcommand{\lab}[1]{\ensuremath{\textbf{lab}(#1)}}
\newcommand{\aproc}{\ensuremath{e}}
\newcommand{\print}{\ensuremath{\textsf{pr}}}
\newcommand{\summary}{\ensuremath{S}}
\newcommand{\nextsummary}[1]{\ensuremath{{#1}^{\texttt{+1}}}}

\newcommand{\arrived}{\ensuremath{\texttt{Done}}}
\newcommand{\conf}[1]{\overset{\bullet}{#1}}
\newcommand{\confstate}[1]{\overset{\bullet}{#1}}
\newcommand{\garbagestate}{\ensuremath{q_u}}
\newcommand{\setof}[1]{\ensuremath{\textbf{set}(#1)}}
\newcommand{\new}[1]{\color{blue}#1\color{black}}
\newcommand{\old}[1]{}

\newcommand{\nbagents}[1]{\ensuremath{\textbf{\#proc}(#1)}}
\newcommand{\nextactionindex}[1]{\ensuremath{\textbf{na-index}(#1)}}
\newcommand{\nextactionstate}[1]{\ensuremath{\textbf{na-state}(#1)}}
\newcommand{\nextaction}[1]{\ensuremath{\textbf{na}(#1)}}

\newcommand{\Loc}{\ensuremath{\textsf{Loc}}}
\newcommand{\CoherentSets}{\ensuremath{\textsf{CoSets}}}
\newcommand{\inc}[1]{\ensuremath{#1}\tiny ++}
\newcommand{\dec}[1]{\ensuremath{#1}- -}

\newcommand{\nextactionindexset}[1]{\ensuremath{\textbf{na-index-set}(#1)}}

\makeatletter
\newcommand{\problemtitle}[1]{\gdef\@problemtitle{#1}}% Store problem title
\newcommand{\probleminput}[1]{\gdef\@probleminput{#1}}% Store problem input
\newcommand{\problemquestion}[1]{\gdef\@problemquestion{#1}}% Store
\newcommand{\problemquestionline}[1]{\gdef\@problemquestionline{#1}}% Store
\NewEnviron{decproblem}{
	\problemtitle{}\probleminput{}\problemquestion{}\problemquestionline{}% Default input is empty
	\BODY% Parse input
	\par\addvspace{.2\baselineskip}
	\noindent
	\fbox{
		
		\centering
		\begin{tabular}{rl}
			\multicolumn{2}{l}{ \@problemtitle} \\% Title
			\hline
			\rule{0pt}{3ex}
			\normalsize \textbf{Input:} &  \@probleminput \\% Input
			\normalsize \textbf{Question:} &  \@problemquestion\\% Question
			\normalsize & \@problemquestionline
		\end{tabular}
	}
	\par\addvspace{.2\baselineskip}
	
}

\title{Wait-Only Broadcast Protocols are Easier to Verify}

\author{Lucie Guillou}{IRIF, CNRS, Universit\'e Paris Cit\'e,
  France}{guillou@irif.fr}{https://orcid.org/0000-0002-6101-2895}{}
\author{Arnaud Sangnier}{DIBRIS, Università di Genova, Italy}{arnaud.sangnier@unige.it}{https://orcid.org/0000-0002-6731-0340}{}
\author{Nathalie Sznajder}{LIP6, CNRS, Sorbonne Universit\'e,
  France}{nathalie.sznajder@lip6.fr}{https://orcid.org/0000-0002-4199-2443}{}
\authorrunning{L. Guillou and A. Sangnier and N. Sznajder}

\Copyright{L. Guillou and A. Sangnier and N. Sznajder}



\relatedversion{}

\EventEditors{}
\EventNoEds{0}
\EventLongTitle{}
\EventShortTitle{}
\EventAcronym{}
\EventYear{2023}
\EventDate{}
\EventLocation{}
\EventLogo{}
\SeriesVolume{}
\ArticleNo{}
\begin{document}

\maketitle
\abstract{We study networks of processes that all execute the same finite-state protocol and communicate via broadcasts. We are interested in two problems with a parameterized number of processes: the synchronization problem which asks whether there is an execution which puts all processes on a given state; and the repeated coverability problem which asks if there is an infinite execution where a given transition is taken infinitely often. Since both problems are undecidable in the general case, we investigate those problems when the protocol is \emph{Wait-Only}, i.e., it has no state from which a process can both broadcast and receive messages. We establish that the synchronization problem becomes Ackermann-complete, and the repeated coverability problem is in \expspace\,  and \pspace-hard.}

\section{Introduction}
Distributed systems are at the core of modern computing, and are widely used in critical applications such as sensor networks, distributed databases and
collaborative control systems. These systems rely on protocols to enable communication, maintain coherence and coordination between the different processes. 
Because of their distributed nature, such protocols have proved to be error-prone. As a result, the formal verification of distributed systems has become an essential area of research. 
Formal verification of
distributed systems presents unique challenges compared to formal verification of centralized systems. One of the most significant issue is the state explosion 
problem: the 
behavior of multiple processes that execute concurrently and exchange information often lead to state spaces that grow exponentially with the number of processes, making the analysis highly challenging. 
When the systems are parameterized, meaning designed to operate for an arbitrary number of processes, which is often the case in distributed protocols, classical techniques
are not useful anymore. The difficulty shifts from state explosion to dealing with an infinite number of system configurations, which leads to undecidability in the 
general
case~\cite{apt86limits}. Despite these challenges, verification becomes more tractable in certain restricted settings. For instance, in parameterized systems, the behavior of individual processes needs not always to be explicitly represented, which can mitigate state explosion. This has motivated researchers to focus on specific subclasses of systems or communication models where decidability can be achieved without sacrificing too much expressiveness.

One such restriction involves protocols where processes are anonymous (i.e., without identities), and communicate via simple synchronous mechanisms, such as rendezvous~\cite{german92}, where two processes exchange a message synchronously. Several variants of this model have been 
studied, such as communication via a shared register containing some value from a finite set~\cite{esparza-param-cav13}, or non-blocking rendezvous~\cite{guillou-safety-concur23,DelzannoRB02},
where a sender is not prevented from sending a message when no process is ready to receive it. In all these cases, the 
processes execute the same protocol, which is described by a finite automaton. 

A more expressive model is broadcast protocols, introduced by Emerson and Namjoshi~\cite{emerson98model-checking}. In these protocols, a process can send a broadcast message that is received simultaneously by all other processes (or by all neighboring processes in a network).
 Several key verification problems arise in this context, including the coverability problem, which asks whether a process can eventually reach a 
 specific (bad) state starting from an initial
 configuration. Another important property is synchronization, which asks whether 
 all processes can simultaneously reach a  given state. Liveness properties, like determining whether infinite executions exist, ensure that certain behaviors  occur
 indefinitely. 
 %Broadcast protocols are more expressive than systems involving rendezvous communication. For instance, one can simulate rendezvous by broadcasts~\cite{EsparzaFM99}. Unsurprisingly, 
 %verification of properties become more difficult. While coverability and synchronization can be solved in polynomial time for rendezvous protocols~\cite{german92}, it becomes respectively 
 While broadcast protocols allow more powerful communication than rendezvous protocols, this added expressiveness comes at the cost of increased verification complexity. Indeed, rendezvous communication can be simulated using broadcasts~\cite{EsparzaFM99}. However, verification becomes significantly harder: while coverability and synchronization can be solved in polynomial time for rendezvous protocols~\cite{german92}, they become respectively
 Ackermann-complete~\cite{EsparzaFM99,schmitz-power-concur13} and undecidable (we show it in this paper) for broadcast protocols. Liveness properties
 are also undecidable for 
 broadcast protocols~\cite{EsparzaFM99}. The challenge in verifying broadcast protocols can be understood through their relationship with Vector Addition Systems with States (VASS). Rendezvous-based protocols can be encoded in a VASS, where counters track the number of processes in each state. A rendezvous is simulated by decreasing the counters corresponding to the sender and receiver states and increasing the counters for their post-communication states. 
 While using a VASS for protocols using rendezvous is certainly not the way to go due to the complexity
 gap between problems on VASS and the existing known polynomial time algorithms to solve verification problems, it is interesting to note that this encoding fails for broadcast protocols because a broadcast message must be received by all processes in a given state, requiring a bulk transfer of a counter value—something not expressible in classical VASS without adding powerful operations such as transfer arcs, which leads to an undecidable reachability problem~\cite{valk-self-icalp78}. 
 %On top of that, in our model, a message can be broadcast even if no process is 
% ready to receive it. This is also not expressible in a VASS since it would require to add a non deterministic choice between only decrementing the state of the sender
% (no process receives it) and decrement the state of the sender and empty all the states of the receivers. Ensuring that the reception of the broadcast actually
% occurs when some processes are ready to receive it amounts to testing some counter to 0, which gives the power of a counter machine. This last feature was also
% present in the model of nonblocking rendezvous introduced in~\cite{guillou-safety-concur23}. 
 However, in some cases, verification of broadcast protocols can become easier, complexity-wise. One example is the setting of Reconfigurable Broadcast Networks, in which communication between processes can be lossy~\cite{DelzannoSZ10}. Liveness properties become decidable in that case~\cite{DelzannoSZ10} and even polynomial time~\cite{ChiniMS22}.
 
 Another such case is given by a syntactic restriction known as \emph{Wait-Only protocols}, introduced in~\cite{guillou-safety-concur23} in the context of non-blocking rendezvous communication.  In Wait-Only protocols, a process cannot both send and receive a message from the same state. From a practical perspective, Wait-Only protocols were proposed to model the non-blocking rendezvous semantics used in Java’s \texttt{wait}/\texttt{notify}/\texttt{notifyAll} synchronization mechanism and C’s \texttt{wait}/\texttt{signal}/\texttt{broadcast} operations with conditional variables, commonly found in multi-threaded programs. In both languages, threads can be awakened by the reception of a message. This naturally leads to distinguishing between action and waiting states: a sleeping thread cannot act on its own and can only be awakened by a signal.
This restriction simplifies the possible behaviours of the system, potentially reducing the complexity of the verification. Actually, the coverability problem for Wait-Only broadcast protocols becomes \pspace-complete~\cite{guillou-safety-petrinets24}. Indeed, when processes are in a state where they can receive broadcasts, they cannot send messages, meaning they will move together to the next state, reducing the need for precise counting. Moreover, when a process is in a broadcasting state, it remains there until it decides to send a message, simplifying execution reconstruction and enabling better verification algorithms.
By leveraging these properties, we design algorithms for the synchronization problem and liveness properties, obtaining new decidability 
results for these challenging problems.

\textbf{Contributions.} In this paper, we address two key verification problems in the context of Wait-Only broadcast protocols: the synchronization problem 
and the repeated coverability problem (a type of liveness property). Our main contributions are as follows: 
\begin{itemize}
\item Synchronization problem. We show that restricting to Wait-Only protocols allows us to 
regain decidability for synchronization. Specifically, we prove that the problem is Ackermann-complete via a two-way reduction to the reachability problem in Vector Addition Systems with States (VASS) (\cref{sec:synchro}). Furthermore, when the target state is an action state (a state that cannot receive messages), the problem becomes \expspace-complete (\cref{sec:synchro-action}). 

\item Repeated Coverability Problem: it asks whether a specific transition can occur infinitely often. We show that this problem is in 
\expspace\, and \pspace-hard (\cref{sec:repeated-cover}). %To the best of our knowledge, it is the first decidability result about a liveness property for broadcast networks.\nas{est-ce qu'on dit que c'est le premier resultat de decidabilite pour une liverness property pour les broadcast protocols?}\lug{ai ajouté la derniere phrase}
\end{itemize}

Proofs omitted due to space contraints can be found in Appendix.

\section{Model and Verification Problems}

In this section, we provide the description of a model for networks with broadcast communication together with some associated verification problems we are interested in. First we introduce some useful mathematical notations. We use $\nat$ to represent the set of natural numbers, $\nat_{>0}$ to represent $\nat \setminus \set{0}$, and for $n, m \in \nat$ with $n \leq m$, we denote by $[n, m]$ the set $\set{n, n+1, \dots, m}$. For a finite set $E$ and a natural $n>0$, the set $E^n$ represents the $n$-dimensional vectors with values in $E$ and for $v \in E^n$ and $1\leq i \leq n$, we use $v(i)$ to represent the $i$-th value of $v$. Furthermore, for such a vector $v$, we denote by $||v||$ its dimension $n$.

%% Let $\Sigma$ be a finite alphabet, define sets $\diamond \Sigma = \set{\diamond a \mid a \in \Sigma}$ for a symbol $\diamond \in \set{!!, ?}$. We denote
%% $\Op_\Sigma $ the set $  !!\Sigma \cup ?\Sigma$.
%% Let $n, m \in \nat$ such that $n \leq m$. We denote the set $\set{n, n+1, \dots, m}$ as $[n, m]$.

\subsection{Networks of Processes using Broadcast Communication}

In the networks we consider, all the processes execute the same protocol, given by a finite state machine. The actions of this machine can either broadcast a message $a$ to the other processes (denoted by $!!a$) or receive a message $a$ (denoted by $?a$). For a finite alphabet of messages $\Sigma$, we use the following notations: $!!\Sigma\overset{\mathrm{def}}{=}\set{!!a \mid a \in \Sigma}$, $?\Sigma\overset{\mathrm{def}}{=}\set{?a \mid a \in \Sigma }$ and $\Op_\Sigma \overset{\mathrm{def}}{=}!!\Sigma \cup ?\Sigma$. 

\begin{definition}
	A \emph{protocol} $\PP$ is a tuple $(Q, \Sigma, \qinit, T)$ where $Q$ is a finite set of states, $\Sigma$ is a finite set of messages, $\qinit \in Q$ is an {initial state}, and $T \subseteq Q \times\Op_\Sigma \times Q$ is a set of transitions.
\end{definition}

For all $q \in Q$, we define $R(q)$ as the set of messages that can be received from state $q$, given by $\set{a \in \Sigma \mid \exists q' \in Q, (q, ?a, q') \in T}$.
A protocol $\PP = (Q, \Sigma, \qinit, T)$ is \emph{Wait-Only} whenever for all $q \in Q$, either $\set{q' \in Q\mid (q, \alpha, q') \in T \text{ with } \alpha \in ?\Sigma} = \emptyset$, or $\set{q' \in Q\mid (q, \alpha, q') \in T \text{ with } \alpha \in !!\Sigma} = \emptyset$. In a Wait-Only protocol, a state $q \in Q$ such that $\set{q' \in Q\mid (q, \alpha, q') \in T \text{ with } \alpha \in ?\Sigma} \neq \emptyset$ is called a \emph{waiting state}. A state that is not a waiting state is an \emph{action state}.
%\nas{j'ai echange waiting et action}
%% We call a state respecting the first or both conditions an \emph{active} state and a state respecting the second condition a \emph{waiting} state.
We denote by $Q_W$ the set of waiting states and by $Q_A$ the set of action states (hence $Q_A=Q \setminus Q_W$). Observe that if $\qinit \in Q_W$, then no process is ever able to broadcast messages, for this reason we will always assume that $\qinit \in Q_A$.%\lug{ici remarque pr qinit}
%A protocol $\PP = (Q, \Sigma, \qinit, T)$ is \emph{Single-Wait-Only} (SWO) whenever it is Wait-Only and for all waiting state $q\in Q_W$, $|\set{q' \in Q\mid (q, \alpha, q') \in T \text{ with } \alpha \in ?\Sigma}| = 1$, i.e. the state has only one outgoing reception transition.

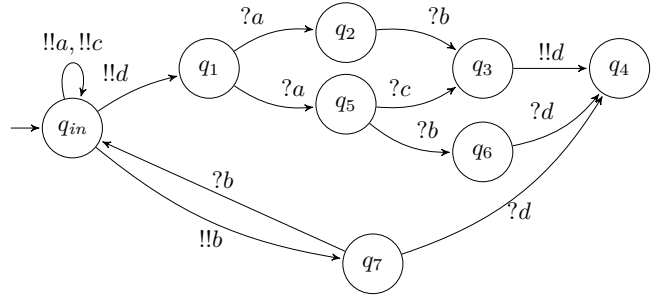
\begin{figure}
  
  \begin{center}
    \scalebox{0.9}{
      \tikzset{box/.style={draw, minimum width=4em, text width=4.5em, text centered, minimum height=17em}}

\begin{tikzpicture}[->, >=stealth', shorten >=1pt,node distance=2cm,on grid,auto, initial text = {}] 
	\node[state, initial] (q0) {$\qinit$};
	\node[state] (q1) [right = of q0, yshift = 25 ] {$q_1$};
	\node[state] (q2) [right = of q1, yshift=-15] {${q_5}$};
	\node[state] (q3) [right  = of q1, yshift = 15] {${q}_{2}$};
	\node[state] (q4) [right  = of q2, yshift = 15] {$q_{3}$};
	\node[state] (q5) [right  = of q2, yshift = -20] {$q_{6}$};
	\node[state] (q6) [right  = of q4] {$q_{4}$};
	\node[state] (q7) [below  = of q0, xshift = 125, yshift = 0] {$q_{7}$};
%	\node[state] (q8) [right  = of q7] {$q_{8}$};
	
	\path[->] 
	(q0) edge [loop above] node [] {$!!a,!!c$} ()
	edge [bend left = 10] node {$!!d$} (q1)
	edge [bend right = 15] node [below] {$!!b$} (q7)
	(q1) edge [bend right = 20] node {$?a$} (q2)
	edge [bend left = 20] node {$?a$} (q3)
	(q3) edge [bend left = 20] node {$?b$} (q4)
	(q2) edge [bend right = 20] node {$?b$} (q5)
	(q2) edge [bend right = 20] node {$?c$} (q4)
	(q4) edge [] node {$!!d$} (q6)
	(q7) edge [bend right = 20] node [below]{$?d$} (q6)
	(q5) edge [bend right = 20] node {$?d$} (q6)
	(q7) edge [bend left = 0] node [above] {$?b$} (q0)
	;
\end{tikzpicture}

%
%\begin{tikzpicture}[->, >=stealth', shorten >=1pt,node distance=2cm,on grid,auto, initial text = {}] 
%	\node[state, initial] (q0) {$\qinit$};
%	\node[state] (q1) [right = of q0, yshift = 25] {$q_1$};
%	\node[state] (q2) [right = of q1] {$q_2$};
%	\node[state] (q3) [right  = of q0, yshift = -25] {${q}_{2}$};
%	\node[state] (q4) [right  = of q3] {$q_{4}$};
%	\node[state, fill=green] (q5) [right  = of q4] {$q_{5}$};
%	
%	\path[->] 
%	(q0) edge [thick,bend right = 20] node  [below]{$\tau$} (q3)
%	edge [thick,bend left = 20] node  [above]{$!!c$} (q3)
%	edge [thick,bend left = 0] node  [above]{$!!b$} (q1)
%	(q1) edge [thick,bend left = 20] node  [above]{$?a$} (q2)
%	(q2) edge [thick,bend left = 20] node  [below]{$!!b$} (q1)
%	(q3) edge [thick,bend left = 20] node  [above]{$?b$} (q4)
%	(q4) edge [thick,bend left = 20] node  [below]{$!a$} (q3)
%	(q4) edge [thick, draw = green] node  [above]{$?c$} (q5)
%	
%	;
%\end{tikzpicture}
      }
	\end{center}
        
	\caption{Example of a Wait-Only protocol $\PP$.}\label{fig:example-1}
\end{figure}

\begin{example}
   \cref{fig:example-1} illustrates an example of a Wait-Only protocol: no state allows both broadcasting and receiving messages simultaneously. The waiting states are $q_1, {q}_{2}, {q}_{5},  {q}_{6}$, and $q_7$, with $R({q}_{5}) = \set{b, c}$.

\end{example}

For $n \in \nat_{>0}$, an $n$-process \emph{configuration} of a protocol $\PP = (Q, \Sigma, \qinit, T)$ is a vector $C \in Q^n$ where $C(\aproc)$ represents the state of the $\aproc$-th process in $C$  for $ 1 \leq \aproc \leq n$. The configuration is said to be \emph{initial} whenever $C(\aproc) = \qinit$ for all $1 \leq \aproc \leq n$. The dimension $||C||$  hence characterizes the number of processes in $C$. For a state $q \in Q$, 
%we denote $||C||_q$ as the number of processes in state $q$ in $C$, i.e. $|\set{i \mid C(i) = q}|$ and 
we use $C^{-1}(q)$ to represent the set of processes in state $q$, formally $C^{-1}(q)=\set{1\leq e\leq ||C|| \mid C(e) = q}$. For a subset $A\subseteq [1,||C||]$ of processes, we use $C(A)$ to identify the set of states of processes in $A$, formally  $C(A)=\set{C(e)\mid e\in A}$. 
%\ars{Check that these notations are used}\lug{j'ai fait un passage en enlevant les occurences de $\CC_n$ et $\CC[P]$ }
We let $\CC$ [resp. $\mathcal{I}$] be the set of all configurations [resp. initial configurations] of $\PP$.
%and $\mathcal{I}$ be the set of 
%all initial configurations of $\PP$. 
%When it is clear from the context, we use simply $\CC$ and $\mathcal{I}$. 

Let $\PP =(Q, \Sigma, \qinit, T)$ be a protocol. We now define the broadcast network semantics associated to $\PP$. For configurations $C, C' \in \CC$, 
transition $t = (q, !!a, q') \in T$ and  $\aproc \in [1,||C||]$, we write  $C \transup{\aproc, t} C'$ whenever $||C|| = ||C'||$, $C(\aproc) = q$, $C'(\aproc) = q'$ and, for all $\aproc' \in [1,||C||] \setminus\set{\aproc}$, either a reception occurs, i.e., $(C(\aproc'), ?a, C'(\aproc')) \in T$, or the process cannot receive $a$, in which case ($a \nin R(C(\aproc'))$ and $C(\aproc') = C'(\aproc')$). 
Intuitively, the $\aproc$-th process of $C$ broadcasts a message $a$, which is received by all processes able to receive it, while the other processes remain in their current state. 
We note $C \trans C'$ if there exists $\aproc \in [1,||C||]$ and $t \in T$ such that $C \transup{\aproc,t} C'$ and use $\trans^\ast$ [resp. $\trans^+$] for the reflexive and transitive [resp. transitive] closure of $\trans$.

A finite [resp. infinite] execution is a finite [resp. infinite] sequence of configurations  $\rho = C_0 \ldots C_\ell$ [resp.  $\rho = C_0 \ldots $] such that $C_0 \in \mathcal{I}$ and $C_{i-1} \trans C_{i}$ for all $0 < i \leq \ell$ [resp. for all $i >0$]. Its length $|\rho|$ is equal to $\ell+1$ [resp. $\omega$]. 
%It is initial whenever $C_0$ is initial. 
We write $\Lambda$ [resp. $\Lambda_\omega$] for the set of finite [resp. infinite] executions. 
%
%% A finite execution is a sequence of configurations $\rho = C_0 \cdots C_n$ such that $C_i \trans C_{i+1}$ for all $0 \leq i < n$. It is initial whenever $C_0$ is initial. Its length $|\rho|$ is equal to $n+1$.
%% An infinite execution is a sequence of configurations $\rho = C_0 C_1 \cdots$ such that $C_i \trans C_{i+1}$ for all $i \in \nat$. 
%% Its length is then $\omega$.
%% It is initial whenever $C_0$ is initial. We write $\Lambda$ (resp. $\Lambda_\omega$) the set of finite executions (resp. infinite).
%
For an execution $\rho = C_0 C_1 \ldots  \in \Lambda \cup \Lambda_\omega$ and $0 \leq i < |\rho|$, we use $\rho_i$ to denote $C_i$, the $i$-th configuration.
Additionally, $\nbagents{\rho}$ represents $||C_0||$, the number of processes in the execution. For $0 \leq i < j < |\rho|$, we define $\rho_{i,j} \overset{\mathrm{def}}{=} \rho_i \ldots \rho_j$. We also let $\rho_{\geq i} \overset{\mathrm{def}}{=} \rho_i \rho_{i+1} \ldots$ and $\rho_{\leq i} \overset{\mathrm{def}}{=} \rho_0 \ldots \rho_i$.
For $\aproc \in [1, \nbagents{\rho}]$, we denote by $\rho(\aproc)$ the sequence of states $\rho_0(\aproc) \rho_1(\aproc) \ldots$.

\begin{example}\label{example:exec}
  Here is an execution with $3$ processes for the protocol of  \cref{fig:example-1}:$$\begin{array}{l}
    (\qinit, \qinit, \qinit) \transup{1, (\qinit, !!d ,q_1)} (q_1, \qinit, \qinit) \transup{2, (\qinit, !!d, q_1)} (q_1, q_1, \qinit) \transup{3, (\qinit, !!a, \qinit)} ({q}_{5}, {q}_{2}, \qinit)\\  \transup{3, (\qinit, !!c, \qinit)} ({q}_{3}, {q}_{2}, \qinit) \transup{3, (\qinit, !!b, q_7)} ({q}_{3}, {q}_{3}, q_7).
  \end{array}
  $$
\end{example}

%	An \emph{almost-execution} is a sequence of configurations $\rho = C_0 \cdots C_n$ such that for all $0 \leq i < n$, $C_i = C_{i+1}$ or $C_i \trans C_{i+1}$. 

%	\subsection{Linear-time temporal Logic}

\subsection{Verification Problems}
We investigate two verification problems over broadcast networks: the synchronization problem and the repeated coverability problem. The former asks whether, given a protocol and a final state, there exists an execution 
that brings all processes to the final state. It is formally defined as follows.
\begin{center}
\begin{decproblem}
	\problemtitle{$\Target$~}
	\probleminput{A protocol  $\PP = (Q, \Sigma, \qinit, T)$ and
		a  state $q_f \in Q$;} 
	\problemquestion{Do there exist $C\in \mathcal{I}$, and $C' \in \CC$ such that $C \trans^\ast C'$, and }\problemquestionline{ $C'(\aproc)=q_f$  for all $\aproc \in [1, ||C'||]$?}
\end{decproblem}
\end{center}

The repeated coverability problem asks whether, given a protocol and a broadcast transition, there exists an infinite execution in which a single process takes
 the given transition infinitely often. It is formally defined as follows.
\begin{center}
\begin{decproblem}
	\problemtitle{$\RepeatedCover$~}
	\probleminput{A protocol  $\PP = (Q, \Sigma, \qinit, T)$ and
		a  transition $t = (q, !!a, q') \in T$;} 
	\problemquestion{Do there exist $\rho \in \Lambda_\omega$ and $e \in [1, \nbagents{\rho}]$, such that for all $i \in \nat$,  }\problemquestionline{there exists $j > i$ verifying $\rho_j \transup{\aproc, t} \rho_{j+1}$ ?}
\end{decproblem}%\nas{an initial execution $\rho$? ou definir $\Lambda_\omega$ comme l'ensemble des executions initiales?}\lug{je ne sais pas de quand date ce commentaire mais du coup toutes les executions sont initiales mtn}
\end{center}

%\arstext{Old Version}
%The repeated coverability problem asks, given a protocol and a transition of the protocol, if there exists an infinite execution in which one process takes infinitely often the given transition. It is defined formally as follows.
%\begin{center}
%\begin{decproblem}
%	\problemtitle{$\RepeatedCover$~}
%	\probleminput{A Protocol  $\PP = (Q, \Sigma, \qinit, T)$ and
%		a  transition $t = (q, \alpha, q') \in T$;} 
%	\problemquestion{Do there exist $\rho \in \Lambda_\omega$ and $e \in [1, \nbagents{\rho}]$, such that for all $i \in \nat$,  }\problemquestionline{there exists $j > i$ with $t = (\rho_j(\aproc), \alpha, \rho_{j+1}(\aproc))$?}
%\end{decproblem}
%\end{center}
%\arstext{EndOld Version}

%In their full generality, both problems are undecidable. We provide the proof of the undecidability of \Target~in \cref{app-sec0} while the undecidability of \RepeatedCover~is established in \cite[Theorem 5.1]{EsparzaFM99}.
\begin{theorem} \Target~and~\RepeatedCover~are undecidable. 
	%% The two following problems are undecidable in Broadcast Networks:
	%% \begin{itemize}
	%% 	\item	\RepeatedCover~\cite[Theorem 5.1]{EsparzaFM99}
	%% 	\item	\Target.
	%% \end{itemize}
\end{theorem}
We provide the proof of the undecidability of \Target~in \cref{app-sec0}, while the undecidability of \RepeatedCover~is established in \cite[Theorem 5.1]{EsparzaFM99}.
%We demonstrate that by restricting to Wait-Only protocols, one can regain decidability.
%
%%% The proof for the undecidability of \Target~can be found on \cref{app-sec0}, it consists of a reduction from the halting problem of two counters machine with tests to 0.
%
%\ars{Remove the table?}
%The following table summarizes our results.
%\begin{center}
%	\begin{tabular}{|l|c|c|c|}
%		\hline
%		& \textbf{Protocols} & \textbf{WO Protocols}& \textbf{WO Protocols}  \\
%		& & & \textbf{with $q_f\in Q_A$}\\
%		%& \textbf{SWO Protocols} \\
%		\hline
%		$\Target$  & {Undecidable}&  Ackermann-& in $2\textsc{ExpSpace}$ and \\
%		%& \textsc{PSpace}-hard and in \textsc{PSpace} \\
%		& & complete&  $\textsc{ExpSpace}$-hard\\
%		% & when $q_f$ is not in a cycle\\
%		\hline
%		$\RepeatedCover$ & Undecidable \cite{EsparzaFM99}& {in $2\textsc{ExpSpace}$ }  &  \\
%		& & and \pspace-hard &  \\
%		% & in P\\
%		\hline
%	\end{tabular}
%\end{center}
In the remainder of the paper, we show that by restricting to Wait-Only protocols, one can regain decidability. 
\begin{remark}\label{rem:noselfloop}
	Throughout the remainder of this paper, we will consider protocols without self-loop broadcast transitions of the form $(q, !!a, q)$. Such transitions can be
	transformed into two transitions $(q, !!a, p_q)$, $(p_q, !!\$, q)$ where $p_q$ is a new state added to the set of states and $\$$ is a new message added to the alphabet. This transformation of the protocol is equivalent to the original one with respect to $\Target$ and $\RepeatedCover$.
%	\color{red} here should come a remark about self loops (transitions of the form $(q, \alpha, q)$). \color{black}
\end{remark}

\subsection{Vector Addition System with States.}
In the following, we will extensively rely on the  model of Vector Addition Systems with States, which we introduce below. 
A Vector Addition System with States (VASS) is a tuple $\mathcal{V} = (\Loc, \ell_0, \Counters, \Delta)$ where $\Loc$ is a finite set of locations, $\ell_0 \in \Loc$ is the \emph{initial} location, $\Counters$ is a finite set of natural variables, called counters, and $\Delta \subseteq \Loc \times \mathbb{Z}^\Counters \times \Loc$ is a finite set of transitions. For a VASS $\mathcal{V}$, a configuration is a pair $(\ell, v)$ in $\Loc \times \mathbb{N}^\Counters$ where $v$ provides a value for each counter. The initial configuration is $(\ell_0, \mathbf{0})$ where $\mathbf{0}(\counter) = 0$ for all $\counter \in \Counters$. Given two configurations $(\ell, v)$, $(\ell',v')$  and a transition $(\ell, \delta, \ell') \in \Delta$, we write $(\ell, v) \abconftransup{\delta} (\ell',v')$ whenever $v'(\counter) = v(\counter) + \delta(\counter)$  for all counters $\counter \in \Counters$. We simply use $(\ell, v) \abconftrans (\ell', v')$ when there exists $(\ell, \delta, \ell') \in \Delta$ such that $(\ell, v) \abconftransup{\delta} (\ell', v')$. We denote $\abconftrans^\ast$ for the transitive and reflexive closure of $\abconftrans$.
An execution (or a run) %\ars{Why we use run here and execution for networks?} 
of the VASS is then a sequence of configurations starting with the initial configuration $(\ell_0, v_0), (\ell_1, v_1), \dots, (\ell_n, v_n)$ such that for all $0 \leq i <n$, it holds that $(\ell_i, v_i) \abconftransup{} (\ell_{i+1}, v_{i+1})$.  When the sequence is infinite, it is called an infinite execution of the VASS.
% is an infinite sequence  of configurations starting with the initial configuration $(\ell_0, v_0), (\ell_1, v_1), \dots, (\ell_n, v_n), \dots$ such that for all $i \in \nat$, it holds that $(\ell_i, v_i) \abconftransup{} (\ell_{i+1}, v_{i+1})$.  

%We consider the famous reachability problem:
%\begin{center}
%\begin{decproblem}
%	\problemtitle{$\Reachability$~}
%	\probleminput{A VASS $\mathcal{V} = (\Loc, \ell_0, \Counters, \Delta)$ and
%		a  location $\ell_f \in \Loc$;} 
%	\problemquestion{Is there a run from $(\ell_0, \mathbf{0})$ to $(\ell_f, \mathbf{0})$?}
%\end{decproblem}
%\end{center}
%\nas{j'ai enleve la boite pour reachability}
\Reachability, the well-known reachability problem for VASS, is defined as follows: given a VASS $\mathcal{V} = (\Loc, \ell_0, \Counters, \Delta)$ and
		a  location $\ell_f \in \Loc$, is there an execution from $(\ell_0, \mathbf{0})$ to $(\ell_f, \mathbf{0})$?
		
\begin{theorem}[Membership {\cite{LerouxS19}, Hardness \cite{CzerwinskiO21,Leroux21}}]\label{vass-reach-complexity}
	\Reachability~  is Ackermann-complete.
\end{theorem}

%Both problems are undecidable for Broadcast Networks: it has been proven for \RepeatedCover~in \cite[Theorem 5.1]{EsparzaFM99}, and we prove it for \Target~in \cref{app-sec0}. 
%\color{red} todo: sota \color{black}

%	\input{vass.tex}
%	\begin{decproblem}
%		\problemtitle{$\Liveness$~}
%		\probleminput{A Protocol  $\PP$ and
%			a LTL-formula $\phi$;} 
%		\problemquestion{Do there exist $\rho \in \Lambda_\omega$ and $\aproc \in [1, \nbagents{\rho}]$ such that $\rho[\aproc] \models \phi$?}
%	\end{decproblem}
%	

\section{Solving  \Target~for Wait-Only Protocols}\label{sec:synchro}

To solve \Target~we show in this section that we can build a VASS that simulates the behavior of a broadcast network
%\nas{on utilise Broadcast Networks ailleurs?} 
running a Wait-Only protocol. An intuitive way to proceed, inspired by the counting abstraction proposed for protocols communicating by pairwise rendez-vous communication \cite{german92}, consists in having one counter per state of the protocol, that stores the number of processes that populate that state. Initially, the counter associated with the initial state can be incremented as much as one wants: this will determine the number of processes of the execution simulated in the VASS. Then, the simulation of broadcast messages $(q, !!a, q')$ amounts to decrementing the counter associated to $q$ and increment the counter associated to $q'$. 
However, simulating receptions of the message (e.g. a transition $(p, ?a, p')$), requires to empty the counter associated to $p$ and transfer its value to the
counter associated to $q'$. This transfer operation  is not something that can be done in VASS unless the model is extended with special transfer transitions, 
leading to an undecidable reachability problem \cite{valk-self-icalp78}. To circumvent this problem, we rely on two properties of Wait-Only protocols:
%the fact that the studied protocols are wait-only and they hence the two following properties: 
(1) processes that occur to be in the same waiting state follow the same path of reception in the protocol, and end in the same action state (we show how to take care of potential non-determinism in the next section) and (2) processes in an action state cannot be influenced by the behaviour of other processes as they cannot receive any message. Thanks to property (1), instead of precisely tracking the number of processes in each waiting state, we only count the processes
 in a ``waiting region'' -- a connected component of waiting states populated by processes that will all reach the same action state simultaneously. The waiting region  allows us also to monitor when the processes go out from a waiting region to an action state. 
 %%%% 
 % PLUS DETAILLE MAIS RALLONGE BEAUCOUP LE PAPIER
 %%%%%%%% 
 %This requires to transfer the values of counters. We will do it by 
 %iterative decrements. The \Target\, objective will ensure that the counter will eventually be emptied, and thanks to property (2), we will show
 %that these decrements can be delayed, relieving us from the task to immediately empty the counter. 
 %%%%
 % PLUS CONCIS
 %%%%%
 We will explain later how we use property (2) to handle the transfer of values of counters within the VASS semantics, and thus simulating processes leaving a waiting region. 
 %\nas{j'ai enlevé ça, je propose de le dire plus tard: 
 	%, and at this point we empty the counter associated to the waiting region and fill accordingly the counter of the action state. However, as already seen, we cannot do a tranfert of the counters, but we use the fact that if we do not empty completely the counter of the  waiting region, it does not harm the simulation (the processes we do not transfert are considered to be in the acton state but unused and will be able to empty properly later the counter of the waiting region, if not the reachability objective will not be fulfilled).
% } 
 %\lug{je suis d'accord, j'ai ajouté une fin de phrase pour dire que ce transfert est maintenant nécessaire lorsque les processus quittent une zone d'attente (ai commenté le texte car on ne voyait pas tes autres commentaires du coup	)}
%% This abstraction is done thanks to the notion of \emph{summary} that we introduce hereafter.

For the rest of this section, we fix $\PP = (Q, \Sigma, \qinit, T)$ a Wait-Only protocol (with $Q_W$ the set of waiting states and  $Q_A$ the set of action states) for which we assume w.l.o.g. the existence of an \emph{uncoverable state} $\garbagestate \in Q$ verifying   $q\neq \garbagestate$ and $q' \neq \garbagestate$ for all $(q, \alpha, q') \in T$  and $\garbagestate \neq \qinit$.
Furthermore, we consider a final waiting state $q_f \in Q_W$.

\begin{remark}
  To ease the reading, we assume in this section that the final state $q_f$ is a waiting state, but our construction could be adapted to the case where $q_f$ is an action state. Moreover, we show in the next section that \Target\ in this latter case is in fact easier to solve. %\nas{j'ai changé cette phrase:This is explained in the next section, where we show that \Target\ in this latter case is easier to solve.}\lug{ok}
\end{remark}

\subsection{Preliminary properties}

%% %Let $\PP = (Q, \Sigma, \qinit, T)$ be a Wait-Only protocol. 
%% In the sequel, we assume that $\PP$ has a special \emph{uncoverable} state \garbagestate. Formally, $\garbagestate \in Q$, $\garbagestate \neq \qinit$ and for all $(q, \alpha, q') \in T$, $q\neq \garbagestate$ and $q' \neq \garbagestate$. 
We present here properties satisfied by Wait-Only protocols, which we will rely on for our construction. %of a VASS to solve \Target. 
We first show with \cref{lemma:resolving-non-determinism} that if, during an execution, two processes fulfill two conditions: ($i$) they are on the same waiting state $q_w$, and ($ii$) the next action state they reach (not necessarily at the same time) is the same (namely $q_a$), then one can build an execution where they both follow the same path between $q_w$ and $q_a$. This is trivial for \emph{deterministic} protocols  (where for all $q \in Q$ and $\alpha \in \Op_\Sigma$, there is at most one transition of the form $(q, \alpha, q') \in T$) and is also true for non-deterministic protocols.

%% from $q_w$, the two processes act exactly the same upon reception of the broadcast messages and reach $q_a$ at the same moment. However, if the protocol is not deterministic, then it might be that they take different reception transitions for the same message, diverging path and then meet again on $q_a$. The lemma explains how to avoid that and resolve non-determinism in this case.

For an execution $\rho \in \Lambda \cup \Lambda_\omega$ of $\PP$, an index $0 \leq j < |\rho|$, a waiting state $q \in Q_W$ and a process number $\aproc \in [1, \nbagents{\rho}]$ such that $\rho_j(\aproc) =q$, we define $\nextactionindex{\rho, j, \aproc } = \min\set{k \mid j \leq k \leq |\rho| \textrm{ and }\rho_k(\aproc) \in Q_A}$, i.e. the first moment after $\rho_j$ where the process $\aproc $ reaches an action state (if such moment does not exist, it is set to $|\rho|$).
%If it is not defined, then $\nextactionindex{i,j,\rho} = |\rho|+1$.
%	When $\nextactionindex{\rho, j, \aproc} < |\rho|$, 
We note $\nextactionstate{\rho, j, \aproc}\overset{\textrm{def}}{=}\rho_{\nextactionindex{\rho, j, \aproc}}(\aproc)$ the next action state for process $\aproc$ from $\rho_j$ if $\nextactionindex{\rho, j , \aproc} \neq |\rho|$ and otherwise we take the convention that $\nextactionstate{\rho, j, \aproc}$ is equal to  $\garbagestate$, the uncoverable state of the protocol \PP. 
%
%If $\nextactionindex{i,j,\rho} = |\rho|+1$, then $\nextactionstate{i,j,\rho}$ takes a special value $\#$.
Finally, we let $\nextaction{\rho, j, \aproc} = (\nextactionstate{\rho,j,\aproc}, \nextactionindex{\rho, j, \aproc})$. %is denoted by $\nextaction{\rho, j, \aproc}$.

%% We are now ready to state a key property. \cref{lemma:resolving-non-determinism}\ states that if during an execution two processes fill two conditions: (1) they are on the same waiting state $q_w$ at some point, and (2) the next action state they reach (not necessarily at the same time) is the same (namely $q_a$), then, one can build an execution in which they both follow the same path between $q_w$ and $q_a$. This is trivial when the protocol is \emph{deterministic} (for all $q \in Q$ and $\alpha \in \Op_\Sigma$, there is at most one transition $(q, \alpha, q') \in T$ for some $q' \in Q$): from $q_w$, the two processes act exactly the same upon reception of the broadcast messages and reach $q_a$ at the same moment. However, if the protocol is not deterministic, then it might be that they take different reception transitions for the same message, diverging path and then meet again on $q_a$. The lemma explains how to avoid that and resolve non-determinism in this case. Its proof can be found in appendix.

\begin{lemma}\label{lemma:resolving-non-determinism}
	Let $\rho \in \Lambda \cup \Lambda_\omega$ be an execution of $\PP$. If there exist $\aproc_1, \aproc_2 \in [1, \nbagents{\rho}]$ and $0 \leq j < |\rho|$ such that 
	\begin{enumerate}[(i)]
		\item $\rho_{j}(\aproc _1) = \rho_{j}(\aproc_2) \in Q_W$, 
		\item $\nextaction{\rho, j ,\aproc_1} = (q, j_1)$, and 
		\item $\nextaction{\rho, j, \aproc_2} = (q,j_2)$ with $j_1 \leq j_2 < |\rho|$,
	\end{enumerate}
%	$(i)$ $(iii)$  
	then there exists an execution $\rho'$ of the form $ \rho_{\leq j} \rho'_{j+1} \ldots \rho'_{j_2}\rho_{\geq j_2+1}$ such that $\rho'_k(\aproc_1) = \rho'_k(\aproc_2) = \rho_k(\aproc_1)$ for all $j+1 \leq k \leq j_1$, and $\rho'_k(\aproc_1) = \rho_k(\aproc_1) $  and $\rho'_k(\aproc_2) = q$ for all $j_1< k \leq j_2$. In particular $\nextaction{\rho', j ,\aproc_1}=\nextaction{\rho', j, \aproc_2} = (q,j_1)$.
	%		$\nextactionstate{i_1, j, \rho} = \nextactionstate{i_2, j, \rho}$, and $\nextactionindex{i_1, j, \rho} \leq \nextactionindex{i_2, j, \rho} $. Then, there exists $\rho' = $
\end{lemma}
%\nas{a-t-on besoin de tous ces details dans le lemme? Est ce que ca ne suffirait pas de dire "there exists an execution $\rho'$ of the form $ \rho_{\leq j} \rho'_{j+1} \ldots \rho'_{j_2}\rho_{\geq j_2+1}$ such that $\nextaction{\rho', j ,\aproc_1}=\nextaction{\rho', j, \aproc_2} = (q,j_1)$"? Il me semble que c'est la seule chose dont on se sert}

\begin{example}
	Consider the protocol $\PP$ in \cref{fig:example-1} and the execution $\rho$ of \cref{example:exec}. %for the protocol $\PP$ in \cref{fig:example-1}. 
	We have $\rho_2(1) = \rho_2(2) = q_1$ and $\nextaction{\rho,2, 1} = (q_3, 4)$ and $\nextaction{\rho,2, 2} = (q_3,5)$. Applying \cref{lemma:resolving-non-determinism}, we can build: $\rho' = (\qinit, \qinit, \qinit) \trans (q_1, \qinit, \qinit) \trans (q_1, q_1, \qinit) \trans (q_5, q_5, \qinit) \trans (q_3, q_3, \qinit) \trans (q_3, q_3, q_7)$ where process $1$ and process $2$ follow the same path between $q_1$ and $q_3$. However, if you consider the following execution: $\rho'' = (\qinit, \qinit, \qinit) \trans (q_1, \qinit, \qinit) \trans (q_1, q_1, \qinit) \trans (q_2, q_5, \qinit) \trans (q_3, q_6, q_7) \trans (q_4, q_4, q_4)$, we cannot apply the lemma as from $q_1$, processes $1$ and $2$ reach two different next action states ($q_3$ and $q_4$). 
\end{example}

%We now prove the lemma.
%% Following the above lemma, we provide a notion of define what it means for an execution to be well-formed. Intuitively, an execution is \emph{well-formed} if at any moments, two processes on the same waiting states and with the same next action state, take the same path of receptions.

%% \begin{definition}
%% 	Let $ \rho$ be an execution of a protocol $\PP$, we say that $\rho$ is well-formed if
%% 	for all $0 \leq i < |\rho|$, for all $\aproc_1, \aproc_2\in [1, \nbagents{\rho}]$ such that $\rho_i(\aproc_1)  = \rho_i(\aproc_2)\in Q_W$ and $\nextactionstate{\rho, i, \aproc_1} = \nextactionstate{\rho, i, \aproc_2} = q$, it holds that for all $i \leq k\leq \nextactionindex{\rho, i, \aproc_1}$,  $\rho_k(\aproc_1) = \rho_k(\aproc_2)$.
%% \end{definition}
%
The above lemma enables us to focus on a specific subset of executions for \Target,~which we refer to as well-formed. In these executions, at any moment, two processes in the same waiting state and with the same next action state, follow the same path of receptions. Formally, an execution $\rho$ of a protocol $\PP$ is \emph{well-formed} iff for all $0 \leq i < |\rho|$, for all $\aproc_1, \aproc_2\in [1, \nbagents{\rho}]$ such that $\rho_i(\aproc_1)  = \rho_i(\aproc_2)\in Q_W$ and $\nextactionstate{\rho, i, \aproc_1} = \nextactionstate{\rho, i, \aproc_2} = q$, it holds that $\rho_k(\aproc_1) = \rho_k(\aproc_2)$ for all $i \leq k\leq \nextactionindex{\rho, i, \aproc_1}$.
%
%We say that an execution $\rho$ is \emph{well-formed} when for all $0 \leq j < |\rho|$, for all agents $\aproc_1, \aproc_2$ such that $\rho_j(\aproc_1) = \rho_j(\aproc_2) \in Q_W$ and 
%
%$\nextactionstate{\rho, j, \aproc_1} =\nextactionstate{\rho, j, \aproc_2}$, it holds that $\nextaction{\rho, j, \aproc_1} =\nextaction{\rho, j, \aproc_2}$.\lug{changer la def de well formed}
From \cref{lemma:resolving-non-determinism}, we immediately get:
\begin{corollary}\label{cor:well-formed}
There exists an execution $\rho = C_0 \ldots C_n$ such that $C_n(\aproc) = q_f$ for all $\aproc \in [1, \nbagents{\rho}]$ iff there exists a well-formed execution from $C_0$ to $ C_n$.
\end{corollary}

We finally provide a result which will allow us to bound the size of the VASS that we will build from the given Wait-Only protocol. Given an execution $\rho \in \Lambda \cup \Lambda_\omega$ of $\PP$, an index $0 \leq j < |\rho|$ and an action state $q_a \in Q_A$,  we define $\nextactionindexset{\rho,j,q_a}$ as the set of indices where processes in a waiting state at $\rho_j$
will reach $q_a$, if it is their next action state: $\nextactionindexset{\rho,j,q_a} \overset{\textrm{def}}{=} \set{i  \mid \textrm{ there exists } \aproc \in [1, \nbagents{\rho}]\text{ s.t. } \rho_j(\aproc) \in Q_W \text{ and }\nextaction{\rho,j,\aproc} = (q_a,i)}$.% be the set of indices (or moments) at which processes in a waiting state in $\rho_j$ will reach $q_a$ as a next action state. Using the definition of well-formed executions, we deduce:

%% Let $q_a \in Q_A$, we denote $\nextactionindexset{\rho,j,q_a} = \set{i \mid \aproc \in [1, \nbagents{\rho}]\text{ s.t. } \rho_j(\aproc) \in Q_W \text{ and }\nextaction{\rho,j,\aproc} = (q_a,i)}$, in other words, $\nextactionindexset{\rho,j,q_a}$ is the set of indices (or moments) in which processes with next action state $q_a$ reach it. 

%%The following lemma holds.
%
%Let $q_a \in Q_A$, and denote $\nextactionset{\rho,j,q_a} = \set{(i, q_a) \mid \aproc \in [1, \nbagents{\rho}]\text{ s.t. } \rho_j[\aproc] \in Q_W \text{ and }\nextaction{\rho,i,\aproc} = (i, q_a)}$. 
%We get the following lemma directly from the definition of well-formed.
\begin{lemma}\label{lemma:well-formed-execution:bound-next-action}For all well-formed executions $\rho$, for all $0 \leq j < |\rho|$, and for all $q_a \in Q_A$, we have $|\nextactionindexset{\rho,j,q_a} | \leq |Q_W|$.
\end{lemma}
\begin{proof}
  If we have $|\nextactionindexset{\rho,j,q_a} | > |Q_W|$, by pigeon hole principle there exist $\aproc_1,\aproc_2 \in [1, \nbagents{\rho}]$ such that $\rho_j(\aproc_1) = \rho_j(\aproc_2) \in Q_W$ and such that  $\nextaction{\rho,j,\aproc_1} = (q_a,i_1)$ and $\nextaction{\rho,j,\aproc_2} = (q_a,i_2)$ with  $i_1 < i_2$. Consequently $\rho_{i_1}(\aproc_1) = q_a \in Q_A$ and $\rho_{i_1}(\aproc_2) \in Q_W$ hence $\rho_{i_1}(\aproc_1) \neq \rho_{i_1}(\aproc_2)$, which contradicts the definition of well-formedness.
\end{proof}

\subsection{Building the VASS that Simulates a Broadcast Network}

\subsubsection{Summaries}

We present here the formal tool we use to represent processes in waiting states. We begin by introducing some notations.
A \emph{print} $\print$  is a set of waiting states. 
Given a non empty subset of states $A = \set{q_1, \dots, q_n} \subseteq Q$, we define a configuration $\conf{A} \in Q^{n}$  such that $\conf{A}(\aproc) = q_\aproc$ for all $\aproc \in [1,n]$. When $A=\set{q}$, we write $\confstate{q} \in Q^1$ instead of $\conf{\set{q}}$.
	Conversely, given a configuration $C$, we define $\setof{C}= \set{q\in Q \mid C^{-1}(q) \neq \emptyset}$.
For two configurations $C\in Q^n$ and  $C' \in Q^m$, we let $C \oplus C'$ be the configuration $C'' \in Q^{n+m}$ defined by: $C''(\aproc) = C(\aproc)$ if $1 \leq \aproc \leq n$ and $C''(\aproc) = C'(\aproc - n)$ if $n+1 \leq \aproc \leq m+n$. 

A \emph{summary} $\summary$ is a tuple $( \print, q_a, k)$ such that $\print\subseteq Q_W$ is a print, $q_a\in Q_A$ is an \emph{exit state}, and $k \in [1, |Q_W|+1]$ is an identifier. The label of $\summary$ is $\lab{\summary}=(q_a,k)$ and we denote its print by $\printof{\summary}$.
Finally, we define a \emph{special} summary \arrived.

In our construction, each process in a waiting state is associated with a summary while processes in action states are handled by the counters of the VASS. Intuitively, a summary $(\print, q_a, k)$ represents a set of processes lying in the set $\print \subseteq Q_W$ that will reach the same next action state  $q_a$ 
simultaneously (this restriction is justified by \cref{lemma:resolving-non-determinism}). Since the set of waiting states
$\print$ evolve during the simulation, each summary must be uniquely identified. To achieve this, we use an integer $k \in [1, |Q_W|+1]$. Hence each summary is identified with the pair $(q_a,k)$. We do not need more than $|Q_W|+1$ different identifiers, because when a process enters a new summary (i.e it arrives in a waiting state $q_w$ from an action state),  aiming for next action state $q_a$, it either joins an existing summary with exit state $q_a$, or create a new one. In the
latter case, well-formed executions ensure that no existing summary $S=(\print, q_a, k)$ with exit state $q_a$ is such that 
$q_w\in\print$. 
%The latter case should only being allowed when there is no summary already aiming for $q_a$ and with state $q_w$ in its print. 
Otherwise two processes would be in the same state $q_w$, aiming for the same action state, but at different moments, contradicting 
%as we focus on well-formed executions, this contradicts 
\cref{lemma:resolving-non-determinism}. Note that we need $|Q_W|+1$ different identifers, and not $|Q_W|$ for technical reasons.
Finally, the special summary \arrived\ is used to indicate when the processes described by a summary reach the exit action state.

We now describe how summaries evolve with the sequence of transitions. Let $S = (\print, q_a, k)$ and $S' = (\print', q_a, k)$ be two summaries (with the same exit state and identifier) and $t = (q, !!b, q') \in T$ such that there exists a configuration $C'$ verifying $\confstate{{q}} \oplus \conf{\print}   \transup{1, t} \confstate{{q'}} \oplus C'$. We then write $S \abtransup{t} S'$ whenever $\setof{C'} = \print' \subseteq Q_W$. This represents the evolution of a summary upon reception of message $b$, when the processes all stay in waiting states, and no new process joins the ``waiting
region'' of the given summary. We write $S \abtransup{t, +q'} S'$ whenever $q'\in Q_W $ and $\setof{C'} \cup \{q'\} = \print' \subseteq Q_W$. In that latter case, the process responsible for the transition $t$ joins the ``waiting region'' represented by $S$. 
This typically occurs when the process’s next action state is $q_a$, and it reaches $q_a$ simultaneously with the processes described by $S$.
Finally, we use $S \abtransup{t} \arrived$ whenever $\setof{C'} = \set{q_a}$. This
	represents the evolution (and disappearance) of the summary when all the processes reach $q_a$ (they all reach it simultaneously).

%%  We define the transition $S \abtransup{t} S'$ whenever
%% {there exists a configuration $C'$ such that $\confstate{{q}} \oplus \conf{\print}   \transup{1, t} \confstate{{q'}} \oplus C'$ and $\setof{C'} = \print' \subseteq Q_W$.
%% This represents the evolution of a summary upon reception of message $b$, when the processes all stay in waiting states, and no new process join the ``waiting
%% region'' of the given summary. And if $q' \in Q_W$, we say that $S \abtransup{t, +q'} S'$ whenever there exists $S'' = (\print'', p, k)$ such that $S \abtransup{t} S''$ and $\print' = \print'' \cup \set{q'}$. This represents the evolution of a summary when the process that has taken the transition $t$ joins the ``waiting region'' represented by $S$.

%% 	We say that $S \abtransup{t} \arrived$ whenever 
%% 	there exists a configuration $C'$ such that $\confstate{{q}} \oplus \conf{\print}   \transup{1, t} \confstate{{q'}} \oplus C'$ and $\setof{C'} = \set{q_a}$. This
%% 	represents the evolution (and disappearance) of the summary when all the processes reach $q_a$ (they all reach it together).

\begin{example}
	Returning to the protocol $\PP$ of \cref{fig:example-1}, consider the summary $S_0 = (\set{q_1}, {q}_{3}, 1)$. Observe that $S_0 \abtransup{(\qinit, !!a, \qinit)} (\set{{q}_{2}}, {q}_{3}, 1)$ and  $S_0 \abtransup{(\qinit, !!a, \qinit)} (\set{{q}_{5}}, {q}_{3}, 1)$. However, by definition, we do not have $S_0 \abtransup{(\qinit, !!a, \qinit)} (\set{{q}_{5}, {q}_{2}}, {q}_{3}, 1)$. Indeed, processes summarized in $S_0$ are forced to all go either on ${q}_{2}$ or on ${q}_{5}$ upon receiving $a$: thanks to \cref{cor:well-formed}, we consider only well-formed executions, where all processes summarized in $S_0$  choose the same next state. %As they will all reach the same action state ${q}_{3}$, we can resolve non-determinism of the protocol thanks to \cref{lemma:resolving-non-determinism}.
Additionnaly, we have $(\set{{q}_{2}}, {q}_{3}, 1) \abtransup{(\qinit, !!b, q_7)} \arrived$
and $(\set{q_1}, {q}_{4}, 1) \abtransup{(\qinit, !!a, \qinit)} (\set{{q}_{5}}, {q}_{4}, 1) \abtransup{(\qinit, !!b, q_7), +q_7} (\set{{q}_{6}, q_7}, {q}_{4}, 1) \abtransup{(\qinit, !!d, q_1)} \arrived$. However, note that we do not have  $(\set{{q}_{2}}, {q}_{4}, 1) \abtransup{(\qinit, !!b, q_7)} \arrived$, since the action state ${q}_{3}$ reached from ${q}_{2}$ upon receiving $b$ is not the exit state ${q}_{4}$.
  %% \nas{add the fact that $S_1\abtransUp{(\qinit, !!a, \qinit)}(\set{{q}_{2}}, {q}_{4}, 1)$ but then it has no successor with $\qinit, !!b, q_7)$?}\lug{je pense que $S_1\abtransUp{(\qinit, !!a, \qinit)}(\set{{q}_{2}}, {q}_{4}, 1)$ et je ne comprends pas pourquoi il n'y aurait pas de successeur ?}
  %%       \nastext{Parce que $(\set{{q}_{2}}, {q}_{4},1)$ n'a aucun successeur pour la transition $(\qinit, !!b, q_7)$: ${q}_{2}$ va obligatoirement en ${q}_{3}$ en recevant $b$, qui
  %%       est un état d'action différent de ${q}_{4}$. Donc ce summary n'a aucun successeur pour la transition $t$. Cela illustrerait ce qui se passe quand on devine n'importe quoi}
\end{example}

\old{
	\new{
		EN DEFINISSANT CA ON VEUT DEFINIR EN MEME TEMPS LA TRANSITION DU VASS 
		Let $\SS$, $\SS'$ be two coherent sets of summaries and $t = (q, !!a, q')$. We write $\SS \abtransUp{t} \SS'$ whenever the following holds:
		\begin{itemize}
			\item there exists $\SS''$ a coherent set of summaries such that (1) for all $S \in \SS$, either $S \trans \arrived$ or there exists $S' \in \SS''$ such that $S \trans S'$; and (2) for all $S' \in \SS''$, there exists $S \in \SS$ such that $S \trans S'$.
			\item if $q' \in Q_A$, $\SS' = \SS''$, otherwise, either (a) there exists $S = (\print, s, k) \in \SS''$ such 
		\end{itemize}
	}
	\color{red}
	Let $\SS$, $\SS'$ be two coherent sets of summaries and $t = (q, !!a, q')$. We write $\SS \abtransUp{t} \SS'$ whenever the following holds:
	\begin{itemize}
		\item there $\SS''$ a coherent set of summaries such that (1) for all $S \in \SS$, either $S \trans \arrived$ or there exists $S' \in \SS''$ such that $S \trans S'$; and (2) for all $S' \in \SS''$, there exists $S \in \SS$ such that $S \trans S'$.
	\end{itemize}
	\color{black}
}

\subsubsection{Definition of the VASS}
\label{sec:def-vass}
We now explain how we use the summaries in the VASS simulating the executions in the network. First we say that  a set $\SS$ of summaries is \emph{coherent} if for all distinct pairs $(\print_1, q_1, k_1), (\print_2, {q}_{5}, k_2)\in \SS$ such that $q_1={q}_{5}$, we have $k_1 \neq k_2$ and $\print_1 \cap \print_2 = \emptyset$. We denote by $\CoherentSets$ the set of coherent sets of summaries. For a set of summaries $\SS$ we let $\lab{\SS} = \set{\lab{S} \mid S \in \SS}$. Observe that, when $\SS$ is coherent, for a label $(q, k) \in \lab{\SS}$, there is a unique $S \in \SS$ such that $\lab{S} = (q,k)$.

We define the VASS $\mathcal{V}_\PP$ simulating the protocol $\PP$ as follows: $\mathcal{V}_\PP = (\Loc, s_0, \Counters, \Delta)$, where $\Loc = \CoherentSets \cup \set{s_0} \cup \set{s_f}$, the set of counters is $\Counters = \set{\counter_q \mid q \in Q_A} \cup \set{\counter_{{(q,k)}} \mid q\in Q_A, k \in [1, |Q_W|+1]}$, and the set of 
transitions  $\Delta = \Delta_0 \cup (\bigcup_{t \in T} \Delta_t) \cup \Delta_e$, is defined as follows:
\begin{itemize}
\item $\Delta_0$ contains exactly  the following transitions:
\begin{itemize}
	\item $(s_0, \delta_0, s_0)$ where $\delta_0(\counter_{\qinit}) = 1$ and $\delta_0(\counter) = 0$ for all other counters;
	\item $(s_0, \textbf{0}, \emptyset)$ where $\textbf{0}$ is the null vector (observe that the empty set is a coherent set of summaries);
	\item $(\set{S_f}, \textbf{0}, s_f) \in \Delta$ where $S_f= (\set{q_f}, q_u, 1)$; %\nas{ca n'est correct que si $q_f\in Q_W$!}
	\item $(s_f, \delta_f, s_f) \in \Delta$ where $\delta_{f}(\counter_{(q_u, 1)}) = -1$ and $\delta_f(\counter) = 0$ for all other counters.
\end{itemize}

%Furthermore, let $\SS, \SS'$ be two coherent set of summaries and 
\item For $t = (q,!!a,q') \in T$, we have  $(\SS, \delta, \SS') \in \Delta_t$ iff one of the following conditions holds:
\begin{enumerate}[a.]
	\item \label{cond-abconftrans-broadcast-action-state}$q' \in Q_A$ and $\SS = \set{S_1, \dots, S_k}$. Then, for all $1 \leq i \leq k$, there exists $S'_i$ such that $S_i \abtransup{t} S'_i$ and $\SS' = \set{S'_i \mid 1 \leq i \leq k}\setminus\set{\arrived}$. Moreover, $\delta(\counter_q) = - 1$, $\delta(\counter_{q'}) = 1$ and $\delta(\counter) = 0$ for all $\counter \in \Counters \setminus \set{\counter_q,\counter_{q'}}$. Here we simply update the sets of states populated by processes in waiting states after the transition $t$. Some summaries
	may have disappeared from $\SS$ if the processes represented by this summary have reached their exit action state when receiving $a$.  For now, we only modify the counters associated to the states $q$ and $q'$. Note that this is well-defined, since we assume in \cref{rem:noselfloop} that there is no self-loop of the form $(q, !!a, q)$.
	
	\item \label{cond-abconftrans-broadcast-join-existing}$q' \in Q_W$, $\SS = \set{S_1, \dots, S_k}$, and there exists $1 \leq i \leq k$ and $S'_i$ such that $S_i \abtransup{t, +q'} S'_i$. For all $j \neq i$, there exists $S'_j$ such that $S_j \abtransup{t} S'_j$. Then $\SS'= \set{S'_j \mid 1 \leq j \leq k}\setminus\set{\arrived}$. 
	Moreover, $\delta(\counter_q) = - 1$,  $\delta(\counter_{\lab{S_i}}) =1$ and $\delta(\counter) = 0$ for all $\counter \in \Counters \setminus \set{\counter_q,\counter_{{\lab{S_i}}}}$. In that case, the process having 
	performed the transition joins an existing summary $S_i$. We have then modified accordingly the counters associated to $q$ and to the appropriate summary.
	
	\item \label{cond-abconftrans-broadcast-new-summary}$q' \in Q_W$ and $\SS = \set{S_1, \dots, S_k}$.  For all $1 \leq j \leq k$, there exists $S'_j$ such that $S_j \abtransup{t} S'_j$. Then $\SS'=(\set{S'_1,\dots, S'_k}\setminus\{\arrived\})\cup \set{(\set{q'}, q_a, k)}$ for some $q_a \in Q_A$ and $k \in [1, |Q_W|+1]$ such that $(q_a, k)\nin \lab{\SS}$. Moreover, $\delta(\counter_q) = - 1$, $\delta(\counter_{(q_a, k)}) = 1$, and $\delta(\counter )= 0$ for all $\counter \in \Counters \setminus \set{\counter_q,\counter_{(q_a,k)}}$. This case happens
	when the process having sent the message $a$ reaches a waiting state, and its expected next action pair (index and state)
%	instant of joining its next active state and its expected next active
%	state 
	does not correspond to any existing summary. In that case, it joins a new summary, the counter associated to the state $q$ is decremented and
	the counter of this new summary is incremented. 
\end{enumerate}

\item Finally, $\Delta_e$ allows to empty the counters associated to summaries that have disappeared from the set of locations. It is defined as the set of transitions $(\SS, \delta, \SS')$ such that:
%\begin{enumerate}[d.]
%	\item 
	%\label{cond-abconftrans-empty} 
	$\SS = \SS'$ and there exists $(q,k) \nin \lab{\SS}$ with $\delta(\counter_{q}) =  + 1$ and $\delta(\counter_{(q, k)}) = - 1$ and for all other counters, $\delta(\counter) = 0$.
	%\ars{I'd like to remove this d. label, should check where it is mentioned}
%\end{enumerate}
\end{itemize}
%% An \emph{S-configuration} is a tuple $\lambda = (\SS, v)$ where $\SS$ is a coherent set of summaries and $v: \Counters \rightarrow \nat$ is a valuation over the counters.
%\nas{J'ai ajouté ici cette remarque (cf precedente note dans lintro de section 3). Qu'en pensez vous?}\lug{top !}
Transitions from $\Delta_e$ allow the transfer of counter values from summaries that have disappeared to the counter of their exit action state. This transfer is not
immediate, as it is done through iterative decrements and increments of counters. In particular, it is possible for the execution of the VASS to continue without
the counter of a disappeared summary being fully transfered. The processes represented by the counter of a disappeared summary behave in the same way
as processes in the exit action state that do not take any action. 
These counters can be emptied later, but always from a location from where no summary with the same label exists. This ensures that there is no ambiguity between processes in a waiting region and those on an action state that have not yet been transferred to the appropriate counter. 
%% Let $\lambda = (\SS, v), \lambda' = (\SS', v')$ be two S-configurations and $t = (q, !!a, q') \in T$. When $\lambda \abconftransup{\delta} \lambda'$ and $(\SS, \delta, \SS') \in \Delta_t$, we say that $\lambda \abconftransup{t} \lambda'$.
\begin{figure}
\begin{minipage}[c]{0.50\linewidth}
%\begin{figure}
		\tikzset{box/.style={draw, minimum width=4em, text width=4.5em, text centered, minimum height=17em}}

\begin{tikzpicture}[->, >=stealth', shorten >=1pt,node distance=2cm,on grid,auto, initial text = {},square/.style={regular polygon,regular polygon sides=4}] 
	\node[state, square, initial] (s0) {$s_0$};
	\node[state, square] (s1) [right = of s0 , xshift = -20] {$\emptyset$};
	\node[state, rectangle] (s) [right = of s1 ] {$\set{q_f}, q_u, 1$};
	\node[state, square] (sf) [right = of s ] {$s_f$};
	
	\node[draw, fill = cyan, fill opacity = 0.2, text opacity = 1, fit=(s1) (s), text height = 0] (VASS) [yshift = 0]{\CoherentSets};
	
	\path[->] 
	(s0) edge [loop above] node [] {\inc{\counter_{\qinit}}} ()
	edge [] node {} (s1)
	(s) edge [] node {} (sf)
	(sf) edge [loop above] node [] {\dec{\counter_{(q_u, 1)}}} ()
	;
\end{tikzpicture}

%
%\begin{tikzpicture}[->, >=stealth', shorten >=1pt,node distance=2cm,on grid,auto, initial text = {}] 
%	\node[state, initial] (q0) {$\qinit$};
%	\node[state] (q1) [right = of q0, yshift = 25] {$q_1$};
%	\node[state] (q2) [right = of q1] {${q}_{5}$};
%	\node[state] (q3) [right  = of q0, yshift = -25] {${q}_{2}$};
%	\node[state] (q4) [right  = of q3] {$q_{4}$};
%	\node[state, fill=green] (q5) [right  = of q4] {$q_{5}$};
%	
%	\path[->] 
%	(q0) edge [thick,bend right = 20] node  [below]{$\tau$} (q3)
%	edge [thick,bend left = 20] node  [above]{$!!c$} (q3)
%	edge [thick,bend left = 0] node  [above]{$!!b$} (q1)
%	(q1) edge [thick,bend left = 20] node  [above]{$?a$} (q2)
%	(q2) edge [thick,bend left = 20] node  [below]{$!!b$} (q1)
%	(q3) edge [thick,bend left = 20] node  [above]{$?b$} (q4)
%	(q4) edge [thick,bend left = 20] node  [below]{$!a$} (q3)
%	(q4) edge [thick, draw = green] node  [above]{$?c$} (q5)
%	
%	;
%\end{tikzpicture}
	\caption{The structure of the VASS $\mathcal{V}$.}\label{fig:big-picture-vass}
%\end{figure}

%\begin{figure}
		\tikzset{box/.style={draw, minimum width=4em, text width=4.5em, text centered, minimum height=17em}}
\tikzset{
	loop/.code={% overwrite default, what will it break?
		\let\tikztotarget\tikztostart,
		\pgfkeysalso{looseness=2,min distance=+5mm,every loop}}}
\begin{tikzpicture}[->, >=stealth', shorten >=1pt,node distance=2.2cm,on grid,auto, initial text = {},square/.style={regular polygon,regular polygon sides=4}] 
	\node[state, rectangle] (s1) {\shortstack{$\set{{q}_{2}}, {q}_{3}, 1$\\ $\set{q_1}, {q}_{3}, 2$\\ $\set{q_1, {q}_{5}}, {q}_{4}, 1$}};
	\node[state, rectangle] (s2) [right = of s1, xshift = 50, inner ysep = 4]{\shortstack{$\set{q_1}, {q}_{3}, 2$\\ $\set{q_1, {q}_{6}, q_7}, {q}_{4}, 1$}};
%	\node[state, rectangle] (s2) [right = of s1, xshift = 30, inner ysep = 6] {\shortstack{$\set{{q}_{2}}, {q}_{3}, 1$\\ $\set{q_1}, {q}_{4}, 1$}};
%	\node[state, rectangle] (s3) [below right  = of s1, xshift = 50, yshift = -20 , inner ysep = 6] {\shortstack{$\set{{q}_{2}}, {q}_{3}, 1$\\ $\set{q_1}, {q}_{3}, 2$}};
%	\node[state, rectangle] (s4) [below = of s1, yshift = 0 ] {$\set{q_1,{q}_{2}}, {q}_{3}, 1$};

	%	\draw [rounded rectangle, fill= green, fill opacity=0.2] (-1,-1) rectangle (0.7,0.5);
	
	\node[ rounded corners, fill = green, fill opacity = 0.2, text opacity = 1, fit=(s1),inner xsep= -3.5 , inner ysep= -16] () [yshift = 13]{};
	\node[ rounded corners, fill = orange, fill opacity = 0.3, text opacity = 1, fit=(s1), inner xsep= -3.5 ,inner ysep= -16] () [yshift = 0]{};
	\node[ rounded corners, fill = magenta, fill opacity = 0.3, text opacity = 1, fit=(s1), inner xsep= -3.5 ,inner ysep= -16] () [yshift = -13]{};
	
	\node[ rounded corners, fill = orange, fill opacity = 0.3, text opacity = 1, fit=(s2), inner xsep= -3.5 ,inner ysep= -11] () [yshift = 7]{};
	\node[ rounded corners, fill = magenta, fill opacity = 0.3, text opacity = 1, fit=(s2), inner xsep= -3.5 ,inner ysep= -11] () [yshift = -7]{};

	\path[->] 
	(s1) edge [] node [] {\small\shortstack{\dec{\counter_{\qinit}}\\ \inc{\counter_{({q}_{4},1)}}}} (s2)
	(s2) edge [loop right] node [] {\small\shortstack{\dec{\counter_{({q}_{3},1)}}\\ \inc{\counter_{{q}_{3}}}}} ()
	;
\end{tikzpicture}
	\caption{One successor of location $\set{(\set{{q}_{2}}, {q}_{3}, 1), (\set{q_1}, {q}_{3}, 2), (\set{q_1, {q}_{5}}, {q}_{4}, 1)}$ (left location) after the broadcast transition $(\qinit, !!b, q_7)$}\label{fig:vass-deleting-summary}
%\end{figure}
\end{minipage} \hfill
\begin{minipage}[c]{0.40\linewidth}
%\begin{figure}
		\tikzset{box/.style={draw, minimum width=4em, text width=4.5em, text centered, minimum height=17em}}

\begin{tikzpicture}[->, >=stealth', shorten >=1pt,node distance=2.5cm,on grid,auto, initial text = {},square/.style={regular polygon,regular polygon sides=4}] 
	\node[state, rectangle] (s1) {$\set{{q}_{2}}, {q}_{3}, 1$};
	\node[state, rectangle] (s2) [right = of s1, xshift = 30, inner ysep = 6] {\shortstack{$\set{{q}_{2}}, {q}_{3}, 1$\\ $\set{q_1}, {q}_{4}, 1$}};
	\node[state, rectangle] (s3) [below right  = of s1, xshift = 50, yshift = -20 , inner ysep = 6] {\shortstack{$\set{{q}_{2}}, {q}_{3}, 1$\\ $\set{q_1}, {q}_{3}, 2$}};
	\node[state, rectangle] (s4) [below = of s1, yshift = 0 ] {$\set{q_1,{q}_{2}}, {q}_{3}, 1$};

%	\draw [rounded rectangle, fill= green, fill opacity=0.2] (-1,-1) rectangle (0.7,0.5);
	
	\node[ rounded corners, fill = green, fill opacity = 0.2, text opacity = 1, fit=(s1), inner sep= -4] () {};
	\node[ rounded corners, fill = green, fill opacity = 0.2, text opacity = 1, fit=(s2), inner xsep= -3.5 ,inner ysep= -11] () [yshift = 7]{};
	\node[ rounded corners, fill = magenta, fill opacity = 0.2, text opacity = 1, fit=(s2), inner xsep= -3.5 ,inner ysep= -11] () [yshift = -7]{};
	\node[ rounded corners, fill = green, fill opacity = 0.2, text opacity = 1, fit=(s3), inner xsep= -3.5 ,inner ysep= -11] () [yshift = 7]{};
	\node[ rounded corners, fill = orange, fill opacity = 0.3, text opacity = 1, fit=(s3), inner xsep= -3.5 ,inner ysep= -11] () [yshift = -7]{};
	\node[ rounded corners, fill = green, fill opacity = 0.2, text opacity = 1, fit=(s4), inner sep= -4] () {};
	
	\path[->] 
	(s1) edge [] node [] {\shortstack{\dec{\counter_{\qinit}}\\ \inc{\counter_{({q}_{4},1)}}}} (s2)
	(s1) edge [] node [yshift = -10, xshift = 5]{\shortstack{\dec{\counter_{\qinit}}\\ \inc{\counter_{({q}_{3},2)}}}} (s3)
	(s1) edge [] node [left]{\shortstack{\dec{\counter_{\qinit}}\\ \inc{\counter_{({q}_{3},1)}}}} (s4)
	;
\end{tikzpicture}
	\caption{Three successors of location $\set{(\set{{q}_{2}}, {q}_{3}, 1)}$ (top left location) after the broadcast transition $(\qinit, !!d, q_1)$.}\label{fig:vass-new-summary}
%\end{figure}
\end{minipage}
\end{figure}

\begin{example}
	\cref{fig:big-picture-vass,fig:vass-new-summary,fig:vass-deleting-summary} illustrate the VASS $\mathcal{V}_\PP$ defined for protocol $\PP$ in \cref{fig:example-1}. \cref{fig:big-picture-vass}\ shows the overall structure of the VASS: any run reaching $(s_f, \mathbf{0})$ starts by incrementing the counter associated with the initial state and ends by decrementing the counter associated with summary $(\set{q_f}, q_u, 1)$. We recall that $q_u$  is an artificially added unreachable state. This ensures that processes counted by this counter end in $q_f$ (as they cannot be in $q_u$).
	
	In \cref{fig:vass-new-summary}, we detail how message receptions and summary creations are handled. The top-left location has a single summary, $(\set{{q}_{2}}, {q}_{3}, 1)$, representing configurations where processes in waiting states are all in ${q}_{2}$, progressing together to ${q}_{3}$. The three successors depict behaviors after the broadcast $(\qinit, !!d, q_1)$.
  %      where $\counter_{\qinit}$ is decremented, and the counter associated with the target summary is incremented.
In the first scenario (right location), the broadcasting process joins a \emph{new} summary (pink) $(\set{q_1}, {q}_{4}, 1)$. In the second scenario (diagonal location), it joins another \emph{new} summary (orange) $(\set{q_1}, {q}_{3}, 2)$, indicating a different timing for reaching ${q}_{3}$. In the third scenario (below location), the process joins an existing summary, and $q_1$ is added to the print.
%
%	
%	On the first scenario (right location), the process broadcasting $d$ joins a \emph{new} summary (the pink one) $(\set{q_1}, {q}_{4}, 1)$.
%	%: this next action state that this process will reach is ${q}_{4}$ and not ${q}_{3}$. 
%	On the second scenario (diagonal location) the process broadcasting $d$ joins a \emph{new} summary (the orange one) $(\set{q_1}, {q}_{3}, 2)$: meaning that it will reach ${q}_{3}$ but not at the same time than processes described by the summary labeled with $({q}_{3}, 1)$.
%%	as the processes described by the green scenario, the next action state it will reach is ${q}_{3}$, however, this will not happen at the same time, hence it must join a new summary with a different label $({q}_{3},2)$ rather than $({q}_{3}, 1)$. 
%	Note that, as we only consider well-formed executions, there can not be too many summaries with the same exit state as their can not be so much different moments (\ref{lemma:well-formed-execution:bound-next-action}). 
%	On the third scenario (below location), the process joins the existing summary and state $q_1$ is added to the print.
%	
%	
	Note that well-formed executions limit the number of summaries with the same exit state 
	 since there are at most $|Q_W|$ different moments where processes can reach a same action state
	(\cref{lemma:well-formed-execution:bound-next-action}).
	
	\cref{fig:vass-deleting-summary} illustrates summary deletion. The transition $(\qinit, !!b, q_7)$ causes the broadcasting process to join the pink summary (the bottom one), incrementing its counter. Processes in the green summary ($\set{{q}_{2}}, {q}_{3}, 1$) receive the message via $({q}_{2}, ?b, {q}_{3})$, reaching their exit state ${q}_{3}$. The green summary is deleted in the next location. By taking the loop transition at the right location, the value of $\counter_{({q}_{3},1)}$ is transferred to the counter $\counter_{{q}_{3}}$. If the loop is not taken enough times, $\counter_{({q}_{3},1)}$ may remain non-zero, but this does not affect correctness, as the counter will be decremented eventually, from a state in which there is no summary labeled with $({q}_{3}, 1)$, and not decrementing soon enough only delay the moment the corresponding processes will move from the action state ${q}_{3}$.
%	 \color{red} ici trouver une manière d'expliquer que ce n'est pas grave \color{black}
\end{example}

\begin{remark}\label{remark:vass:size}
	In the sequel of this paper, the size of $\mathcal{V}_\PP$ will be of interest to us (Sections 4.1 and 5.1). Hence, observe that $|\Loc| = |\CoherentSets|+ 3 \leq 2^{|Q_A|\times( |Q_W| +1) \times 2^{|Q_W|}}+3$ (as one summary is composed of a state in $Q_A$, one label in $[1, |Q_W| +1]$ and one set of waiting states), and $|\Counters| = |Q_A|+ |Q_A| \times( |Q_W| +1)$.
\end{remark}

%---------------------------------------------------------------------------------------------------

\subsubsection{Soundness of the construction}\label{subsubsec:target-wo-soundness-part}
We show now that if there exists a run in the VASS $\mathcal{V}_\PP$ from $(s_0, \mathbf{0})$ to $(s_f, \mathbf{0})$, then in the network built from $\PP$, there exist $C\in \mathcal{I}$ and $C' \in \CC$ such that $C \trans^\ast C'$ and  $C'(\aproc)=q_f$ for all $\aproc \in [1, ||C'||]$.

%\ars{I keep this here, cause I need to know where t is used} \lug{only used in appendix - section D}\nas{Du coup j'ai déplacé la notation dans l'annexe}      

%We first introduce some useful definition. 
We say that a configuration $(\ell, v)$ of $\mathcal{V}_\PP$ is an \emph{S-configuration} if $\ell=\SS$ for some $\SS \in \CoherentSets$. The \emph{implementation} of an  S-configuration $\lambda=(\SS,v)$ is the set of network configurations $\implem{\lambda} \subseteq \CC$ defined as follows: $C\in\implem{\lambda}$ if 
and only if there exists a function $f:[1,||C||]\rightarrow \Counters$ such that  $|f^{-1}(\counter)|=v(\counter)$ for all $\counter\in\Counters$, and
\begin{enumerate}[{\color{purple}\textbf{CondImpl}}1,leftmargin=7.5em]
	\item for all $\counter_q\in\Counters$ with $q \in Q_A$, we have $C(\aproc) = q$  for all $\aproc \in f^{-1}(\counter_q)$,;
	\label{cond:implem-definition:action-state}
	\item for all $\counter_{(q, k)}\in\Counters$, if there exists $(\print, q, k) \in \SS$ for some $\print\subseteq Q_W$, then
	 $C(\aproc) \in \print \cup \set{q}$ for all $\aproc \in f^{-1}(\counter_{{q,k}})$;
	 	\label{cond:implem-definition:present-summary}
	\item for all $\counter_{(q, k)}\in\Counters$, if $(q,k)\notin\lab{\SS}$, then $C(\aproc) = q$ for all $\aproc \in f^{-1}(\counter_{{q,k}})$.
	\label{cond:implem-definition:absent-summary}
\end{enumerate}

In an implementation of $\lambda$, the processes populate states according to the values of the counters. All processes associated with a counter $\counter_q$ of an
action state will populate this exact state~(\ref{cond:implem-definition:action-state}). However, processes associated with a counter $\counter_{(q, k)}$ of a summary do not necessarily
populate waiting states. This occurs when the label $(q,k)$ does not appear in $\SS$ while the associated counter remains strictly positive. Such a situation arises
 when a summary has been previously deleted, but its
counter has not been emptied. Since
all associated processes  have already reached the exit action state, the processes associated to the corresponding counter have to populate this active state (\ref{cond:implem-definition:absent-summary}).
  If the summary with label $(q,k)$ is active (i.e. its label appears in $\SS$), the processes associated with $\counter_{(q,k)}$ will be in the corresponding waiting region,
or in the exit action state (\ref{cond:implem-definition:present-summary}). This may seem contradictory, as all processes are supposed to reach their exit state simultaneously. However,
the processes already on the exit action state are ``old'' processes, that remained after a previous deletion of this summary (cf. \ref{cond:implem-definition:absent-summary}). These processes will be allowed to send messages only when they are identified with the counter of the active state, meaning when the corresponding
transition in $\Delta_e$ will be taken. 
% 
%For all summaries in $\SS$, then all the processes described by its associated counter are on the print of the summary or already on the target state (\ref{cond:implem-definition:present-summary}). If one label is absent from $\SS$ and its associated counter is above 0, then the associated processes are all on the target state (\ref{cond:implem-definition:absent-summary}).
%
The next lemma allows to build an execution of the protocol $\PP$ from an execution of the VASS $\mathcal{V}_\PP$. 
%
%states a soundness property: for two S-configurations in the VASS $\lambda \abconftrans \lambda'$ and for any configuraiton $C$ represented by $\lambda$, one can find a configuration $C'$ represented by $\lambda'$ and such that $C \trans^\ast C'$. To prove that, we make a process in the appropriate action state broadcast a message (when $\lambda \abconftrans \lambda'$ because of cases \ref{cond-abconftrans-broadcast-action-state}, \ref{cond-abconftrans-broadcast-join-existing} or \ref{cond-abconftrans-broadcast-new-summary}), and we make processes in the waiting states follow the summary they are associated with in \ref{cond:implem-definition:present-summary}. If $\lambda \abconftrans \lambda'$ because of case \ref{cond-abconftrans-empty}, then there is nothing to do as $C \in \implem{\lambda'}$. Formal proof can be found in appendix.
%%	If a summary disappears (because some processes on waiting states reached their goal), the partition of states should change. If a summary is created (because of condition \ref{cond-abconftrans-broadcast-new-summary}), then the partition should also change by creating a new set of one process (the )
%%Its proof can be found in appendix. 
%%	\color{red} ici mettre de l'intuition \color{black}
%
\begin{lemma}\label{lemma:soundness-local-ppty}
	Let $\lambda, \lambda'$ be two S-configurations of $\mathcal{V}_\PP$ and $C \in \implem{\lambda}$. If $\lambda \abconftransup{\delta} \lambda'$ in $\mathcal{V}_\PP$ then there exists $C' \in \implem{\lambda'}$ such that $C \trans^\ast C'$ for $\PP$. %\nas{inutile maintenant qu'on a deux notations différentes: for $\PP$.}
\end{lemma}
 \begin{proof}[Sketch of Proof.]
Assume $\lambda=(\SS,v)$ and $\lambda'=(\SS',v')$. Then either $(\SS,\delta, \SS')\in \Delta_e$ or $(\SS,\delta, \SS')\in \Delta_t$ for some $t\in T$. In the first
case, one can show that if $C\in\implem{\lambda}$, then $C\in \implem{\lambda'}$. Otherwise, let $t=(q,!!a,q')$ and there is a process $e$ such that $C(e)=q$.
 The transition $\delta$ in the VASS makes all the summaries and counters evolve according to the reception of the message $a$. According to the
different cases \ref{cond-abconftrans-broadcast-action-state}, \ref{cond-abconftrans-broadcast-join-existing} or \ref{cond-abconftrans-broadcast-new-summary},
one can build a configuration $C'\in\implem{\lambda'}$ such that $C\trans C'$.
\end{proof}

%\nas{utile? We are now able to prove the main lemma of this subsection.}

\begin{lemma}\label{lemma:soundness}
	If $(s_0, \mathbf{0}) \abconftrans^\ast (s_f, \mathbf{0})$ in $\mathcal{V}_\PP$, then there exist $C \in \mathcal{I}$ and $C' \in \CC$ such that $C \trans^\ast C'$ and $C'(\aproc) = q_f$  for all $\aproc \in [1, ||C'||]$.
\end{lemma}
\begin{proof}
	From the definition of $\mathcal{V}_\PP$, any run from $(s_0, \mathbf{0})$ to $(s_f, \mathbf{0})$ is of the form $(s_0, \mathbf{0}) \abconftrans^\ast (\emptyset, v_{init}) \abconftrans^\ast (\set{S_f}, v_f) \abconftrans^\ast (s_f, \mathbf{0})$ where $v_{init}(\counter_{\qinit}) = n$ for some $n \in \nat$ and $v_{init}(\counter) = 0$  for all  $\counter\in \Counters \setminus \set{\counter_{\qinit}}$, whereas $v_f(\counter_{(q_u, 1)}) = v_{init}(\counter_{\qinit})=n$ and  $v_{f}(\counter) = 0$ for all $\counter\in\Counters \setminus \set{\counter_{(q_u, 1)}}$. Define $C$ as the initial configuration in $\mathcal{I}$ with $n$ processes (i.e. $||C||=n$). Trivially, $C \in \implem{(\emptyset, v_{init})}$ and with \cref{lemma:soundness-local-ppty} and a simple induction, we get that there exists $C' \in \implem{(\set{S_f}, v_{f})}$ such that $C\trans^\ast C'$.
	Observe that by definition of $v_f$, and since $C'\in\implem{(\set{S_f}, v_f)}$, there exists $f:[1,n]\rightarrow \Counters$ such that
	$f^{-1}(\counter_{(q_u, 1)}) = [1, n]$.
	Using \ref{cond:implem-definition:present-summary}, we have $C'(\aproc) \in \set{q_f,q_u}$  for all $\aproc \in [1,n]$ with $n=||C'||$. Since $q_u$ is unreachable by construction, we deduce that $C'(\aproc) = q_f$ for all $\aproc \in [1,||C'||]$.	
\end{proof}

%---------------------------------------------------------------------------------------------------

\subsubsection{Completeness of the construction}
\label{subsubsec:target-wo-completeness-part}
%We now show that our simulation by the VASS $\mathcal{V}_\PP$ is complete for $\Target$, i.e. if there exists an execution for $\PP$ reaching a configuration where all processes are in $q_f$ then we have $(s_0, \mathbf{0}) \abconftrans^\ast (s_f, \mathbf{0})$ in $\mathcal{V}_\PP$. Note that thanks to \cref{cor:well-formed}, we can consider only well-formed executions for$\PP$.

Let us first introduce some new definitions. Given a well-formed execution $\rho \in \Lambda \cup \Lambda_\omega$, two indices $0 \leq i < j < |\rho|$ and an action state $q_a \in Q_A$, we define $E^{\rho,i}_{q_a, j}$, the set of processes in waiting states in $\rho_i$, and whose next action state is $q_a$, reached at $\rho_j$.
%the  $j$-th step of $\rho$. 
Formally, $E^{\rho,i}_{q_a, j}=\big\{\aproc\in [1,\nbagents{\rho}] \mid \rho_i(\aproc) \in Q_W\textrm{ and }
	  \nextaction{\rho, i, \aproc} = (q_a, j) \big\}$. Furthermore, an S-configuration $\lambda = (\SS, v)$ is a \emph{representative} of the configuration $\rho_i$ in $\rho$ iff the following conditions are respected:%\ars{Tali and Lucie, check if you are ok with this}\lug{ok -- à changer dans l'annexe : l'énum de condRepr1 - 3 et $E^{\rho, i}_{q_a,j}$ au lieu de $E^{ i}_{q_a,j}$}
\begin{enumerate}[{\color{violet}\textbf{CondRepr}}1,leftmargin=7.5em]
	\item \label{new-cond:repr:action-state} 
	for all $q \in Q_A$, $v(\counter_q) = |\rho_i^{-1}(q)|$;
\end{enumerate}
and for all $q_a\in Q_A$, there is an injective function $r_{q_a} : \nextactionindexset{\rho, i ,q_a} \rightarrow [1, |Q_W|+1]$ s.t.:
\begin{enumerate}[{\color{violet}\textbf{CondRepr}}1,leftmargin=7.5em]
        \setcounter{enumi}{1}
\item \label{new-cond:summaries}$\SS= \set{(\rho_i(E^{\rho,i}_{q_a,j}), q_a, r_{q_a}(j)) \mid q_a \in Q_A, j \in \nextactionindexset{\rho,i, q_a}}$
\item \label{new-cond:repr:renaming}
	for all $q_a\in Q_A$, we have  $v(\counter_{(q_a, r_{q_a}(j))}) = |E^{\rho,i}_{q_a, j}|$ for all $j \in \nextactionindexset{\rho, i, q_a}$ 
	%\item \label{new-cond:repr:null-counters} For 
	 and  $v(\counter_{(q_a, k)}) = 0$ for all $k \notin r_{q_a}(\nextactionindexset{\rho, i ,q_a})$.      
\end{enumerate}
%          
%%Let $\rho \in \Lambda \cup \Lambda_\omega$ be a \emph{well-formed} execution, $0 \leq i < |\rho|$, and let $C = \rho_i$. 
%% For all $q_a\in Q_A$ and $i<j<|\rho|$, let $E^i_{q_a, j} = \set{\aproc\in\nbagents{\rho} \mid C(\aproc) \in Q_W\textrm{ and }
%% 	  \nextaction{\rho, i, \aproc} = (q_a, j)}$ be the set of agents who are on waiting states at the $i$th step of the computation, and whose next active state is $q_a$,
%% 	  reached at the $j$-th step of the computation.
%         
%	  
%% An S-configuration $\lambda = (\SS, v)$ is a \emph{representative} of $C$ if, for all $q_a\in Q_A$, there exists an injective function $r_{q_a} : \nextactionindexset{\rho, i ,q_a} \rightarrow [1, |Q_W|+1]$ such that:
%% \begin{enumerate}[{\color{violet}\textbf{CondRepr}}1,leftmargin=7.5em]
%% 	\item \label{new-cond:repr:action-state} 
%% 	for all $q \in Q_A$, $v(\counter_q) = |C^{-1}(q)|$;
%% 	\item \label{new-cond:repr:renaming}
%% 	 for all $q_a \in Q_A$, for all $j \in \nextactionindexset{\rho, i, q_a}$, $v(\counter_{q_a, r_{q_a}(j)}) = |E^i_{q_a, j}|$
%% 	%\item \label{new-cond:repr:null-counters} For 
%% 	 and for all $q_a \in Q_A$, for all for all $k \notin r_{q_a}(\nextactionindexset{\rho, i ,q_a})$, $v(\counter_{q_a, k}) = 0$;
%% 	\item \label{new-cond:summaries}$\SS= \set{(C(E^i_{q_a,j}), q_a, r_{q_a}(j)) \mid q_a \in Q_A, j \in \nextactionindexset{\rho,i, q_a}}$.
%% \end{enumerate}

In a representative of the configuration $\rho_i$, the counters faithfully count the number of processes on active states (\ref{new-cond:repr:action-state}), and on waiting regions.
The set $E_{q_a,j}^{\rho,i}$ gathers exactly the set of processes that populate a same waiting region in a representative of $\rho_i$ in $\rho$: they  all reach the same next action state $q_a$ at
the same instant $j$. 
The set $\nextactionindexset{\rho,i,q_a}$ being bounded by $|Q_W|$ (cf. \cref{lemma:well-formed-execution:bound-next-action}), we can associate with each of these
indices a unique identifier, given by the injective function $r_{q_a}$. We use a spare identifier for technical reasons, to handle more easily situations when $|Q_W|$ summaries exist in the location of the VASS, and a new summary is created while one of the previous summaries is deleted at the next step.
Then, \ref{new-cond:summaries} and \ref{new-cond:repr:renaming} require that the counters corresponding to the summaries and the summaries themselves reflect
faithfully the situation in the configuration $\rho_i$ with respect to $\rho$. Observe that such a representative  is tightly linked to the simulated execution, since we need to
know the future behavior of the different processes to determine the different summaries. Finally, if we let $\abconffrom{\rho}{i}$ %\lug{nouvelle macro... pas convaincue} 
be the set of all $S$-configurations that are representatives of $\rho_i$ in $\rho$, we establish the following result before proving the main lemma of this subsection (\Cref{lemma:completeness}).

\begin{lemma}\label{lemma:target:completeness:local}
	Let $\rho$ be a \emph{well-formed} execution, and $0 \leq i < |\rho|$. For all $\lambda_i \in \abconffrom{\rho}{i}$, there exists $\lambda_{i+1} \in \abconffrom{\rho}{i+1}$ such that $\lambda_i \abconftrans^\ast \lambda_{i+1}$ in $\mathcal{V}_\PP$. 
\end{lemma}
\begin{proof}[Sketch of Proof.]
From an S-configuration $\lambda_i=(\SS_i, v_i)$ in $\abconffrom{\rho}{i}$, we can compute the effect on $\SS_i$ of the transition $t=(q,!!a,q')$ taken in $\rho$ to go from $\rho_i$ to $\rho_{i+1}$, and find a matching transition $(\SS_i, \delta, \SS_{i+1}) \in \Delta_t$ (according to whether $q'$ is a waiting state or not, and in the latter case, whether it joins an existing summary or a new one, depending of the execution $\rho$). We then obtain a new configuration $\lambda'=(\SS_{i+1},v')$ such that $\lambda_i \abconftransup{\delta} \lambda'$. Note that $\lambda'$ is not necessarily in $\abconffrom{\rho}{i+1}$:
if some summaries have been deleted between $\lambda_i$ and $\lambda'$, the transition $(\SS_i, \delta, \SS_{i+1})$  in $\mathcal{V}_\PP$ has not updated the corresponding
counters. Hence we need to apply transitions from $\Delta_e$ to  empty counters corresponding to deleted summaries and accordingly increase corresponding counters on matching action states to obtain $\lambda_{i+1}=(\SS_{i+1}, v_{i+1})$ in $\abconffrom{\rho}{i+1}$ .
\end{proof}

\begin{lemma}\label{lemma:completeness}
      If  there exist $C\in \mathcal{I}$, and $C' \in \CC$ such that $C \trans^\ast C'$ in $\PP$ and $C'(\aproc)=q_f$  for all $\aproc \in [1, ||C'||]$, then $(s_0, \mathbf{0}) \abconftrans^\ast (s_f, \mathbf{0})$ in $\mathcal{V}_\PP$.
\end{lemma}
\begin{proof}
%<<<<<<< HEAD
 Thanks to \cref{cor:well-formed}, we know there exists a well-formed execution $\rho$ from $C$ to $C'$.  Let $K = ||C||$ and $n=|\rho|-1$. First, consider the sequence $(s_0, \mathbf{0}) \abconftrans^\ast (s_0, v_0) \abconftrans (\emptyset, v_0)$ where $v_0(\counter_{\qinit}) = K$ and for all other counters $\counter$, $v_0(\counter_{}) = 0$. Observe that, since $\qinit\in Q_A$, $(\emptyset,v_0) \in \abconffrom{\rho}{0}$. By applying inductively \cref{lemma:target:completeness:local}, we get that $(\emptyset, v_0) \abconftrans^\ast (\SS_1, v_1)  \abconftrans^\ast \cdots \abconftrans^\ast (\SS_n, v_n)$ with $(\SS_n, v_n) \in \abconffrom{\rho}{n}$. By definition, every time that $\nextactionstate{\rho, i, \aproc}=q_u$ for some process $e$, then $\nextactionindex{\rho, i, \aproc}=n+1$, hence there will always be at most one summary with exit state $q_u$. So in the execution we build, every time that some summary $(\print, q_u,k)\in\SS_i$, we chose $k=1$.  %so w.l.o.g. we can assume that $\SS_n=\set{(\set{q_f}, q_u, 1)}= \set{S_f}$\footnote{All the processes that enter a summary with exit state $q_u$ will
%=======
% Thanks to \cref{cor:well-formed}, we know there exists a well-formed execution $\rho$ from $C$ to $C'$.  Let $K = ||C||$ and $n=|\rho|-1$. First, consider the sequence $(s_0, \mathbf{0}) \abconftrans^\ast (s_0, v_0) \abconftrans (\emptyset, v_0)$ where $v_0(\counter_{\qinit}) = K$ and for all other counters $\counter$, $v_0(\counter_{}) = 0$. Observe that, as $\qinit\in Q_A$, $(\emptyset,v_0) \in \abconffrom{\rho}{0}$. By applying inductively \cref{lemma:target:completeness:local}, we get that $(\emptyset, v_0) \abconftrans^\ast (\SS_1, v_1)  \abconftrans^\ast \cdots \abconftrans^\ast (\SS_n, v_n)$ with $(\SS_n, v_n) \in \abconffrom{\rho}{n}$. By definition, if $\nextactionstate{\rho, i, \aproc}=q_u$ for some process $e$, then $\nextactionindex{\rho, i, \aproc}=n+1$, hence there will always be one summary with exit state $q_u$. So in the run we build, each time some summary $(\print, q_u,k)\in\SS_i$, we chose $k=1$. \nas{we have $\nextaction{\rho, n, \aproc} = (q_u, n+1)$ for all $\aproc
%\in [1,K]$,} %so w.l.o.g. we can assume that $\SS_n=\set{(\set{q_f}, q_u, 1)}= \set{S_f}$\footnote{All the processes that enter a summary with exit state $q_u$ will
%>>>>>>> 60041575e2e165d9c99e9bc8f6d733c3a53e85f5
%be in the same summary, because }\nas{je me demande si c'est assez clair, et si ça suffit. En fait, dès qu'on construit un summary avec $q_u$, on peut lui mettre l'étiquette 1, est-ce qu'il faut le dire?} 
We then  have $v_n(\counter_{q_u,1}) = K$ and $v_n(\counter) = 0$ for all other $\counter \in \Counters \setminus \set{\counter_{q_u,1}}$. Hence, by construction of $\mathcal{V}_\PP$, we have $(\SS_n, v_n) \abconftrans (s_f, v_n) \abconftrans^\ast (s_f, \mathbf{0})$.
%	
%	\color{red} gérer le special résumé avec $q_u$, expliquer qu'on peut supposer sans perte de généralité que l'index est 1\color{black}
\end{proof}

Using \cref{lemma:target:completeness:local}\ and \cref{lemma:completeness}\ and the fact that the reachability problem for VASS is decidable (see \cref{vass-reach-complexity}) we get the main theorem of this section.
\begin{theorem}\label{th:upper-bound-target}
	$\Target$ restricted to Wait-Only protocols is decidable.
\end{theorem}

\subsection{Lower Bound for \Target~in Wait-Only Protocols}\label{subsec:synchro-lower-bound}
\label{subsec:lower-bound-target-waitonly}

\begin{figure}
	\begin{center}
		\tikzset{box/.style={draw, minimum width=4em, text width=4.5em, text centered, minimum height=17em}}

\begin{tikzpicture}[->, >=stealth', shorten >=1pt,node distance=2.5cm,on grid,auto, initial text = {}] 
	\node[state, initial] (q0) {$\qinit$};
	\node[state] (q1) [below = of q0, yshift = 23 ] {$q_1$};

%	\node[state] (q3) [right  = of q0, xshift = -20] {$p_1$};
%
%	\node[state] (q4) [right  = of q3, xshift = 0] {$p_2$};
	\node[state] (q6) [ right = of q0, xshift = 0, yshift =0] {$\ell_0$};
%	\node[state] (q7bis) [below  = of q4, xshift = 0, yshift = 20] {$p_3$};
	\node[state] (q7) [right = of q1, yshift = 0] {$\textsf{zero}$};
	\node[state] (sf) [right  = of q6, xshift = -20] {$\ell_f$};
	\node[state] (spf) [right  = of sf, xshift = 0] {$\ell'_f$};
	\node[state,accepting] (qf) [right = of spf ] {$q_f$};
	\node[state] (err) [right = of qf] {$\textsf{err}$};
	\node[state, inner sep= -2] (q8) [below right  = of q7, xshift = 20, yshift = 20] {$\textsf{unit}_{\counter}$};
	\node[state, inner sep= -2] (q9) [right of = q7, yshift = 15] {$\textsf{z-end}$};
	\node[state, inner sep= -2] (q10) [right of = q8, xshift = 40, yshift = 0] {$\textsf{err}'$};
%	\node[state] (q7) [right  = of q0, xshift = 50, yshift = 45] {$q_{7}$};
	%	\node[state] (q8) [right  = of q7] {$q_{8}$};
	
	\path[->] 
	(q0) 
	%edge [bend left = 20] node [] {$!!\textsf{election}$} (q1)
	edge [bend right = 0] node {$!!\textsf{start}$} (q6)
	edge [bend right = 0] node {$!!\$$} (q1)
	(q1) edge node {$?\textsf{start}$} (q7)
%	(q1) edge [bend right = 0] node {$?\textsf{election}$} (qf)
%	(q3) edge [bend left = 0] node {$?\textsf{election}$} (q4)
%	(q4) edge []node {$!!\textsf{election}$} (q6)
%	(q3) edge [bend right = 0] node [below, xshift = -20] {$?\textsf{election}$} (q7bis)
%	(q7bis) edge node {$?\textsf{election}$} (q7)
	(q7) edge [bend left = 10] node[sloped] {$!!\textsf{inc}_{\counter}$} (q8)
	(q8) edge [bend left = 20] node [sloped] {$!!\textsf{dec}_{\counter}$} (q7)
	(spf) edge [bend right = 0] node []{$!!\textsf{verif}$}  (qf)
%	(q9) edge [out=-30, in=60, looseness = 6]  node {$?\textsf{end}_2$}(qf)
	
	(sf) edge node {$?\textsf{end}$} (spf)
	(q7) edge node [sloped] {$!!\textsf{end}$} (q9)
	
	(q9) edge [bend right = 15]node [sloped] {$?\textsf{verif}$} (qf)
	;
	
%	\draw[->] (q9.east) |- ++(2,0) -| ++(0,3.8) -| (qf.north); 
%	\node[anchor=west, xshift = 55, yshift=20] at (q9.east) {$?\textsf{end}_2$};
	
		\node[draw, fill = cyan, fill opacity = 0.2, text opacity = 1, fit=(q6) (sf), text height=0.09 \columnwidth, inner sep = 3.5] (PV) {$\PP'_{\mathcal{V}}$};
	
	\path[->,dashed]
	(PV) edge [bend left = 20] node {$?\textsf{start}$} (err);

	\path[->,very thick]
	(qf) edge node {$?\textsf{verif}$} (err)
	(q9) edge [bend right = 10]node [sloped] {\shortstack{$?\textsf{inc}_\counter, ?\textsf{dec}_\counter$ $\counter \in \Counters$}} (q10); 
\end{tikzpicture}
	\end{center}
	\vspace{-10pt}
	\caption{Protocol $\PP$ associated to a VASS $\mathcal{V}$. The dashed edge $(\ell_f, ?\textsf{start}, \textsf{err})$ will only be considered in Section 4.2.}\label{fig:prot-from-vass}
\end{figure}
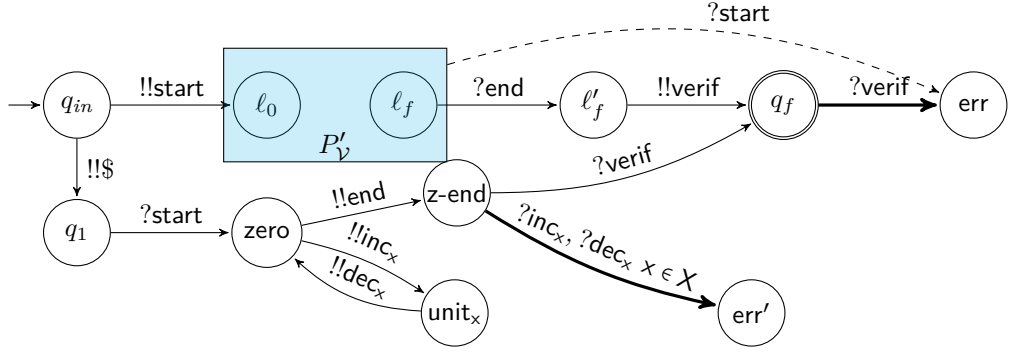

In this subsection we reduce the reachability problem for VASS to \Target. %we present a lower bound for the \Target~problem for Wait-Only Protocols. We show how to simulate a VASS for the reachability problem with Wait-Only Protocols.
We fix a VASS $\mathcal{V} = (\Loc, \ell_0, \Counters, \Delta)$ and a final location $\ell_f\in \Loc$. Without loss of generality, we suppose that any transition in 
$\Delta$ either increments or decrements of 1 exactly one counter. Hence, for a transition $(\ell,\delta,\ell')\in\Loc$, we might describe $\delta$ by giving only
$\delta(\counter)$ for the only $\counter\in\Counters$ such that $\delta(\counter)\neq 0$. 
%As we can suppose wlog that a VASS has only transitions incrementing or decrementing of 1 exactly one counter, we do so in the rest of this subsection. 
%To this end, fix a VASS $\mathcal{V} = (\Loc, s_0, \Counters, \Delta)$ and a final location $s_f$ to reach. The \Reachability~problem asks if there exists a run from $(s_0, \mathbf{0})$ to $(s_f, \mathbf{0})$. 

We explain how to build a protocol $\PP_\mathcal{V}$ that simulates $\mathcal{V}$: it is depicted in \cref{fig:prot-from-vass} \emph{in which we don't consider the dashed edge}. 
The states $\textsf{zero}$ and $\set{\textsf{unit}_\counter\mid\counter\in\Counters}$  represent the evolution of the different counters of the VASS while the
 blue box $\PP'_{\mathcal{V}}$ describes the evolution of $\mathcal{V}$ with respect to its locations. Formally,  $\PP'_{\mathcal{V}}$ consists of a set of
 states $Q_{\mathcal{V}} = \Loc \cup \set{\frownie}$ and a set of transitions $T_{\mathcal{V}} = \set{(\ell, ?\textsf{inc}_{\counter}, \ell') \mid (\ell, 
\delta, \ell') \in \Delta \textrm{ and } \delta(\counter) = 1} \cup \set{(\ell, ?\textsf{dec}_{\counter}, \ell') \mid (\ell, \delta, \ell') \in \Delta \textrm{ and } \delta(\counter) = -1}
 \cup \set{(\ell, ?m, \frownie) \mid  m \in \set{\textsf{inc}_{\counter}, \textsf{dec}_{\counter} \mid \counter \in \Counters}}$.
%We recall that we assume that all the transitions of the VASS $(\ell, \delta, \ell') \in \Delta$ are such that there exists exactly one counter $\counter$ such that $\delta(\counter) \neq 0$ and $\delta(\counter) \in \set{-1, 1}$. 
 Some transitions of $T_{\mathcal{V}}$ are depicted in \cref{fig:prot-from-vass-pv}.
% \ars{In the figure you can put as well nearby an example $(\ell'', \delta,\ell''')$ of a decrement}
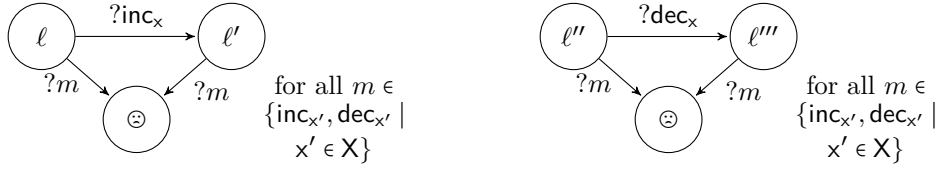
\begin{figure}
  \begin{minipage}{.5\textwidth}
  	\begin{center}
  		\tikzset{box/.style={draw, minimum width=4em, text width=4.5em, text centered, minimum height=17em}}

\begin{tikzpicture}[->, >=stealth', shorten >=1pt,node distance=2.5cm,on grid,auto, initial text = {}] 
	\node[state] (q0) {$\ell$};
	\node[state] (q1) [right = of q0, yshift = 0 ] {$\ell'$};
	\node[state] (bad) [below = of q0, xshift = 1.25cm , yshift = 1.4cm] {$\frownie$};
	\node[] (text) [right =of bad, yshift = 0.0cm, xshift = 0.1cm] {
		\shortstack{
		for all $m \in $\\
		$\set{\textsf{inc}_{\counter'}, \textsf{dec}_{\counter'} \mid $\\
			$\counter' \in \Counters}$}};
	
	\path[->] 
	(q0) edge node {$?\textsf{inc}_{\counter}$} (q1)
	(q0) edge node [left, yshift = -2]{$?m$} (bad)
	(q1) edge node [yshift = 2]{$?m$} (bad)
	;

%	\node[draw, fill = cyan, fill opacity = 0.2, text opacity = 1, fit=(q6) (sf), text height=0.11 \columnwidth] (VASS) {$\PP_{\mathcal{V}}$};
\end{tikzpicture}
  	\end{center}
  \end{minipage}%
  \begin{minipage}{0.5\textwidth}
  	\begin{center}
  		\tikzset{box/.style={draw, minimum width=4em, text width=4.5em, text centered, minimum height=17em}}

\begin{tikzpicture}[->, >=stealth', shorten >=1pt,node distance=2.5cm,on grid,auto, initial text = {}] 
	\node[state] (q0) {$\ell''$};
	\node[state] (q1) [right = of q0, yshift = 0 ] {$\ell'''$};
	\node[state] (bad) [below = of q0, xshift = 1.25cm , yshift = 1.4cm] {$\frownie$};
	\node[] (text) [right =of bad, yshift = 0.0cm, xshift = 0.1cm] {
		\shortstack{
			for all $m \in $\\
			$\set{\textsf{inc}_{\counter'}, \textsf{dec}_{\counter'} \mid $\\
				$\counter' \in \Counters}$}};
	
	\path[->] 
	(q0) edge node {$?\textsf{dec}_{\counter}$} (q1)
	(q0) edge node [left, yshift = -2]{$?m$} (bad)
	(q1) edge node [yshift = 2]{$?m$} (bad)
	;

	%	\node[draw, fill = cyan, fill opacity = 0.2, text opacity = 1, fit=(q6) (sf), text height=0.11 \columnwidth] (VASS) {$\PP_{\mathcal{V}}$};
\end{tikzpicture}
  	\end{center}
%  	\caption{Part of $\PP'_\mathcal{V}$ that simulates the transition $(\ell, \delta,\ell')\in\Delta$ with $\delta(\counter)=-1$.}
  \end{minipage}
	\caption{Part of $\PP'_\mathcal{V}$ that simulates respectively the transition $(\ell, \delta,\ell')\in\Delta$ with $\delta(\counter)=1$ at the left, and $(\ell'', \delta,\ell''')\in\Delta$ with $\delta(\counter)=-1$ at the right.}
	\label{fig:prot-from-vass-pv}
\end{figure}
We then obtain the following theorem.
\begin{theorem}\label{th:synchro-ack-complete}
	\Target\ with Wait-Only protocols is Ackermann-complete.
\end{theorem}

%The fact that this protocol can simulate an execution of the VASS is formalized in the following lemma.
%\begin{lemma}\label{lemma:lower-bound-target-waitonly:completeness}
%	There exist $C_0\in\mathcal{I}$ and $C_f\in \CC$ such that $C_0 \trans^\ast C_f$ in $\PP_\mathcal{V}$ and  $C_f(\aproc) = q_f$ for all $\aproc \in [1, ||C_f||]$ \emph{if and only if} $(\ell_0, \mathbf{0}) \abconftrans^\ast (\ell_f, \mathbf{0})$ in $\mathcal{V}$.
%\end{lemma}
\begin{proof}[Sketch of Proof.]
The upper bound follows from~\cref{th:upper-bound-target} and \cref{vass-reach-complexity}. For the lower bound, we give the ideas that prove that the reduction 
described above is correct. In particular, there exist $C_0\in\mathcal{I}$ and $C_f\in \CC$ such that $C_0 \trans^\ast C_f$ in $\PP_\mathcal{V}$ and  $C_f(\aproc) = q_f$ for all $\aproc \in [1, ||C_f||]$ \emph{if and only if} $(\ell_0, \mathbf{0}) \abconftrans^\ast (\ell_f, \mathbf{0})$ in $\mathcal{V}$.

Assume first that $(\ell_0, \mathbf{0}) \abconftrans^\ast (\ell_f, \mathbf{0})$ in $\mathcal{V}$. Let $K\in\mathbb{N}$ be the maximum value $\Sigma_{\counter\in\Counters} v(\counter)$ encountered during the execution of $\mathcal{V}$. We choose $C_0$ such that $||C_0||=K+1$. The execution of the protocol goes as follows. All processes except for one (called the leader) broadcast a (dummy message) $\$$. The remaining leader process broadcasts $\textsf{start}$ and reaches $\ell_0$, causing all
 processes in $q_1$ to reach $\textsf{zero}$. These processes will simulate the counter values throughout the execution of the VASS, while the leader at $\ell_0$ simulates the VASS itself by traversing its locations in $\PP'_{\mathcal{V}}$. The value of the counter $\counter$ in the different steps of the execution of the VASS is represented by the number of processes in $\textsf{unit}_x$. To simulate an increment of $\counter$ in $\mathcal{V}$ from $(\ell_i, v_i)$ to $(\ell_{i+1}, v_{i+1})$, one of the processes in state $\textsf{zero}$ sends the message $\textsf{inc}_\counter$. This message is received by the leader, making it transitionning
  to $\ell_{i+1}$. If the
 protocol's configuration correctly represented $(\ell_i, v_i)$, then the new configuration faithfully represents $(\ell_{i+1}, v_{i+1})$. Once the simulation of the execution of the VASS is finished,
 the leader process is in $\ell_f$, and other processes remain in $\textsf{zero}$. At this point, all processes in $\textsf{zero}$ broadcast the message 
 $\textsf{end}$, and the leader, now on 
 $\ell'_f$, broadcasts the message $\textsf{verif}$, which gathers everyone in $q_f$. 

To prove the other direction, we need to prove that the processes in the protocol cannot cheat in their simulation of the VASS. First observe that to gather all the processes in $q_f$, 
an execution of the protocol needs to see the broadcast of $\textsf{verif}$ exactly once (a second broadcast of this message will necessarily send the first process to have broadcast it
in the state $\textsf{err}$). This enforces that in the execution of the protocol, exactly one process will evolve in the $\PP_\mathcal{V}$ part, simulating correctly the states of the VASS.
Second, observe that the modification of the counters must follow the transitions
of the VASS. %\nas{il me semble que cette phrase est inutile, puisque la suivante couvre tous les cas.No decrement can occur if the corresponding counter is empty since in that case, no process will be able to send the adequate message. }
The processes cannot modify a counter
in a way that is not allowed by the VASS: the process in the states of $\PP_\mathcal{V}$ would receive the message and would reach losing state $\frownie$. 
Finally, at the end of the simulation (when the leader reaches $\ell_f$), all the other processes must be in state $\textsf{zero}$. Otherwise the following will happen when trying to reach $q_f$: one process (called the helper process) broadcasts $\textsf{end}$, the message is received by the leader, which reaches $\ell'_f$. Then, all processes (in states $\textsf{zero}$ and $\textsf{unit}_\counter$) need to reach state $\textsf{z-end}$ before the unique broadcast of $\textsf{verif}$, as otherwise, they can never receive the message and reach $q_f$. If a process in state $\textsf{unit}_\counter$ joins state $\textsf{z-end}$, it has to broadcast message $\textsf{dec}_\counter$, causing processes in state $\textsf{z-end}$ (at least the helper process, and processes
from state $\textsf{zero}$ that may have already reached $\textsf{z-end}$) to receive the message and reach state $\textsf{err}'$. As a consequence, the helper process can no longer receive $\textsf{verif}$ and reach $q_f$, preventing the execution from gathering all the processes in $q_f$. Hence, the simulation correctly ends in a configuration where all processes but the leader are on state $\textsf{zero}$, faithfully simulating counters' values.
%
%Processes in state $\textsf{zero}$ can safely join state $\textsf{z-end}$ as message $\textsf{end}$ is not received by anyone at that point. However, when a process in state $\textsf{unit}_\counter$ joins state $\textsf{z-end}$, it has to broadcast message $\textsf{dec}_\counter$, causing processes in state $\textsf{z-end}$ (including the helper process) to receive the message and reach state $\textsf{err}'$. As a consequence, the helper process can no longer receive $\textsf{verif}$ and reach $q_f$, preventing the execution to put everyone in $q_f$. Hence, the simulation correctly ends in a configuration where all processes but the leader are on state $\textsf{zero}$, faithfully simulating counters' values.
\end{proof}

\section{\Target\ is Easier when the Target State is an Action State}\label{sec:synchro-action}

In this section, we prove that when $q_f \in Q_A$, the $\Target$ problem becomes $\expspace$-complete. For the upper bound, we show that \Target\ can be reduced to a simpler problem on VASS. The lower bound comes
from the coverability problem on VASS. 
%For both membership and hardness proofs we reuse their counterparts in the case $q_f\in Q_W$. Membership proof reuses the VASS presented in \cref{sec:def-vass}\ (with small modifications), but this time we will see that our problem reduces to a simpler problem on VASS: mutual reachability. Hardness proof reuses the protocol defined in \ref{subsec:lower-bound-target-waitonly}, again with small modifications, in order to simulate a VASS but this time, for the coverability problem.  

%\lugtext{expliquer comment on obtient membership}

\subsection{\Target\ if $q_f\in Q_A$ is in \expspace}
%
%We present here the reduction from \Target\ when $q_f \in Q_A$ to the mutual reachability problem on VASS.
We show that \Target\ when $q_f\in Q_A$ can be reduced to the mutual reachability problem on VASS, that is defined as follows.
% \nas{j'ai laissé runs au lieu d'executions pour les VASS dans cette section, c'était lourdingue executions (et ça rajoute des lignes)}
%
Given a VASS $\mathcal{V} = (\Loc, \ell_0, \Counters, \Delta)$ and a location $\ell_f$, the mutual reachability problem on VASS asks if there exist a run from $(\ell_0, \mathbf{0})$ to $(\ell_f, \mathbf{0})$ and a
run from $(\ell_f, \mathbf{0})$ to $(\ell_0, \mathbf{0})$. This problem is shown to be in \expspace\ \cite{Leroux11}.
More precisely, we get the following theorem.
\begin{theorem}[\cite{Leroux11} Corollary 10.6, Transformation of \cite{HopcroftP79}]\label{thm:leroux:hopcroft:bound-runs}
	Given a VASS $\mathcal{V} = (\Loc, \ell_0, \Counters, \Delta)$, if there are two runs $(\ell_0, \mathbf{0}) \abconftrans^\ast (\ell_f, \mathbf{0})$ and  $(\ell_f, \mathbf{0}) \abconftrans^\ast (\ell_0, \mathbf{0})$, then there are two runs
	 $(\ell_0, \mathbf{0}) \abconftrans^\ast (\ell_f, \mathbf{0})$ and $(\ell_f, \mathbf{0})  \abconftrans^\ast (\ell_0, \mathbf{0})$ whose  lengths are bounded by $17(|\Counters| + 3)^2x^{15(|\Counters|+3)^{|\Counters| + 5}}$ where $x = (1+ 2(|\Loc| +1)^2)^2$.
\end{theorem}
\begin{remark}
	In \cite{Leroux11}, Leroux establishes the \expspace -membership and provides a bound on the lengths of the runs, for the mutual reachability problem in VAS (Vector Addition Systems). 
	However, by applying the VASS-to-VAS transformation presented in \cite{HopcroftP79}, we get the aforementioned bound. 
	%Details can be found in appendix.
%	\lugtext{todo}
\end{remark}

%We will use this theorem to conclude about the complexity of the \Target\ problem when $q_f \in Q_A$. We are now ready to present the reduction. 
%This reduction builds on the VASS defined in \cref{sec:def-vass}. 
Given a Wait-Only protocol $\PP = (Q, \Sigma, \qinit, T)$ and a target state $q_f\in Q_A$, we explain how to transform it into a VASS where mutual reachability
is equivalent to \Target. We build the VASS $\mathcal{V}_\PP = (\Loc, \Counters, s_0, \Delta)$ as described in \cref{sec:def-vass}, with some
modifications. The first two modifications simply encode the fact that the target state is not a waiting sate anymore, while the third one ensures the mutual
reachability. %We apply the following modifications to adapt the VASS for the case $q_f \in Q_A$:
\begin{itemize}
	\item We add a state $s'_f$ and the transition $(\set{S_f}, \mathbf{0}, s_f)$ is replaced by two transitions: $(\emptyset, \delta, s'_f)$ and $(s'_f, \delta', s_f)$. Here, $\delta(\counter_{q_f}) = -1$, $\delta'(\counter_{q_f}) = 1$, and for all other counters $\counter$, $\delta(\counter) = \delta'(\counter) = 0$. This is meant to prevent the reachability of $(s_f, \mathbf{0})$ to be trivial with the run $(s_0, \mathbf{0}) \abconftrans (\emptyset, \mathbf{0}) \abconftrans (s_f, \mathbf{0})$.
	\item The transition $(s_f, \delta_f, s_f)$ is replaced by $(s_f, \delta'_f, s_f)$, where $\delta'_f(\counter_{q_f}) = -1$ and $\delta'_f(\counter) = 0$ for all other counters $\counter$.
	\item We add a transition $(s_f, \mathbf{0}, s_0)$ to allow the resetting of $\mathcal{V}_P$ to its initial state after reaching $s_f$.
\end{itemize}

%The addition of $s'_f$ % is not considered a positive instance of the mutual reachability problem.

 \begin{lemma}\label{lemma:synchro-action:mutu-reach:reduction}
 	There exists $C \in \mathcal{I}$ and $C' \in \CC$ such that $C \trans^\ast C'$ and, for all $\aproc \in [1, ||C'||]$, $C'(\aproc) = q_f$ if and only if there is a run $(s_0, \mathbf{0})  \abconftrans^\ast  (s_f, \mathbf{0})$ and a run $(s_f, \mathbf{0}) \abconftrans^\ast (s_0, \mathbf{0})$ in $\mathcal{V}_P$.
 \end{lemma}
 \begin{proof}[Sketch of Proof.]
%\emph{Completeness of the reduction.}
Assume that there exist $C \in \mathcal{I}$ and $C' \in \CC$ such that $C \trans^\ast C'$ and, for all $\aproc \in [1, ||C'||]$, $C'(\aproc) = q_f$. 
Since the main body of $\mathcal{V}_P$ is the same as in~\cref{sec:def-vass} we can, like in \cref{lemma:completeness}, %(and \cref{cor:well-formed}), 
build a run of $\mathcal{V}_\PP$ from $(s_0, \mathbf{0})$ to $(s_f, \mathbf{0})$. By construction of $\mathcal{V}_\PP$, 
$(s_f, \mathbf{0}) \abconftrans (s_0, \mathbf{0})$.%Using techniques from the case where $q_f \in Q_W$, it follows that if there is an execution of the protocol that puts all processes in $q_f$, then one can construct a run from $(s_0, \mathbf{0})$ to $(s_f, \mathbf{0})$. By construction, $(s_f, \mathbf{0}) \abconftrans (s_0, \mathbf{0})$, which ensures completeness of the reduction.

%\emph{Soundness of the reduction.}
%We now provide intuition for the soundness of the reduction. 
Assume now that $(s_0, \mathbf{0}) \abconftrans^\ast (s_f, \mathbf{0})$ (and $(s_f, \mathbf{0}) \abconftrans (s_0, \mathbf{0})$). 
%Then there exists a sequence of vectors $v_0, \dots, v_n$ with $v_0 = v_n = \mathbf{0}$ such that $(s_0, v_0) \abconftrans^\ast (s_f, v_1) \abconftrans (s_0, v_1) \abconftrans^\ast \cdots \abconftrans^\ast (s_f, v_n)$.
Observe that it might be the case that the run looks like $(s_0, \mathbf{0}) \abconftrans^\ast (s_f, v_1)\abconftrans (s_0, v_1)\abconftrans^\ast (s_f, v_2)\dots
(s_0, v_k) \abconftrans^\ast (s_f,\mathbf{0})$. 
By construction of the VASS, we know that there is an execution 
$(s_0, \mathbf{0}) \abconftrans^\ast (\emptyset, v'_0)\abconftrans^\ast (\emptyset, v''_0)\abconftrans(s'_f, v'''_0)\abconftrans (s_f, v''_0)\abconftrans^\ast (s_f, v_1)$.
Here $v'_0$ is the valuation obtained after having increased the counter $\counter_{\qinit}$, $v''_0$ is the valuation obtained just before going to $s'_f$, and
$v_1$ is obtained after having decremented counter $\counter_{q_f}$.  
Adapting the proof of~\cref{lemma:soundness} to this case, we deduce the existence of an execution of the protocol from an initial configuration $C_0$ to a configuration $C'_0$ such that $|{C'_0}^{-1}(q)|=0$ for all
$q\in Q_W$, and $|{C'_0}^{-1}(q)|=v''_0(\counter_q) + \sum_{i=1,\dots,{|Q_W|+1}}v''_0(\counter_{(q,i)})$ for all $q\in Q_A$. 
From the portion of run of the VASS $(s_0, v_1)\abconftrans^\ast (\emptyset, v'_1)\abconftrans^\ast(\emptyset, v''_1)\abconftrans (s'_f, v'''_1)\abconftrans (s_f, v''_1)\abconftrans^\ast (s_f, v_2)$, we similarly get another sequence of configurations of the protocol from a configuration $C_1$ (not necessarily initial
because $v_1$ might not be equal to 0 for some counters other than $\counter_{\qinit}$) to a configuration
$C'_1$ such that $C_1 \trans^\ast C'_1$ and $|{C'_1}^{-1}(q)|=0$ for all
$q\in Q_W$. %\nas{j'ai enlevé cette phrase que je trouvais pénible: , and $|{C'_1}^{-1}(q)|=v''_1(\counter_q)$ for all $q\in Q_A$.} 
Observe however that these two sequences can be merged into a single one, of size
 $v'_0(\counter_{\qinit})+ ( v'_1(\counter_{\qinit}) - v_1(\counter_{\qinit}))$ (the second part corresponds to processes added in $\qinit$ with transition 
 $(s_0, \delta_{\textit{in}}, s_0)$, and $\delta_{in}(\counter_{\qinit}) = 1$).
The execution then
first behaves like the execution from $C_0$ to $C'_0$, potentially leaving some processes in the initial state (in particular processes from $v'_1(\counter_{\qinit}) - v_1(\counter_{\qinit})$).
%>>>>>>> fa82394090714329f5f49b80950ed8d3ff1ad3db
%\lug{pareil ici, c'était écrit $v'_1(\counter_{\qinit})$}). 
Once this first part is over,
all the processes are either in the initial state, or in an action state (because ${C'_0}^{-1}(q)=0$ for all $q\in Q_W$). Then, one can simulate the second sequence, 
(processes already in $q_f$ from the first execution won't be affected because they are in an action state, so they won't receive any message. This would not 
be true if $q_f\in Q_W$). Since the execution 
of $\mathcal{V}_P$ eventually reaches $(s_f, \mathbf{0})$, it means that it reaches $(s_f, v_{k+1})$ where $v_{k+1}(\counter)=0$ for all $\counter\neq \counter_{q_f}$. Then, by
iterating the construction described above, one obtains an execution of $\PP$ from an initial configuration to a configuration $C_f$ such that $C_f(e)=q_f$ 
for all $e\in [1,||C_f||]$. \end{proof}

As runs of doubly exponential length can be guessed in exponential space in the size of the VASS and $\expspace = \textsc{NExpspace}$, we get the following theorem using \cref{remark:vass:size,thm:leroux:hopcroft:bound-runs}.
Observe that even if the number of locations is doubly exponential in the size of the protocol, the bound of \cref{thm:leroux:hopcroft:bound-runs}~is polynomial in $|\Loc|$ and doubly exponential in $|\Counters|$, hence lengths of the runs to guess remain doubly exponential in the size of the protocol.\lug{ici }
%\lug{ici je cite la remarque sur la taille du vasss}, 
%that checking mutual reachability in the VASS constructed from a protocol $\PP$ requires only an exponential amount of space relative to the size of the protocol. 
%Hence the following theorem.

\begin{theorem}\label{thm:target-wo-qf-action-membership}
	 \Target\ for Wait-Only protocols is in \expspace\ when $q_f \in Q_A$.
\end{theorem}

\subsection{\Target\ if $q_f\in Q_A$ is \expspace-hard}

We get $\expspace$-hardness by proving that if one can solve $\Target$ problem with $q_f \in Q_A$, then one can solve the coverability problem in VASS, which
is \expspace-hard~\cite{lipton76}, and stated as follows: given a VASS $\mathcal{V} = (\Loc, \ell_0, \Counters, \Delta)$ of dimension $d \in \nat$, and a location $\ell_f \in \Loc$, decide whether there is an execution from $(\ell_0, \mathbf{0})$ to $(\ell_f, v)$ for some $v \in \nat^d$.
%
%
%\begin{theorem}[\cite{lipton76,Rackoff78}]
%	%[Coverability \cite{lipton76,Rackoff78}]
%	%, Mutual Reachability \cite{Leroux11}]
%	\label{thm:cover-mutu-reach-expspace-complete}
%	The coverability problem 
%	%and the mutual reachability problem 
%	for VASS 
%	%are both 
%	is $\expspace$-complete.
%\end{theorem}

%To show that the $\Target$ problem with $q_f$ as an action state is $\expspace\ $-hard, we reduce the coverability problem in VASS to it.
%We present here the reduction from the coverability problem in VASS to $\Target$ with $q_f \in Q_A$.
%Consider a VASS $\mathcal{V} = (\Loc, \Counters, s_0, \Delta)$ and $s_f \in \Loc$. 
For the reduction, we use the protocol depicted in \cref{fig:prot-from-vass}, this time including the dashed edge but excluding the thick edges. Hence, $q_f$ becomes an action state. Then $\ell_f$ is coverable if and only if there exists an execution $C_0 \trans^\ast C_f$ in which all processes reach $q_f$ starting from an initial configuration $C_0 \in \mathcal{I}$.
If $\ell_f$ is coverable, the execution of the VASS can be simulated in the protocol like in \cref{subsec:synchro-lower-bound}. The processes that may still remain in 
the states $\textsf{unit}_\counter$ can reach $\textsf{z-end}$ even when $\ell'_f$ is populated (recall that all the thick edges have been removed). Once this is done,
all processes can gather in $q_f$. 
%However, in this case, some processes may still be in $\textsf{unit}_\counter$ for some $\counter \in \Counters$ when the leader process reaches $s'_f$. Since the thick edges from $\textsf{z-end}$ to $\textsf{err}'$ are not present, processes in unit states can broadcast $\textsf{dec}_\counter$ and $\textsf{end}$ until all processes reach $\textsf{z-end}$. The rest of the execution proceeds as before: the leader process broadcasts $\textsf{verif}$, which is received by all other processes, placing all processes in $q_f$.
Conversely, if there exists an execution $C_0 \trans^\ast C_f$ in which all processes reach $q_f$, it means that $\ell_f$ is coverable. Let $\aproc_0$ be the first process to broadcast $\textsf{start}$, then an execution of the VASS covering $\ell_f$ can be built based on the sequence of states visited by $\aproc_0$ between $\ell_0$ and $\ell_f$.
This relies on three keys observations:
	(1) The process $\aproc_0$ must visit $\ell_f$, otherwise it cannot reach $q_f$.
	(2) When $\aproc_0$ is between $\ell_0$ and $\ell_f$, $\aproc_0$ is the only process in this region. If another process attempts to join by broadcasting $\textsf{start}$, $\aproc_0$ will receive the message and go to $\textsf{err}$, preventing it from reaching $q_f$.
	(3) When $\aproc_0$ reaches $\ell_0$ for the first time, no other process is in the states $\set{\textsf{unit}_\counter\mid \counter\in\Counters}$, since $\textsf{start}$ is sent for the first time at this point.
%
%Using these observations, an execution of the VASS covering $s_f$ can be reconstructed similarly to the case $q_f \in Q_W$.
%
%From \cref{thm:cover-mutu-reach-expspace-complete}\ and \cref{thm:target-wo-qf-action-membership}, we get the main theorem of this section:
Together with \cref{thm:target-wo-qf-action-membership}, this establishes the following result. 
\begin{theorem}
	$\Target$ for Wait-Only protocols is $\expspace$-complete when $q_f\in Q_A$.
\end{theorem}

\section{Solving the Repeated Coverability problem for Wait-Only Protocols}\label{sec:repeated-cover}

\subparagraph*{Upper Bound.}
\label{sec:rep-cover-in-expspace}
%\lugtext{toute cette section est à changer}
We show now that \RepeatedCover\ on Wait-Only protocols is in \expspace. This is done via the repeated coverability problem in VASS, which is stated as
follows: given a VASS $\mathcal{V}=(\Loc,\ell_\textit{init}, \Counters, \Delta)$ and a location $\ell_f\in \Loc$, is there an infinite execution $(\ell_{init}, \mathbf{0}) \abconftrans(\ell_1, v_1) \abconftrans\dots$  such that, for all $i \in \nat$, there exists $j > i$ such that $\ell_j = \ell_f$?

%In this section, we show that Repeated Coverability with Wait-Only protocols is in \textsc{ExpSpace}. To prove this, we show a reduction from this problem to the repeated coverability problem in VASS which asks: given a VASS  $\mathcal{V} = (\Loc, \ell_{init}, \Counters, \Delta)$ and a location $\ell_f\in \Loc$, whether there exists an infinite run $(\ell_{init}, \mathbf{0}) (\ell_1, v_1) \dots$  s.t. that for all $i \in \nat$, there exists $j > i$ such that $\ell_j = \ell_f$?
%\begin{decproblem}
%	\problemtitle{$\RepeatedCover_{\text{VASS}}$~}
%	\probleminput{A VASS $\mathcal{V} = (\Loc, \ell_{init}, \Counters, \Delta)$ and a location $\ell_f\in \Loc$.} 
%	\problemquestion{Is there an infinite run $(\ell_{init}, \mathbf{0}) (\ell_1, v_1) \dots$  s.t. that for all $i \in \nat$, } \problemquestionline{ there exists $j > i$ such that $\ell_j = \ell_f$?}
%\end{decproblem}
%\color{red}todo: problem set of locations to cover infinitely often + theorem from \cite{Habermehl97}\color{black}
%In \cite{Habermehl97}, Habermehl proves the following theorem.
We rely on the following theorem.
\begin{theorem}[Theorem 3.1 of \cite{Habermehl97}]\label{thm:habermehl97}
	For a VASS $\mathcal{V} = (\Loc, \ell_\textit{init}, \Counters, \Delta)$,
	the repeated coverability problem\ is solvable in $O(\log(l) + \log(|\Loc|))2^{c|\Counters|\log(|\Counters|)}$ nondeterministic space for some constant $c$ independent of $\mathcal{V}$, and where the absolute values of components of vectors in $\Delta$ are smaller than $l$.
\end{theorem}

Let $\PP= (Q, \Sigma, \qinit, T)$ be a Wait-Only protocol and $t_f=(q,!!a, q')$ the transition that has to occur infinitely often. We build a VASS $\mathcal{V}_P$ based on the construction
presented in~\cref{sec:def-vass}. The difference is that we add a new set of states $\CoherentSets_{\smiley}=\set{(\SS,\smiley)\mid \SS\in\CoherentSets}$.
%\nas{on pourrait aussi
%supprimer $s_f$} 
We also add the transitions
$\set{(\SS,\delta,(\SS',\smiley))\mid (\SS,\delta,\SS')\in \Delta_{t_f}}\cup\set{(\SS,\smiley),\mathbf{0}, \SS)\mid \SS\in\CoherentSets}$. This means that when
a process takes the transition $t_f$, it is highlighted in $\mathcal{V}_P$ by going in a location tagged with $\smiley$, before continuing the execution. 
%We depict an example in \cref{fig:vass-transformed-repeated-cover}: t
Taking the protocol of~\cref{fig:example-1}, when the transition $t_f=(\qinit, !!d, q_1)$ is
the transition that has to occur infinitely often, 
%we show how
the previous transition from the location of the VASS $\set{(\set{q_3}, q_4, 1)}$ to the location $\set{(\set{q_3}, q_4, 1), (\set{q_1}, q_6, 1)}$ (see \cref{fig:vass-new-summary}) is
now transformed into two transitions in the VASS we build here: one from $(\set{(\set{q_3}, q_4, 1)})$ to the location $(\set{(\set{q_3}, q_4, 1), (\set{q_1}, q_6, 1)}, \smiley)$ and one from $(\set{(\set{q_3}, q_4, 1), (\set{q_1}, q_6, 1)}, \smiley)$ to $\set{(\set{q_3}, q_4, 1), (\set{q_1}, q_6, 1)}$. %\lug{ai enlevé la figure et modifié le texte ici}
Then, there is an execution of $\PP$ with a process that takes $t_f$ infinitely often if and only if there is an execution of $\mathcal{V}_P$ where
one state in $\CoherentSets_{\smiley}$ is visited infinitely often. The procedure to decide \RepeatedCover\ is then the following: in $\mathcal{V}_P$,
guess a location $\ell_f$ in $\CoherentSets_{\smiley}$ and solve the repeated coverability problem for $\mathcal{V}$ and $\ell_f$. From~\cref{remark:vass:size,thm:habermehl97},
this can be done in $O(\log(2) + (|Q|+1)^2 \times 2^{|Q|} +2)2^{c(|Q| + 1)^3\log((|Q| +1)^3)}$ nondeterministic space. The overall procedure is then in 
$\textsc{NExpSpace} = \expspace$. Hence, the following theorem. 
\begin{theorem}
	\RepeatedCover\ is in \expspace.
\end{theorem}

\begin{remark}
One could define \RepeatedCover\ with a reception transition that has to occur infinitely often. The VASS we described hereabove can be adapted by enlarging each location
with the current state of the witness process. This can also be done in \expspace. 
\end{remark}

 \subparagraph*{Lower bound.}
 \label{sec:pspace-hard}
 
  We reduce the intersection non-emptiness problem for deterministic finite automata, which is known
 to be \pspace-complete \cite{Kozen77}, to \RepeatedCover, hence obtaining the following result.

 %In this subsection, we prove that the \RepeatedCover\ problem is \pspace -hard.
 
 \begin{theorem}
	\RepeatedCover\ is \pspace -hard.
\end{theorem}

 Let $\mathcal{A}_1, \dots, \mathcal{A}_n$ be a list of deterministic finite and \emph{complete} automata with $\mathcal{A}_i = (A, Q_i, q_{i}^0, \set{q_i^f},\linebreak[1] \Delta_i)$ for all $1 \leq i \leq n$. The problem we reduce takes as input automata with a unique accepting state, which does not change the complexity of the 
 problem. 
 We let $A^\ast$ be the set of words over the finite alphabet $A$ and $\Delta^\ast_i$ the function extending $\Delta_i$ to $A^\ast$: for all $q \in Q_i$, $\Delta_i^\ast(q, \varepsilon) = q$, and for all $w \in A^\ast$ and $a \in A$, $\Delta^\ast_i(q, wa) = \Delta_i(\Delta^\ast_i(q,w), a)$.

\begin{figure}
	\begin{minipage}[c]{0.50\linewidth}
		\tikzset{box/.style={draw, minimum width=4em, text width=4.5em, text centered, minimum height=17em}}

\begin{tikzpicture}[->, >=stealth', shorten >=1pt,node distance=2.0cm,on grid,auto, initial text = {}] 
	\node[state] (q1) [above right = of q0, yshift = -10 ] {$q_1$};
	\node[state] (q2) [right = of q1, xshift =-10] {$q_2$};
	\node[state] (q3) [right = of q2] {$q_3$};
	\node[state] (qn) [right = of q3, xshift = 10] {$q_{n+1}$};
	\node[state] (frb) [above = of q3, yshift =25] {$\frownie$};
	\node[state] (qnp) [below = of q3, yshift = 35, xshift = -25] {$q_{n+2}$};

	\path[->] 
	(q1) edge [bend left = 20] node[sloped] {\shortstack{$?\#, ?\text{end}_i,$ \\$ 1 \leq i \leq n$}} (frb)
	(q1) edge  node {$?\$$} (q2)
	(q2) edge  node {$?\text{end}_1$} (q3)
	(qn) edge [bend left = 15] node {$?\text{end}_n$} (qnp)
	(qnp) edge [bend left = 15] node {$!!\smiley$} (q1)
	(q2) edge node [sloped]{\shortstack{$?m$, $ m\in \Sigma$}} (frb)
	(q3) edge node [sloped] {\shortstack{$?m$, $ m \in \Sigma$}} (frb)
	(qn) edge [bend right = 20]node [sloped] {\shortstack{$?m$, $ m\in \Sigma$}} (frb)
	;

	\path (q3) -- node[auto=false]{\ldots} (qn);
\end{tikzpicture}
		\caption{The part of $\PP$ with transition $t_f = (q_{n+2}, !!\smiley, q_1)$. A process enters this protocol through transition $(\qinit, !!\#, q_1)$.}\label{fig:pspace-hard-rec}
	\end{minipage} \hfill
	\begin{minipage}[c]{0.40\linewidth}
		\tikzset{box/.style={draw, minimum width=4em, text width=4.5em, text centered, minimum height=17em}}

\begin{tikzpicture}[->, >=stealth', shorten >=1pt,node distance=2.0cm,on grid,auto, initial text = {}] 
	\node[state] (qi0) [] {$q_i^0$};
	\node[state] (qif) [right = of qi0, xshift =-10] {$q_i^f$};
	\node[state] (qid) [right = of qif, xshift =-10] {$q_i^{\$}$};
	\node[state] (qie) [right = of qid] {$q_i^{\text{e}}$};
	\node[state] (fr) [above = of qid, yshift = -10, xshift = -20] {$\frownie$};
	
	\path[->] 
	(qif) edge node {$?\$$} (qid)
	(qid) edge [bend left =0] node {$!!\text{end}_i$} (qie)
	(qie.south) edge [bend left =15] node {$?\smiley$} (qi0.south)
	(qie) edge [bend right =20] node [sloped] {\shortstack{$?m, m \nin \{\text{end}_{i+1},  $\\
		$ \dots, \text{end}_n\}$}} (fr)
	
	;

	\node[draw, fill = cyan, fill opacity = 0.2, text opacity = 1, fit=(qi0) (qif), text height=0.05 \columnwidth, inner sep = 0.5] (Ai) {$\PP_i$};
	
	\path[->] 
	(Ai.north) edge [bend left = 15] node [sloped]{\shortstack{$?\$, ?i$\\
			$ ?\text{end}_{1, \dots, n}$}} (fr);

\end{tikzpicture}
		\caption{The part of $\PP$ simulating automaton $i$. A process enters this protocol through transition $(\qinit, !!i, q_i^0)$.}\label{fig:pspace-hard-auto}
	\end{minipage}
\end{figure}

We define a protocol $\PP = (Q, \Sigma, \qinit, T)$ that consists of several components.
%illustrated in \cref{fig:pspace-hard-rec} and \cref{fig:pspace-hard-auto}. 
For each automaton $\mathcal{A}_i$, $\PP$ contains one component pictured on 
\cref{fig:pspace-hard-auto}. Here,  the blue box $\PP_i = (Q^i, T^i)$ is defined as follows: $Q^i$ is the set of states of automaton $i$ (i.e. $Q_i$) and $T^i = \set{(q_1^i, ?a, q_2^i) \mid (q_1^i, a , q_2^i) \in \Delta_i}$. The other main component, (\cref{fig:pspace-hard-rec}), includes the transition $t_f = (q_{n+2}, !!\smiley, q_1)$, that must be taken infinitely often. A process taking $t_f$ infinitely often also receives infinitely often the sequence of messages 
$\$ \cdot \text{end}_1 \cdot  \text{end}_2 \dots  \text{end}_n$. Each message $\text{end}_i$ is broadcast from the component corresponding to $\mathcal{A}_i$
(\cref{fig:pspace-hard-auto}).
In addition to transitions described in \cref{fig:pspace-hard-rec,fig:pspace-hard-auto}, $T$ contains the following set of transitions: $\set{(\qinit, !!\#, q_1)} \cup \set{(\qinit, !!i, q_i^0 )\mid 1 \leq i \leq n} \cup \set{(\qinit, !!a, \qinit) \mid a \in A} \cup \set{(\qinit, !!\$, \qinit)}$.
%
%In \cref{fig:pspace-hard-auto}, the blue box $\PP_i = (Q^i, T^i)$ is defined as follows: $Q^i$ is the set of states of automaton $i$ (i.e. $Q_i$) and $T^i = \set{(q_1^i, ?a, q_2^i) \mid (q_1^i, a , q_2^i) \in \Delta_i}$.

%The first part, shown in \cref{fig:pspace-hard-rec}, includes the critical transition $t_f = (q_{n+2}, !!\smiley, q_1)$, which must be taken infinitely often. A process taking infinitely often $t_f$ will also receive continuously the sequence of messages $\$ \cdot \text{end}_1 \cdot  \text{end}_2 \cdots  \text{end}_n$. 

%Each message $\text{end}_i$ is broadcast from the second part of the protocol, (\cref{fig:pspace-hard-auto}).
If there is a word $w=a_1\dots a_n$ in the intersection of the languages of the $n$ automata, the execution where, from the initial configuration, we reach
a configuration with one process in $q_1$, one process in $\qinit$, and 
one process in each of the $q^0_i$ will go as follows: the process in $\qinit$ broadcasts the sequence of messages $a_1 \cdots a_n$, received by all the processes
in the parts simulating the automata. Since the word is accepted by all the automata, they all reach $q_i^f$ together. Then, the process in $q_1$ broadcasts
$\$$, receives the sequence of $\text{end}_1\dots\text{end}_n$, and broadcasts $\smiley$, that gathers all the processes that are $q_i^\text{e}$ in $q_i^0$.
This can be repeated an infinite number of times, and $t_f$ will occur infinitely often.
%Speicifically, for each automaton $\mathcal{A}_i$, a process: starts in $q_0^i$, receives a sequence of messages $a_1 \cdots a_n$ simulating a word in $A^\star$, until it reaches $q_i^f$, receives $\$$, broadcasts $\text{end}_i$, receives $\smiley$ and restarts the simulation of $\mathcal{A}_i$. 

Conversely, assume that the transition $t_f$ occurs infinitely often. At the beginning of the execution, 
it is possible that the automata are not faithfully simulated, because several processes populate the blue zone of some $P_i$. A process can then reach the state
$q_i^f$ with a different word than the others. However, this can happen only a finite number of times in an infinite execution, because every time a new process
reaches the blue zone, the processes that already populate it are sent to the state $\frownie$,  which is a sink state. Also, only one process at a time populates
the zone pictured in~\cref{fig:pspace-hard-rec}. So in an execution where $t_f$ occurs
infinitely often, there exists a process $e_0$ that takes it infinitely often, and eventually no process reaches $\frownie$. Once we have reached this point, when
 $\aproc_0$ takes transition $t_f$, there is exactly one process in $q_i^0$ for all $1 \leq i \leq n$. When $\$$ is broadcast, all the processes have received
 the same word, and they are all in a final state of their respective automaton (unless they go to $\frownie$, contradicting our assumption). So the word is in the intersection of languages. 

\bibliography{bib}

\newpage
\appendix
\section{\Target~is undecidable}
\label{app-sec0}

In order to prove undecidability of \Target, we reduce the halting problem of a 2-counter machine. 

\paragraph*{2-counter machines}
A 2-counter machine is a tuple $M =(S, s_{in}, s_f, \counter_1, \counter_2, \Delta)$ such that $S$ is a finite set of states, $s_{in} \in S$ is the initial state, $s_f \in S$ is the final state, $\counter_1$ and $\counter_2$ are two counters that take their values in $\nat$, and $\Delta \subseteq S \times \textrm{Op} \times S$ is a finite set of transitions, with $\textrm{Op}=\set{\inc{\counter}, \dec{\counter}, \counter =0 \mid \counter \in \set{\counter_1, \counter_2}}$. A configuration of the machine is a tuple $(s, x_1, x_2)$ with $s \in S$, and $x_1 \in \nat$ (resp. $x_2 \in \nat$) is the value of counter $\counter_1$ (resp. $\counter_2$).

Consider a transition $\delta = (s, \text{op}, s') \in \Delta$, then for $x_1, x_2 \in \nat$, we have:
\begin{itemize}
	\item $(s, x_1, x_2) \abconftransup{\delta} (s', x_1 +1 , x_2)$ iff $\text{op} = \inc{\counter_1}$;
	\item $(s, x_1, x_2) \abconftransup{\delta} (s', x_1 , x_2 + 1)$ iff $\text{op} = \inc{\counter_2}$;
	\item $(s, x_1, x_2) \abconftransup{\delta} (s', x_1 -1 , x_2)$ iff $x_1 > 0$ and $\text{op} = \dec{\counter_1}$;
	\item $(s, x_1, x_2) \abconftransup{\delta} (s', x_1  , x_2 - 1)$ iff $x_2 > 0$ and $\text{op} = \dec{\counter_2}$;
	\item $(s, x_1, x_2) \abconftransup{\delta} (s', x_1  , x_2)$ iff $x_1 = 0$ and $\text{op} = (\counter_1 =0 )$;
	\item $(s, x_1, x_2) \abconftransup{\delta} (s', x_1  , x_2)$ iff $x_2 = 0$ and $\text{op} = (\counter_2 =0 )$;
\end{itemize}
We let $\Delta_{\textrm{test}}=\set{(s,\textrm{op}, s') \in \Delta\mid\textrm{op}\in \set{\counter_1=0, \counter_2=0}}\subseteq \Delta$. 
Given two configurations $(s, x_1, x_2)$ and $(s', x'_1, x'_2)$, we write $(s, x_1, x_2) \abconftrans (s', x_1', x_2')$ whenever there exists a transition $\delta \in \Delta$ such that $(s, x_1, x_2) \abconftransup{\delta} (s', x_1', x_2')$.
Finally we denote $\abconftrans^\ast$ the reflexive and transitive closure of $\abconftrans$. We assume that $\Delta$ is deterministic, and that $s_f$ has no outgoing 
transition. 
The halting problem asks whether $(s_{in}, 0, 0) \abconftrans^\ast (s_f, 0, 0)$ and it is undecidable.

\paragraph*{Simulation.} 
Given a machine $M =(S, s_{in}, s_f, \counter_1, \counter_2, \Delta)$, we define a protocol $\PP$ depicted in \cref{fig:prot-from-machine}, where $\PP_M$ is defined as follows: $\PP_M = (S \cup \set{\frownie}, T_M)$ with $T_M = \set{(s, ?\text{op}, s') \mid (s, \text{op}, s') \in \Delta\setminus\Delta_\textrm{test}} \cup
\set{(s, !!\text{\counter=0}, s') \mid (s, \text{\counter=0}, s') \in \Delta_{\textrm{test}}}
\cup \set{(s, ? \text{op}, \frownie) \mid s\in S, \text{op} \in \set{\inc{\counter}, \dec{\counter}\mid \counter \in \set{\counter_1, \counter_2}} }$.

\begin{figure}
	\begin{center}
		\tikzset{box/.style={draw, minimum width=4em, text width=4.5em, text centered, minimum height=17em}}

\begin{tikzpicture}[->, >=stealth', shorten >=1pt,node distance=2.5cm,on grid,auto, initial text = {}] 
	\node[state, initial] (q0) {$\qinit$};
	\node[state] (0) [below = of q0, yshift = 40, xshift = 40 ] {$0$};

	\node[state] (s0) [ above = of q0, xshift = 40, yshift =-40] {$s_{in}$};

	%	\node[state] (q7bis) [below  = of q4, xshift = 0, yshift = 20] {$p_3$};
	\node[state] (x1) [below = of 0, xshift = 80, yshift = 50] {$x_1$};
	\node[state] (x2) [below = of 0, xshift = -80, yshift = 50] {$x_2$};
	
	\node[state,accepting] (sf) [right  = of s0, xshift = 20] {$s_f$};
	
	\node [state] (err) [below = of 0, yshift = 10] {\frownie};

%	\node[state, inner sep= -2] (q8) [below right  = of q7, yshift = 10] {$\textsf{unit}_{\counter}$};
%	\node[state, inner sep= -2] (q9) [right of = q7, yshift = 10] {$\textsf{z-end}$};
%	\node[state, inner sep= -2] (q10) [right of = q8] {$\textsf{err}'$};
	%	\node[state] (q7) [right  = of q0, xshift = 50, yshift = 45] {$q_{7}$};
	%	\node[state] (q8) [right  = of q7] {$q_{8}$};
	
	\path[->] 
	(q0) edge [bend left = 10] node {$!!\ell$} (s0)
	edge [bend right = 10] node {$?\ell$} (0)
	(0) edge [bend left = 10] node[xshift = 0] {$!!\inc{\counter_1}$} (x1)
	edge [bend left = 10] node [xshift = -10] {$!!\inc{\counter_2}$} (x2)
	(x1) edge [bend left = 10] node [xshift = 10]{$!!\dec{\counter_1}$} (0)
	(x2)edge [bend left = 10] node  {$!!\dec{\counter_2}$} (0)
	(0) edge [bend right = 10]node [below]{$?d$} (sf)
	(sf) edge [loop above] node {$!!d$} ()
	(x1) edge node {$?\counter_1=0, ?d$} (err)
	(x2) edge node [below, xshift = -10]{$? \counter_2 =0, ?d $} (err)
	;
	
	%	\draw[->] (q9.east) |- ++(2,0) -| ++(0,3.8) -| (qf.north); 
	%	\node[anchor=west, xshift = 55, yshift=20] at (q9.east) {$?\textsf{end}_2$};
	
	\node[draw, fill = cyan, fill opacity = 0.2, text opacity = 1, fit=(s0) (sf), text height=0.11 \columnwidth, inner sep = 1] (PV) {$\PP_{{M}}$};
	
%	\path[->,dashed]
%	(PV) edge [bend left = 25] node {$?\textsf{start}$} (err); 
	
%	
%	\path[->,very thick]
%	(qf) edge node {$?\textsf{verif}$} (err)
%	(q9) edge node [yshift = -10] {\shortstack{$?\textsf{inc}_\counter, ?\textsf{dec}_\counter$\\ $\counter \in \Counters$}} (q10); 
\end{tikzpicture}
	\end{center}
	\caption{Protocol $\PP$ associated to a machine $M =(S, s_{in}, s_f, \counter_1, \counter_2, \Delta)$.}\label{fig:prot-from-machine}
\end{figure}

If $(s_{in}, 0, 0) \abconftrans^\ast (s_f, 0, 0)$, there is an execution of the protocol that gathers all the processes in $s_f$. The number of processes needed
is $1+\max(x_1+x_2)$ where $\max(x_1+x_2)$ denotes the maximum value of the sum of the counters during the execution of $M$. Initially, a single process (the leader),
broadcasts $\ell$, and goes to $s_{in}$, while all the others go to the state $0$. From that point onward, the configurations of the 2-counter machines are encoded in a the protocol's configurations as follows: 
the state of the leader process determines the location of $M$, and the number of processes in state $x_1$ (resp. $x_2$) represents the value of counter 
$\counter_1$ (resp. $\counter_2$).

From a location $s$, if the transition taken in the execution of $M$ belongs to $\Delta_{\textrm{test}}$, the leader process broadcasts the corresponding message 
$!!\counter_i=0$. Since the counter $\counter_i$ is zero at that moment, no process is in state $x_i$, allowing the simulation to continue. If the transition involves an increment (respectively, a decrement) of counter $\counter_i$, then, by our choice of the number of processes, there are some processes in state 0 (respectively, in state $x_i$). One such process broadcasts the message $!!\counter_i++$ (respectively, $!!\counter_i--$), which causes the leader to go to the next location. The protocol configuration remains equivalent to the configuration of $M$. Since $(s_f,0,0)$ is reached in $M$, this simulation ensures that the leader process eventually reaches state $s_f$, while all other processes remain in state 0. At that moment, the leader broadcasts $!!d$, bringing all processes to $s_f$.

Conversely, if there exists an execution of the protocol that gathers all processes in $s_f$, then at some point, one process must have broadcast $!!\ell$, after which all other processes must have reached state 0. From this point, if any process deviates from the intended behavior—such as a counter sending an incorrect operation to the leader or the leader sending a zero-test message when the corresponding counter is not empty—an erroneous message is received by some processes, leading them to the sink state $\frownie$, from which it is impossible to gather all processes in $s_f$. When the leader process eventually reaches $s_f$ (which must happen since the execution gathers all processes in $s_f$), the only way to ensure that all other processes join is by broadcasting $!!d$. If any process is not in state 0 at that moment, it will be sent to the sink state $\frownie$. Consequently, the corresponding execution in the 2-counter machine must be correct and must reach $(s_f,0,0)$.

\section{Proof of \cref{lemma:resolving-non-determinism}}
\begin{proof}
	Let $\rho \in \Lambda \cup \Lambda_\omega$ be an execution of $\PP$ such that there exist $\aproc_1, \aproc_2 \in [1, \nbagents{\rho}]$ and $0 \leq j_0 < |\rho|$ with $\rho_{j_0}(\aproc _1) = \rho_{j_0}(\aproc_2) \in Q_W$ and $\nextaction{\rho, j_0 ,\aproc_1} = (q,j_1)$ and $\nextaction{\rho, j_0, \aproc_2} = (q,j_2)$ with $j_1 \leq j_2 < |\rho|$.
	
	We define the sequence of configurations $\rho' = \rho'_{j_0+1} \cdots \rho'_{j_2}$ as follows: 
	\begin{itemize}
		\item for all $j_0 < k \leq j_1$, $\rho'_k(\aproc_2) = \rho_k(\aproc_1)$ and for all $\aproc \neq \aproc_2$, $\rho'_k(\aproc) = \rho_k(\aproc)$;
		\item for all $j_1 < k \leq j_2$, $\rho'_k(\aproc_2) = q$ and for all $\aproc \neq \aproc_2$, $\rho'_k(\aproc) = \rho_k(\aproc)$.
	\end{itemize}
	Observe that $\rho_{j_2}(\aproc_2) = q$ as $\nextaction{\rho, j_0, \aproc_2} = (q,j_2)$. Hence, $\rho'_{j_2} = \rho_{j_2}$.
	%Furthermore, by assumption, $\rho_{j_0}(\aproc _1) = \rho_{j_0}(\aproc_2)$, hence $\rho'_{j_0} = \rho_{j_0}$. 
	Hence, it remains to prove that $\rho_{j_0} \trans \rho'_{j_0 +1} \trans \cdots \trans \rho'_{j_2}$ in order to prove that $\rho_{\leq  j_0} \cdot \rho'\cdot \rho_{\geq j_2+1}$ is an execution.
	
	We do so by induction: we let $\rho'_{j_0}=\rho_{j_0}$ and we prove by induction on $j_0 \leq k \leq j_2$ that $\rho'_{j_0} \trans \rho'_{j_0 +1} \trans \cdots \trans \rho'_k$.
	
	For $k = j_0$, there is nothing to do.
	% as $\rho'_{j_0, k} = \rho'_{j_0}$.
	Let $j_0 \leq k < j_2$ and assume that $\rho'_{j_0} \trans \rho'_{j_0 +1} \trans \cdots \trans \rho'_k$. Let us prove that $\rho'_{k} \trans \rho'_{k+1}$. 
	To do so, let $\rho_k \transup{\aproc_0, t} \rho_{k+1}$ with $t = (q, !!a, q')$.
	We distinguish between three cases:
	\begin{itemize}
		\item Case $k < j_1$. First observe that $k < j_1 = \nextactionindex{\rho, j_0, \aproc_1} \leq \nextactionindex{\rho, j_0, \aproc_2} = j_2$, hence $\rho_k(\aproc_1) \in Q_W$, $\rho_k(\aproc_2) \in Q_W$ and so $\aproc_0 \neq \aproc_1$ and $\aproc_0 \neq \aproc_2$. 
		
		Observe that we are in the case where $\rho'_k(\aproc_2) = \rho_k(\aproc_1)$ and $\rho'_{k+1}(\aproc_2) = \rho_{k+1}(\aproc_1)$. As $\rho_k \trans \rho_{k+1}$, either
		$(\rho'_k(\aproc_2), ?a, \rho'_{k+1}(\aproc_2)) = (\rho_k(\aproc_1), ?a, \rho_{k+1}(\aproc_1)) \in T$ or $a \nin R(\rho_k(\aproc_1)) = R(\rho'_k(\aproc_2))$ and $\rho'_{k+1}(\aproc_2) = \rho_{k+1}(\aproc_1) = \rho_k(\aproc_1) = \rho'_k(\aproc_2)$. Hence, as for all other process $\aproc$, $\rho'_k(\aproc) = \rho_k(\aproc)$ and $\rho'_{k+1}(\aproc) = \rho_{k+1}(\aproc)$, it holds that $\rho'_k\transup{\aproc_0, t} \rho'_{k+1}$.
		
		\item Case $k > j_1$. First observe that $k < j_2 =  \nextactionindex{\rho, j_0, \aproc_2}$, hence $\rho_k(\aproc_2) \in Q_W$, and so $\aproc_0 \neq \aproc_2$.
		
		Observe that we are in the case where $\rho'_k(\aproc_2) = \rho'_{k+1}(\aproc_2) = q$. Recall that $q \in Q_A$. Hence, $a \nin R(q) = \emptyset$.
		
		As $\rho'_k(\aproc) = \rho_k(\aproc)$ and $\rho'_{k+1}(\aproc) = \rho_{k+1}(\aproc)$ for all other process $\aproc$, it holds that $\rho'_k \transup{\aproc, t} \rho'_{k+1}$.
		 As in the previous case, we conclude that $\rho'_k \transup{\aproc_0, t} \rho'_{k+1}$.
		
		\item Case $k = j_1$. In that case $\rho'_{j_1}(\aproc_2) = \rho_{j_1}(\aproc_1) = q$ and $\rho'_{j_1 +1}(\aproc_2) = q$. As $q \in Q_A$, $a \nin R(q) = \emptyset$ and so, as $\rho'_k(\aproc) = \rho_k(\aproc)$ and $\rho'_{k+1}(\aproc) = \rho_{k+1}(\aproc)$ for all other process $\aproc \neq \aproc_2$, it holds that $\rho'_k \transup{\aproc_0, t} \rho'_{k+1}$.
	\end{itemize}

\end{proof}

\section{Proofs of \cref{subsubsec:target-wo-soundness-part}}
\label{app-sec:target-wo-soundness-part}
Let $\lambda = (\SS, v)$ and $\lambda' = (\SS', v')$ be two S-configurations and $t = (q, !!a, q') \in T$. When $\lambda \abconftransup{\delta} \lambda'$ and $(\SS, \delta, \SS') \in \Delta_t$, we say that $\lambda \abconftransup{t} \lambda'$.

\begin{proof}[Proof of \cref{lemma:soundness-local-ppty}]
Let $\lambda\abconftransup{\delta}\lambda'$ with $\lambda=(\SS,v)$ and $\lambda'=(\SS',v')$. Let $C\in\implem{\lambda}$ with $||C||=n$ and 
$f:[1,n]\rightarrow \Counters$.

\begin{itemize}
\item If $\delta\in \Delta_e$, then $\SS=\SS'$ and let $(p,k)\notin\lab{\SS}$ such that $\delta(\counter_p)=1$ and $\delta(\counter_{p,k})=-1$. In that case, we let
$C'=C$ and $C\trans^\ast C'$. It remains to show that $C=C'\in\implem{\lambda'}$. For that, we will build another function $f':[1,n]\rightarrow \Counters$, that
reflects the modification of counters $\counter_{p,k}$ and $\counter_p$. Since $\delta(\counter_{p,k})=-1$, $v(\counter_{p,k})>0$ and hence $|f^{-1}(\counter_{p,k})|>0$.
Thus, let $e\in f^{-1}(\counter_{p,k})$. Since $(p,k)\notin\lab{\SS}$, by \ref{cond:implem-definition:absent-summary}, $C(e)=C'(e)=p$. Then, we let $f'(e)=\counter_p$,
and for all other processes $e'\in[1,n]$, we let $f'(e)=f(e)$. Then, $|{f'}^{-1}(\counter_p)| = |f^{-1}(\counter_p)|+1$, $|{f'}^{-1}(\counter_{p,k})|=f^{-1}(\counter_{p,k})|-1$,
and for all other counter $\counter$, $|{f'}^{-1}(\counter)| = |f^{-1}(\counter)|$. Then, $v'(\counter_p)=v(\counter_p)+1 = |f^{-1}(\counter_p)|+1 = |{f'}^{-1}(\counter_p)|$,
$v'(\counter_{p,k})= v(\counter_{p,k})-1 = |f^{-1}(\counter_{p,k})|-1 = |{f'}^{-1}(\counter_{p,k})|$, and for all other counter $v'(\counter)=v(\counter)=|{f'}^{-1}(\counter)|$.
Since $C'(e)=p$, it is easy to see that $f'$ meets the requirements \ref{cond:implem-definition:action-state}, \ref{cond:implem-definition:present-summary}\ and \ref{cond:implem-definition:absent-summary}, hence $C'\in\implem{\lambda'}$. 

\item Otherwise, $\delta\nin \Delta_e$ and let $t=(q,!!a,q')$ such that $\lambda\abconftransup{t}\lambda'$. Since $\delta(\counter_q)=-1$, $v(\counter_q)>0$ and hence there exists $e_q\in [1,n]$ such 
that $f(e_q)=\counter_q$, and, by  \ref{cond:implem-definition:action-state}, $C(e_q)=q$. We build $C'$ as follows: $C'(e_q)=q'$. Then for all other $e\in [1,n]$, if
$a\notin R(C(e))$, we let $C'(e)=C(e)$. Otherwise, $C(e)\in Q_W$. From the fact that $C\in\implem{\lambda}$, we deduce that $f(e) = \counter_{(p,k)}$, with 
$(p,k)\in\lab{\SS}$, and let $S=(\print,p,k)\in \SS$. Now,
	\begin{itemize}
	\item if $(p,k)\notin\lab{\SS'}$, it means that $S\abtransup{t}\arrived$, and we let $C'(e)=p$. By definition of $S\abtransup{t}\arrived$, we deduce that $(C(e), ?a, p)\in T$. 
	\item Otherwise, let $S'=(\print', p, k)\in \SS'$. By definition of the VASS, it means that either $S\abtransup{t} S'$, or $S\abtransup{t, +q'} S'$. In both cases, 
	there exists a configuration $C''$ such that $\confstate{{q}} \oplus \conf{\print}   \transup{1, t} \confstate{{q'}} \oplus C''$ and either $\setof{C''} = \print'$ or
	$\setof{C''}\cup\set{q'}=\print'$. Let $e'$ such that $\confstate{{q}} \oplus \conf{\print}(e')=C(e)$ (we will justify its existence just after), then, we define $C'(e)=\confstate{{q'}} \oplus C''(e')$. Observe that we know that
	such a process $e'$ exists, thanks to \ref{cond:implem-definition:present-summary}: since $e\in f^{-1}(\counter_{p,k})$, $C(e)\in \print\cup\set{p}$. Moreover, $C(e)\in Q_W$,
	hence $C(e)\in\print$. Then we know that $(C(e), ?a, C'(e))\in T$ or $C(\aproc) = C'(\aproc)$ and $a \nin R(C(\aproc))$. 
	\end{itemize}
	
	It is clear from the construction of $C'$ that $C\trans C'$. We show now that $C'\in\implem{\lambda'}$ by defining $f':[1,n]\rightarrow\Counters$ as follows. 
	
	\begin{align*}
	f'(e_q) &=\begin{cases} \counter_{q'} & \textrm{if $q'\in Q_A$}\\
					\counter_{p,k} & \textrm{ if $q'\in Q_W$, and $\lab{\SS}=\lab{\SS'}$, with $\counter_{p,k}$ the unique counter such that $\delta(\counter_{p,k})=1$}\\
					\counter_{p,k} & \textrm{ otherwise, with $(p,k)\in\lab{\SS'}\setminus\lab{\SS}$}
	\end{cases}\\
	f'(e)&=f(e) \quad \textrm{ for all other $e\in [1,n]$}
	\end{align*}

%	\begin{itemize}
%	\item If $q'\in Q_A$, $f'(e_q)=\counter_{q'}$.
%	\item If $q'\in Q_W$, if $\lab{\SS}=\lab{SS'}$, let $\counter_{p,k}$ be the unique counter such that $\delta(\counter_{p,k})=1$. We let $f'(e_q)=\counter_{p,k}$. 
%	\end{itemize}
\end{itemize}

We have to prove \ref{cond:implem-definition:action-state}, \ref{cond:implem-definition:present-summary}\ and \ref{cond:implem-definition:absent-summary}, but
also that for all $\counter\in \Counters$, $|f'^{-1}(\counter)| = v'(\counter)$.

Let $p\in Q_A$ and $e\neq e_q$ such that $f'(e)=p=f(e)$ by construction. Then $C(e)=p$, and by definition of $C'$, $C'(e)=p$. If $e_q$ is such that $f'(e_q)=p$,
it means that $p=q'\in Q_A$. In that case, $C'(e_q)=p$. In any case, \ref{cond:implem-definition:action-state} is satisfied. Moreover, if $p\neq q'$, $v'(\counter_p)=
v(\counter_p)$ and ${f'}^{-1}(\counter_p) = f^{-1}(\counter_p)$, hence $|{f'}^{-1}(\counter_p)|=v'(\counter_p)$. If $p=q'$, then $v'(\counter_p)=
v(\counter_p)+1$ and ${f'}^{-1}(\counter_p) = f^{-1}(\counter_p)\cup\set{e_q}$ thus $|{f'}^{-1}(\counter_p)|=v'(\counter_p)$.

Let $(p,k)$ such that $S'=(\print', p,k)\in\SS'$ with $(p,k)\in \lab{\SS}$. Let $e\in {f'}^{-1}(\counter_{p,k})$. If $e\neq e_q$, then $e\in f^{-1}(\counter_{p,k})$. Let $S=(\print, p,k)\in \SS$.
Then $C(e)\in\print\cup \set{p}$. By construction of $C'$, either $C(e)=p=C'(e)$, or $C(e)\in\print$ and $C'(e)\in\print'$. If $e=e_q$, by definition of $f'$, $q'\in Q_W$, 
$\lab{\SS}=\lab{\SS'}$, and $\delta(\counter_{p,k})=1$, and, by definition of the VASS, this means that $S\abtransup{t, +q'} S'$, hence $q'\in\print'$. 
Hence, $C'(e)=q'\in\print'$. 
Let now $(p,k)$ such that $S'=(\print', p,k)\in\SS'$ with $(p,k)\nin \lab{\SS}$. Let $e\in {f'}^{-1}(\counter_{p,k})$ such that $e\neq e_q$. In that case, $e\in f^{-1}(\counter_{p,k})$.
 Since $(p,k)\nin \lab{\SS}$, by \ref{cond:implem-definition:absent-summary}, $C(e)=p\in Q_A$, and then $C'(e)=p\in\print'\cup\set{p}$. If $e=e_q$, $C'(e)=q'$, and
 by definition of the VASS, $S'=(\set{q'},p,k)$.  

Then \ref{cond:implem-definition:present-summary} is met. Moreover, if $f'(e_q)=\counter_{p,k}$ then $|{f'}^{-1}(\counter_{p,k})|=
|f^{-1}(\counter_{p,k})|+1$, and in that case, $v'(\counter_{p,k})=v(\counter_{p,k})+1$. Otherwise, $|{f'}^{-1}(\counter_{p,k})|=
|f^{-1}(\counter_{p,k})|=v(\counter_{p,k})=v'(\counter_{p,k})$. 

Finally, let $(p,k)\notin\lab{\SS'}$, and let $e\in {f'}^{-1}(\counter_{p,k})$. By definition of $f'$, $e\in f^{-1}(\counter_{p,k})$. 
If $(p,k)\notin\lab{\SS}$, then by \ref{cond:implem-definition:absent-summary} on $C$, $C(e)=p\in Q_A$, and then $C'(e)=p$. Otherwise,
let $S=(\print,p,k)\in\SS$, we must have $S\abtransup{t}\arrived$. By construction, $C'(e)=p$. Then \ref{cond:implem-definition:absent-summary} is satisfied by $C'$. 
Moreover, $v(\counter_{p,k})=v'(\counter_{p,k})$ and $|{f'}^{-1}(\counter_{p,k})|=|{f}^{-1}(\counter_{p,k})|$.

We have hence proved that $C'\in\implem{\lambda'}$.

\end{proof}

\section{Proofs of \cref{subsubsec:target-wo-completeness-part}}
\label{app-sec:target-wo-completeness-part}

\begin{proof}[Proof of \cref{lemma:target:completeness:local}]
	
	Let $C=\rho_i$ and let $\lambda_{i} = (\SS, v) \in \abconffrom{\rho}{i}$, then 
	%for all $q_a\in Q_A$, there exists an injective function $r_{q_a} : \nextactionindexset{\rho, i ,q_a} \rightarrow [1, |Q_W|+1]$ such that:
%\begin{itemize}
%	\item \ref{new-cond:repr:action-state} 
%	for all $q \in Q_A$, $v(\counter_q) = |C^{-1}(q)|$;
%	\item \ref{new-cond:repr:renaming}
%	 for all $q_a \in Q_A$, for all $j \in \nextactionindexset{\rho, i, q_a}$, $v(\counter_{q_a, r_{q_a}(j)}) = |E^{\rho, i}_{q_a, j}|$
%	%\item \label{new-cond:repr:null-counters} For 
%	 and for all $q_a \in Q_A$, for all $k \notin r_{q_a}(\nextactionindexset{\rho, i ,q_a}) $, $v(\counter_{q_a, k}) = 0$;
%	%\nas{reunion CondRepr2 et CondRepr3}
%	\item \ref{new-cond:summaries} $\SS= \set{(C(E^{\rho, i}_{q_a,j}), q_a, r_{q_a}(j)) \mid q_a \in Q_A, j \in \nextactionindexset{\rho,i, q_a}}$.
%	%\nas{ajouter notation $C(A)$ pour $A$ un ensemble}
%\end{itemize}
%\lug{ici}

\begin{itemize}
	\item \ref{new-cond:repr:action-state}\ 
	for all $q \in Q_A$, $v(\counter_q) = |C^{-1}(q)|$;
\end{itemize}
and for all $q_a\in Q_A$, there is an injective function $r_{q_a} : \nextactionindexset{\rho, i ,q_a} \rightarrow [1, |Q_W|+1]$ s.t.:
\begin{itemize}
	\item \ref{new-cond:summaries}\ $\SS= \set{(C(E^{\rho,i}_{q_a,j}), q_a, r_{q_a}(j)) \mid q_a \in Q_A, j \in \nextactionindexset{\rho,i, q_a}}$,and,
	\item \ref{new-cond:repr:renaming}
	for all $q_a\in Q_A$, we have  $v(\counter_{q_a, r_{q_a}(j)}) = |E^{\rho,i}_{q_a, j}|$ for all $j \in \nextactionindexset{\rho, i, q_a}$ 
	%\item \label{new-cond:repr:null-counters} For 
	and  $v(\counter_{q_a, k}) = 0$ for all $k \notin r_{q_a}(\nextactionindexset{\rho, i ,q_a})$.      
\end{itemize}

%	for all $q_a\in Q_A$, there exists an injective function $r_{q_a} : \nextactionindexset{\rho, i ,q_a} \rightarrow [1, |Q_W|+1]$ such that:
%\begin{enumerate}[{\color{violet}\textbf{CondRepr}}1,leftmargin=7.5em]
%	\item \label{new-cond:repr:action-state} 
%	for all $q \in Q_A$, $v(\counter_q) = |C^{-1}(q)|$;
%	\item \label{new-cond:repr:renaming}
%	 for all $q_a \in Q_A$, for all $j \in \nextactionindexset{\rho, i, q_a}$, let $E^{\rho, i}_{q_a, j} = \set{\aproc\in[1, \nbagents{\rho}] \mid C(\aproc) \in Q_W\textrm{ and }
%	  \nextaction{\rho, i, \aproc} = (q_a, j)}$. 
%	 Then $v(\counter_{q_a, r_{q_a}(j)}) = |E^{\rho, i}_{q_a, j}|$;
%	\item \label{new-cond:repr:null-counters} For all $q_a \in Q_A$, for all $k \in [1, |Q_W|+1]$ such that $r_{q_a}^{-1}(k) = \emptyset$, $v(\counter_{q_a, k}) = 0$;
%	\nas{reunion CondRepr2 et CondRepr3?}
%	\item \label{new-cond:summaries}$\SS= \set{(C(E^{\rho, i}_{q_a,j}), q_a, r_{q_a}(j)) \mid q_a \in Q_A, j \in \nextactionindexset{\rho,i, q_a}}$.
%	\end{enumerate}
%	
	Let $t=(q,!!m,q')$ and $e_q\in[1, \nbagents{\rho}]$ such that $C\transup{e_q,t} C_{i+1}$.
	%\nas{change all $e$ to $e_q$ and check $C$ vs $C_i$!} 
	We start by building an $S$-configuration $\lambda'=(\SS',v')$
	such that $\lambda_i\abconftransup{t}\lambda'$. We will then show that an appropriate sequence of transitions from $\Delta_e$ allows to go from $\lambda'$ to
	a configuration $\lambda_{i+1}\in\abconffrom{\rho}{i+1}$.
	
	We state first some remarks about $E^{\rho, i}_{q_a,j}$, for $q_a\in Q_A$ and $j\in\nextactionindexset{\rho,i,q_a}$. Observe that, since $E^{\rho, i}_{q_a,j}=\set{e\in [1, \nbagents{\rho}]\mid C(e)\subseteq Q_W\textrm{ and }\nextaction{\rho,i,e}=(q_a,j)}$, if $C_{i+1}(e)\in Q_A$ for
	some $e\in E^{\rho, i}_{q_a,j}$, then it means that $j=i+1$ and $C_{i+1}(e)=q_a$ for all $e\in E^{\rho, i}_{q_a,j}$. Hence $C_{i+1}(E^{\rho, i}_{q_a,j})\subseteq Q_W$ or $C_{i+1}(E^{\rho, i}_{q_a,j})=\set{q_a}$. 
	
	Observe also that if $j \in \nextactionindexset{\rho, i, q_a}$:
	\begin{displaymath}
E^{\rho, i+1}_{q_a,j}=\begin{cases}\emptyset & \textrm{ if $j=i+1$}\\
E^{\rho, i}_{q_a,j} \cup\set{e_q}& \textrm{if $q'\in Q_W$ and $\nextaction{\rho,i+1, e_q}=(q_a,j)$}\\
E^{\rho, i}_{q_a,j} & \textrm{otherwise.}
\end{cases}
\end{displaymath}
Moreover, if $q'\in Q_W$, and $\nextaction{\rho,i,e_q}=(q_a,j)$ with $j\notin\nextactionindexset{\rho,i,q_a}$, then $E^{\rho, i}_{q_a,j}=\emptyset$ hence $E^{\rho, i+1}_{q_a,j}=\set{e_q}$.

	\paragraph*{Definition of $\lambda'$.} For all $\summary=(\print, q_a,k)\in\SS$, we let $j=r_{q_a}^{-1}(k)$ and we define 
	\begin{displaymath}
	\nextsummary{\summary}=\begin{cases} (C_{i+1}(E^{\rho, i}_{q_a,j}), q_a, k) & \textrm{ if }C_{i+1}(E^{\rho, i}_{q_a,j})\subseteq Q_W\\
									\arrived & \textrm{ otherwise.}
									\end{cases}
	\end{displaymath}
	
	Recall that $t = (q, !!m, q')$.
	$\nextsummary{\summary}$ represents the set of states reached by the states in the print of $\summary$ after the sending of the message $m$ in $\rho$. 
	We then let $\SS_0=\set{\nextsummary{\summary}\mid\summary\in\SS}\setminus\set{\arrived}$. 
	
	We first show that $\SS_0$ is coherent:
	%\nas{si on utilise
	%knowledge, alors lien a mettre}:
	let $\summary_1=(\print_1,q_a,k_1), \summary_2=(\print_2,q_a,k_2)\in\SS$ such that $\nextsummary{\summary_1}=(\print'_1, q_a,k_1)$ and $\nextsummary{\summary_2}=(\print'_2, q_a, k_2)$ are two distinct summaries in $\SS_0$. Since $\SS$ is coherent, we know that $k_1\neq k_2$. 
	Assume that there is some $p\in\print'_1\cap\print'_2$. It means that there exist $e_1,e_2\in[1, \nbagents{\rho}]$ such that $C_{i+1}(e_1)=p=C_{i+1}(e_2)$, and
	$e_1\in E^{\rho, i}_{q_a, r_{q_a}^{-1}(k_1)}$ and $e_2\in E^{\rho, i}_{q_a, r_{q_a}^{-1}(k_2)}$.  Hence by definition, $\nextaction{\rho, i, e_1} = (q_a, r^{-1}_{q_a}(k_1))
	=\nextaction{\rho, i+1, e_1}$
	and $\nextaction{\rho, i, e_2} = (q_a, r^{-1}_{q_a}(k_2))=\nextaction{\rho, i+1, e_2}$. Since $\rho$ is well-formed, we have
	%\nas{knowledge}, we have 
	that $r^{-1}_{q_a}(k_1)=r^{-1}_{q_a}(k_2)$,
	and hence, $k_1= k_2$, which is a contradiction. Then for any two distinct summaries $\nextsummary{\summary_1}=(\print'_1, q_a,k_1)$ and $\nextsummary{\summary_2}=(\print'_2, q_a, k_2)$, then $k_1\neq k_2$ and $\print'_1\cap\print'_2=\emptyset$. 
	
	Now that we have defined $\SS_0$, which captures the behavior of the processes that were in waiting states when the transition $t$ has been taken, we define
	$\SS'$ that incorporates the process $e_q$ that arrives in state $q'$, along with the new valuation $v'$. In particular, if $q'\in Q_W$, then 
	$e_q$ reaches a zone of waiting states and must enter one of the summaries. Observe first that $C(e_q)=q$, then $v(\counter_q)=|C^{-1}(q)|>0$, then we let $v'(\counter_q)=v(\counter_q)-1$. 
	\begin{itemize}
	\item If $q'\in Q_A$, $\SS'=\SS_0$ and $v'(\counter_{q'})=v(\counter_{q'})+1$. For all other counter $\counter$, $v'(\counter)=v(\counter)$. 
	\item If $q'\in Q_W$, let $\nextaction{\rho, i+1, e_q}=(p,j)$.
		\begin{enumerate}
		\item If $j\in\nextactionindexset{\rho, i, p}$, then $e_q$ will join an existing summary $\summary=(\print, p, r_p(j))\in \SS$. Since
		$\nextaction{\rho,i+1,e_q}=(p,j)$ and $q'\in Q_W$, $j\neq i+1$ and $C_{i+1}(E^{\rho, i}_{q_a,j})\subseteq Q_W$. Hence, $\nextsummary{\summary}\neq\arrived$. We then let 
		$\SS'=(\SS_0\setminus\set{\nextsummary{\summary}})\cup \set{(\printof{\nextsummary{\summary}}\cup\set{q'}, p, r_p(j))}$. Then $v'(\counter_{p,r_p(j)})=
		v(\counter_{p,r_p(j)})+1$ and $v'(\counter)=v(\counter)$ for all other counter $\counter$. 
		
		\item \label{it:valuation-new-summary}If $j\notin\nextactionindexset{\rho, i, p}$, then $\aproc_q$ will enter a new summary. Observe first that $\nextactionindexset{\rho, i+1, p}=(\nextactionindexset{\rho,i,p}\cup\set{j})\setminus\set{i+1}$.  So $|\nextactionindexset{\rho,i,p}|\leq|\nextactionindexset{\rho,i+1,p}|\leq|Q_W|< 1+|Q_W|$, and
		there exists $k \in [1, |Q_W|+1]$ such that $k \notin r_p(\nextactionindexset{\rho,i,p})$. This means that $(p,k)$ is a label that is not used in the summaries of $\SS_0$
		and we let $\SS'=\SS_0\cup\set{(\set{q'},p,k)}$. Now, $v'(\counter_{p,k})=v(\counter_{p,k})+1$, and for all other counter $\counter$, $v'(\counter)=v(\counter)$. 
		\end{enumerate} 
	\end{itemize}
	
	\paragraph*{Proof that $\lambda'=(\SS',v')$ is an $S$-configuration.}
	To prove that $\lambda'$ is an $S$-configuration, we need to show that $\SS'$ is coherent. The definition of $\SS'$ depends on the belonging of $q'$ to $Q_A$ or $Q_W$.
	\begin{itemize}
	\item If $q'\in Q_A$, then $\SS'=\SS_0$, and we have proved above that $\SS_0$ is coherent. 
	\item if $q'\in Q_W$, let $\summary_1=(\print_1, q_a, k_1)$ and $\summary_2=(\print_2, q_a, k_2)\in \SS_1$, for some $q_a\in Q_A$. If $\summary_1,\summary_2\in \SS_0$, by coherence of $\SS_0$, we have that 
	$k_1\neq k_2$ and $\print_1\cap\print_2=\emptyset$. Otherwise, assume that $\summary_2\in \SS'\setminus\SS_0$. Depending
	on whether $\aproc_q$ has joined an existing summary or a new one, $\summary_2=(\print\cup\set{q'}, p,k_2)$ with $\print$ empty or not. Then $\summary_1=(\print_1,p,k_1)$ and
	if $\print\neq\emptyset$, $(\print,p,k_2)\in \SS_0$. By coherence of $\SS_0$, $k_1\neq k_2$ and $\print_1\cap\print=\emptyset$. Moreover, if $\summary_2=(\set{q'}, p, k_2)$, by construction, $k_2\neq k_1$. Assume now that there exists $p'\in\print_1\cap\print_2$. Then $p'=q'$, and there exists $e_1\in
	[1, \nbagents{\rho}]$ such that $C_{i+1}(e_1)=q'$ and $\nextaction{\rho,i,e_1}=(p, r^{-1}_{p}(k_1))$. But $\nextaction{\rho, i, e_q}=(p, j)$. Since $\rho$ is well-formed,
	$j=r^{-1}_p(k_1)$, which is impossible because $k_2\neq k_1$. Hence $\print_1\cap\print_2=\emptyset$, and $\SS_1$ is coherent. 
	\end{itemize}
	\paragraph*{Proof that $\lambda_i \abconftransup{t} \lambda'$.}
	We first show  that for all $\summary\in \SS$, $S\abtransup{t}\nextsummary{\summary}$. Let $\summary=(\print, q_a, k)\in\SS$ with $j'=r^{-1}_{q_a}(k)$. By 
	construction, $\print=C(E^{\rho, i}_{q_a,j'})$. Assume now that there exist $e_1,e_2\in E^{\rho, i}_{q_a,j'}$ such that $C(e_1)=C(e_2)$. Since $\rho$
	is well-defined, $C_{i+1}(e_1)=C_{i+1}(e_2)$. Then, for all $s\in \print$, there exists a unique
	$\nextsummary{s}\in Q$ such that $C_{i+1}(e)=\nextsummary{s}$ for all $e\in C^{-1}_i(s)$. Moreover,
	$C_{i+1}(E^{\rho, i}_{q_a,j'})=\set{\nextsummary{s}\mid s\in \print}$. We define a configuration $C'$ as follows. Let $e\in [1, ||\conf{\print}||]$ such that $\conf{\print}(e)=p_e\in\print$. 
	Then $C'(e)=\nextsummary{p_e}$. Hence $\setof{C'}=\set{\nextsummary{s}\mid s\in \print}= C_{i+1}(E^{\rho, i}_{q_a, j'})$.  We show now that $\conf{q}\oplus \conf{1, \print} \transup{t} \conf{q'}\oplus C'$. Let $e\in [1,||\conf{\print}||]$ such that $\conf{\print}(e)=p_e$, then there exists $e'\in E^{\rho, i}_{q_a, j'}$ such that 
	$C(e')=p_e$ and $\nextsummary{p_e}=C_{i+1}(e')$. Since $C\transup{t} C_{i+1}$, then either $(p_e, ?m, \nextsummary{p_e})\in T$ or $m\notin R(p_e)$
	and $p_e=\nextsummary{p_e}$, and $\conf{q}\oplus \conf{\print} \transup{1, t} \conf{q'}\oplus C'$. We have already established that 
	$\setof{C'}=C_{i+1}(E^{\rho, i}_{q_a, j'})\subseteq Q_W$ or $\setof{C'}=C_{i+1}(E^{\rho, i}_{q_a, j'})=\set{q_a}$. In the latter case, $\nextsummary{\summary}=\arrived$. 
	In the former case, $\nextsummary{\summary}= (C_{i+1}(E_{q_a, j'}^i), q_a, k)$ and $\summary\abtransup{t}\nextsummary{\summary}$. In any case,
	we have that $\summary\abtransup{t}\nextsummary{\summary}$. 
	
	Now we show that there exists $\delta$ such that $\lambda_i\abconftransup{\delta}\lambda'$ and $(\SS,\delta,\SS')\in \Delta_t$. 
	\begin{itemize}
	\item If $q'\in Q_A$, we must have 
	$\delta(\counter_q)=-1$, $\delta(\counter_{q'})=1$ and for all other counter $\counter$, $\delta(\counter)=0$. Observe that we have $v'(\counter_q)=v(\counter_q)-1$,
	$v'(\counter_{q'})=v(\counter_{q'})+1$, and $v'(\counter)=v(\counter)$ for all other counter $\counter$. Moreover, $\SS'=\SS_0=\set{\nextsummary{\summary}\mid \summary\in \SS}\setminus\set{\arrived}.$ Since $\summary\abtransup{t}\nextsummary{\summary}$ for all $\summary\in \SS$, we have $(\SS,\delta, \SS')\in
	\Delta_t$ and $\lambda_i\abconftransup{\delta}\lambda'$. 
	\item If $q'\in Q_W$ and there exists some $\summary=(\print, p, r_p(j))$ such that $\SS'=(\SS_0\setminus\set{\nextsummary{\summary}})\cup \set{(\printof{\nextsummary{\summary}}\cup\set{q'}, p, r_p(j))}$.
	Then $v'(\counter_q)=v(\counter_q)-1$, $v'(\counter_{p,r_p(j)})=
		v(\counter_{p,r_p(j)})+1$ and $v'(\counter)=v(\counter)$ for all other counter $\counter$. By definition, $\summary\abtransup{t, +q'} (\printof{\nextsummary{\summary}}\cup\set{q'}, p, r_p(j))$, because $\summary\abtransup{t}\nextsummary{\summary}$. Then, with $\delta(\counter_q)=-1$,
		$\delta(\counter_{p,r_p(j)})=1$ and $\delta(\counter)=0$ for all other counter, we have $(\SS,\delta, \SS')\in \Delta_t$, and $\lambda_i\abconftransup{\delta}\lambda'$.
	\item If $q' \in Q_W$ and $\SS'=\SS_0\cup\set{(\set{q'},p,k)}$ for some $k\in [1,|Q_W|+1]$. Then $(p,k)\notin \lab{\SS}$ by construction. Moreover,
	$v'(\counter_q)=v(\counter_q)-1$, and $v'(\counter_{p,k})=1$. For all other counter $\counter$, $v'(\counter)=v(\counter)$. For $\delta$ such that
	$\delta(\counter_q) = - 1$, $\delta(\counter_{p,k})=1$ and $\delta(\counter)=0$ for all other $\counter$, $(\SS,\delta,\SS')\in \Delta_t$ and 
	$\lambda_i\abconftransup{\delta}\lambda'$. 
	\end{itemize}
	
	Then $\lambda_i\abconftransup{t}\lambda'$. 
	
	\paragraph*{Definition of $\lambda_{i+1}$.}
	First observe that for all $q_a\in Q_A$, there is at most one $\summary=(\print, q_a,k)\in \SS$ such that $\nextsummary{\summary}=\arrived$. Indeed, 
	suppose there exist $\summary_1=(\print_1,q_a,k_1), \summary_2=(\print_2,q_a,k_2)\in \SS$ such that $\nextsummary{\summary_1}=\nextsummary{\summary_2}=\arrived$. By definition, $\print_1=C(E_{q_a,r^{-1}_{q_a}(k_1)}^{\rho, i})$ and $\print_2=C(E_{q_a,r^{-1}_{q_a}(k_2)}^{\rho, i})$. Let $j_1=r^{-1}_{q_a}(k_1)$
	and $j_2=r^{-1}_{q_a}(k_2)$. Then, since $C_{i+1}(E^{\rho, i}_{q_a,j_i})\notin Q_W$ for $i\in \set{1,2}$, necessarily $C_{i+1}(E^{\rho, i}_{q_a,j_1})=C_{i+1}(E^{\rho, i}_{q_a,j_2})
	=\set{q_a}$. %\nas{mettre en exergue qqpart que soit $C_{i+1}(E^{\rho, i}_{q_a,j_1})\subseteq Q_W$ soit $C_{i+1}(E^{\rho, i}_{q_a,j_1})=\set{q_a}$ et y faire reference ici}
	This implies that $j_1=j_2=i+1$ and then that $k_1=k_2$, which contradicts the coherence of $\SS$. This means that, there is at most one summary
	for each $q_a\in Q_A$ that disappears while going from $\SS$ to $\SS'$. This corresponds to a set of processes leaving a waiting zone and reaching the state $q_a$ together. We
	capture this new configuration of the protocol by transferring the value of the counter $\counter_{(q_a,r_{q_a}(i+1))}$ into the counter $\counter_{q_a}$. 
	
	So we define $v_{i+1}$ as follows. For all $q_a\in Q_A$, 
	\begin{itemize}
	\item if $i+1\in \nextactionindexset{\rho, i, q_a}$, let $k=r_{q_a}(i+1)$ and $v_{i+1}(\counter_{q_a})=v'(\counter_{q_a})+v'(\counter_{q_a,k})$ and
	$v_{i+1}(\counter_{q_a,k})=0$.
	\item Otherwise, $v_{i+1}(\counter_{q_a})=v'(\counter_{q_a})$.
	\end{itemize}
	For all other counter $\counter$, $v_{i+1}(\counter)=v'(\counter)$.
	
	We let $\lambda_{i+1}=(\SS',v_{i+1})$. 
	\paragraph*{Proof that $\lambda_{i+1}\in\abconffrom{\rho}{i+1}$.}
	We show that \ref{new-cond:repr:action-state} holds. Let $q_a\in Q_A$. Observe first that, for all $e\in[1, \nbagents{\rho}]$, $(C(e), ?m, q_a)\in T$ if and only if
	$e\in E_{q_a,i+1}^{\rho, i}$. Let $e_q\in [1, \nbagents{\rho}]$ such that $C\transup{e_q,t}C_{i+1}$. Then, $C(e_q)=q$, $C_{i+1}(e_q)=q'$. 
	For all $e\in[1, \nbagents{\rho}]$, $C_{i+1}(e)=q_a$ if and only if 
%	($C(e)=q_a$ and $q\neq q_a$, or $C(e)=q_a$ and $e\neq e_q$, or $(C(e), ?m, q_a)\in T$, or $q_a=q'$ and $e=e_q$) if and only if
	($C(e)=q_a$ and $q\neq q_a$, or $C(e)=q_a$ and $e\neq e_q$, or $(C(e), ?m, q_a)\in T$ and $\nextaction{\rho,i,\aproc} = (q_a, i+1)$, or $q_a=q'$ and $e=e_q$).	
	From this we can deduce that

\begin{displaymath}
C_{i+1}^{-1}(q_a) = \begin{cases}
C^{-1}(q_a) \cup E_{q_a, i+1}^{\rho, i}& \textrm{if $q_a\neq q$ and $q_a\neq q'$}\\
C^{-1}(q_a)\cup\set{e_q}\cup E_{q_a, i+1}^{\rho, i} & \textrm{if $q_a=q'$}\\
(C^{-1}(q_a)\setminus\set{e_q})\cup E_{q_a, i+1}^{\rho, i} & \textrm{if $q_a=q$}.
\end{cases}
\end{displaymath}

Since $C^{-1}(q_a)$ and $E_{q_a, i+1}^{\rho, i}$ are disjoint, and if $e_q\in C^{-1}(q)\setminus C^{-1}(q')$ (which is always true as $q\neq q'$ by hypothesis), we have that 
\begin{displaymath}
|C_{i+1}^{-1}(q_a)| = \begin{cases}
|C^{-1}(q_a)| + |E_{q_a, i+1}^{\rho, i}| & \textrm{if $q_a\neq q$ and $q_a\neq q'$}\\
|C^{-1}(q_a)+1+ |E_{q_a, i+1}^{\rho, i}| & \textrm{if $q_a=q'$}\\
|C^{-1}(q_a)|-1+|E_{q_a, i+1}^{\rho, i}| & \textrm{if $q_a=q$}.
\end{cases}
\end{displaymath}

Since $\lambda=(\SS,v)\in\abconffrom{C}{i}$, $v(\counter_{q_a})=|C^{-1}(q_a)|$ and, by definition of $v'$, we have that 
\begin{displaymath}
|C_{i+1}^{-1}(q_a)| =
v'(\counter_{q_a}) + |E_{q_a, i+1}^{\rho, i}| 
\end{displaymath}

Now, if $i+1\notin \nextactionindexset{\rho, i, q_a}$, it means that $E_{q_a, i+1}^{\rho, i}=\emptyset$. In that case, by definition of $v_{i+1}$, $v_{i+1}(\counter_{q_a})=v'(\counter_{q_a})=|C_{i+1}^{-1}(q_a)|$.

Otherwise, $v_{i+1}(\counter_{q_a})=v'(\counter_{q_a})+v'(\counter_{(q_a, r_{q_a}(i+1))})$. The only case where $v'(\counter_{(q_a, r_{q_a}(i+1))})\neq v(\counter_{(q_a, r_{q_a}(i+1))})$ is if $q'\in Q_W$ and $\nextaction{\rho, i+1, e_q}=(q_a,i+1)$ which is impossible. So $v'(\counter_{q_a, r_{q_a}(i+1)})=v(\counter_{q_a, r_{q_a}(i+1)})$ and, since $\lambda \in\abconffrom{C}{i}$, $v(\counter_{q_a, r_{q_a}(i+1)})=|E_{q_a, i+1}^{\rho, i}|$. Hence, 
$v_{i+1}(\counter_{q_a})=v'(\counter_{q_a})+|E_{q_a, i+1}^{\rho, i}|=|C_{i+1}^{-1}(q_a)|$, proving \ref{new-cond:repr:action-state}.

We now show that \ref{new-cond:repr:renaming} holds. Let $q_a\in Q_A$. We define an injective function $r'_{q_a}: \nextactionindexset{\rho, i+1, q_a}\rightarrow [1, |Q_W|+1]$ as follows.

Recall that $$\nextactionindexset{\rho, i+1, q_a}=\begin{cases}(\nextactionindexset{\rho, i, q_a}\setminus\set{i+1})\cup&\\
 \set{\nextactionindex{\rho, i+1, e_q}} & 
\textrm{if $q'\in Q_W$ and $\nextactionstate{\rho, i+1, e_q}=q_a$}\\
\nextactionindexset{\rho, i, q_a}\setminus\set{i+1} & \textrm{otherwise.}
\end{cases}$$

The new injective function keeps the labelling of the summaries as in the previous configuration, and if needed, associates the labelling of the newly created 
summary to the index of the next moment $e_q$ will reach an action state. 
\begin{displaymath}
r'_{q_a}(j)=\begin{cases} r_{q_a}(j) & \textrm{if $j\in \nextactionindexset{\rho, i, q_a}\setminus\set{i+1}$}\\
k & \textrm{otherwise, for the unique $k$ such that $(\set{q'}, q_a, k)\in\SS'$ and $(q_a, k)\nin \lab{\SS}$}
\end{cases}
\end{displaymath}
	
	\begin{itemize}
	\item For $j\in\nextactionindexset{\rho, i, q_a}\setminus\set{i+1}$, $v_{i+1}(\counter_{q_a, r'_{q_a}(j)})=v_{i+1}(\counter_{q_a, r_{q_a}(j)})=v'(\counter_{q_a, r_{q_a}(j)})$. By definition of $v'$, 
	
	\begin{displaymath}
v'(\counter_{q_a, r_{q_a}(j)})=\begin{cases} v(\counter_{q_a, r_{q_a}(j)})+1 & \textrm{if $q'\in Q_W$ and $\nextaction{\rho, i+1, e_q}=(q_a,j)$}\\
v(\counter_{q_a, r_{q_a}(j)}) & \textrm{otherwise}.
\end{cases}
\end{displaymath}

We know that $v(\counter_{q_a, r_{q_a}(j)})=|E^{\rho, i}_{q_a,j}|$. Moreover, 
\begin{displaymath}
E^{\rho, i+1}_{q_a,j}=\begin{cases}E^{\rho, i}_{q_a,j} \cup\set{e_q}& \textrm{if $q'\in Q_W$ and $\nextaction{\rho,i+1, e_q}=(q_a,j)$}\\
E^{\rho, i}_{q_a,j} & \textrm{otherwise.}
\end{cases}
\end{displaymath}
Hence, $v_{i+1}(\counter_{q_a, r'_{q_a}(j)})=v'(\counter_{q_a, r_{q_a}(j)})=|E^{\rho, i+1}_{q_a,j}|$. 
\item If $\nextaction{\rho, i+1, e_q} = (q_a,j)$ and $j\notin \nextactionindexset{\rho, i, q_a}$, by definition $v_{i+1}(\counter_{q_a, r'_{q_a}(j)})=v'(\counter_{q_a, r'_{q_a}(j)})
=v'(\counter_{q_a,k})$. This corresponds to the case~\ref{it:valuation-new-summary} of the definition of $v'$, where $v'(\counter_{q_a,k})=v(\counter_{q_a,k})+1$. 
By definition of $k$, $k\notin r_p(\nextactionindexset{\rho,i,q_a})$, then, since $\lambda=(\SS,v)\in \abconffrom{C}{i}$, $v(\counter_{q_a,k})=0$ and
$v_{i+1}(\counter_{q_a, r'_{q_a}(j)})=1$. Moreover, $E^{\rho, i}_{q_a,j}=\{e\in[1, \nbagents{\rho}]\mid C(e)\in Q_W\textrm{ and }\nextaction{\rho, i, e}= (q_a,j)\}$; since
$j\notin \nextactionindexset{\rho, i, q_a}$,
$E^{\rho, i}_{q_a,j}=\emptyset$. 
Hence, $E_{q_a, j}^{\rho, i+1}=E^{\rho, i}_{q_a,j} \cup\set{e_q}=\set{e_q}$ and $v_{i+1}(\counter_{q_a, r'_{q_a}(j)})=1=|E^{\rho, i+1}_{q_a,j}|$.
\end{itemize}

Let $x\notin {r'}_{q_a}(\nextactionindexset{\rho, i+1, q_a})$. Then, either $x\notin\nextactionindexset{\rho, i, q_a}$, or $x=r_{q_a}(i+1)$. 
In the first case, $v(\counter_{q_a,x})=0$, and $v_{i+1}(\counter_{q_a,x})=v'(\counter_{q_a,x})=v(\counter_{q_a,x})=0$. In the second case, $v_{i+1}(\counter_{q_a,x})=0$ by construction. 

Finally we show that \ref{new-cond:summaries} holds, i.e. that $\SS'=\set{(C_{i+1}(E^{\rho, i+1}_{q_a,j}),q_a,r'_{q_a}(j))\mid q_a\in Q_A, j\in\nextactionindexset{\rho,i+1,q_a}}$.

\begin{itemize}
	\item First, let $\summary\in \SS'$. By construction, either $\summary\in\SS_0$, or there exists $\summary_0=(\print, q_a,r_{q_a}(j))\in\SS_0$ such that $\summary=(\print\cup\set{q'}, q_a,r_{q_a}(j))$,
	or $\summary=(\set{q'}, q_a, k)$.  
	\begin{itemize}
		\item If $\summary\in\SS_0$ and $\nextaction{\rho,i,q_a,e_q} \neq (q_a, j)$ or $q'\in Q_A$, then $\summary=\nextsummary{\summary_0}$ for some $\summary_0\in \SS$. Since $(\SS,v)\in\abconffrom{C}{i}$, 
		$\summary_0=(C(E^{\rho, i}_{q_a,j}), q_a, r_{q_a}(j))$ for some $q_a\in Q_A$ and $j\in\nextactionindexset{\rho,i,q_a}$.  Then $\summary=(C_{i+1}(E^{\rho, i}_{q_a,j}),q_a,r_{q_a}(j))$. 
%		By definition of $\SS'$,
%		it means that $\nextaction{\rho,i+1,e_q}\neq(q_a,j)$, and
		Then $E^{\rho, i+1}_{q_a,j}=E^{\rho, i}_{q_a,j}$. Moreover, $\nextsummary{\summary_0}\neq\arrived$ hence $j\neq i+1$, then $r'_{q_a}(j)=r_{q_a}(j)$. 
		Then $\summary=(C_{i+1}(E^{\rho, i+1}_{q_a,j}), q_a, r'_{q_a}(j))$. 
		\item If there exists $\summary_0=(\print, q_a,r_{q_a}(j))\in\SS$ such that $\summary=(\print\cup\set{q'}, q_a,r_{q_a}(j))$ and
		$q'\in Q_W$ and $\nextaction{\rho, i+1, e_q}=(q_a,j)$, then, in that case, $E^{\rho, i+1}_{q_a,j}=E^{\rho, i}_{q_a,j} \cup\set{e_q}$ and $j\neq i+1$ so $r'_{q_a}(j)=r_{q_a}(j)$. Then $\summary=(C_{i+1}(E^{\rho, i+1}_{q_a,j}), q_a, r'_{q_a}(j))$. 
		\item Finally, if $\summary=(\set{q'}, q_a, k)$, and $q'\in Q_W$, with $\nextaction{\rho,i+1, e_q}=(q_a,j)$ and $j\notin\nextactionindexset{\rho,i,q_a}$, then $E^{\rho, i+1}_{q_a,j}=\set{e_q}$ and 
		$r'_{q_a}(j)=k$. Then, 
		$\summary=(C_{i+1}(E^{\rho, i+1}_{q_a,j}),q_a,r'_{q_a}(j))$. 
	\end{itemize}
	\item Now, let $\summary=(C_{i+1}(E^{\rho, i+1}_{q_a,j}), q_a, r'_{q_a}(j))$ for some $q_a\in Q_A$, $j\in\nextactionindexset{\rho,i+1,q_a}$. 
		\begin{itemize}
		\item If $j\in\nextactionindexset{\rho, i, q_a}$, then $j\neq i+1$ and $r'_{q_a}(j)=r_{q_a}(j)$.
	Let $\summary_0=(C(E^{\rho, i}_{q_a,j}), q_a, r_{q_a}(j))\in \SS$. Such a summary exists because $\lambda=(\SS,v)$ is in $\abconffrom{C}{i}$. If $q'\in Q_W$ and $\nextaction{\rho,i+1, e_q}=(q_a,j)$,
	then, $E^{\rho, i+1}_{q_a,j}=E^{\rho, i}_{q_a,j}\cup\set{e_q}$. Moreover, by construction of $\SS'$, $(C_{i+1}(E^{\rho, i}_{q_a,j})\cup\set{q'}, q_a, r_{q_a}(j))=(C_{i+1}(E^{\rho, i}_{q_a,j})\cup C_{i+1}(e_q), q_a, r_{q_a}(j))=(C_{i+1}(E^{\rho, i+1}_{q_a,j}), q_a, r'_{q_a}(j))\in \SS'$. Otherwise, $E^{\rho, i+1}_{q_a,j}=E^{\rho, i}_{q_a,j}$ and, by construction of $\SS'$, $\nextsummary{\summary_0}=(C_{i+1}(E^{\rho, i}_{q_a,j}), q_a, r_{q_a}(j))=(C_{i+1}(E^{\rho, i+1}_{q_a,j}), q_a, r'_{q_a}(j))\in \SS'$.  
		\item If $j\notin\nextactionindexset{\rho,i,q_a}$, then $r'_{q_a}(j)=k$ and $(\set{q'}, q_a,k)\in\SS'$. Moreover, since $j\in\nextactionindexset{\rho,i+1, q_a}\setminus\nextactionindexset{\rho,i,q_a}$, it means
		that $q'\in Q_W$ and $\nextaction{\rho, i+1, e_q}=(q_a,j)$. Then $E^{\rho, i+1}_{q_a,j}=\set{e_q}$ and $\summary=(C_{i+1}(E^{\rho, i+1}_{q_a,j}, q_a, r'_{q_a}(j))=(\set{q'}, q_a,k)\in \SS'$. 
		\end{itemize}
	\end{itemize}

	\paragraph*{Proof that $\lambda'\abconftrans^*\lambda_{i+1}$.}
	The only difference between $\lambda'$ and $\lambda_{i+1}$ is in the valuation. We show that $\lambda'\abconftrans^*\lambda_{i+1}$ by successive transitions from $\Delta_e$ in order to transfer the 
	values of counters associated to summaries that have reached $\arrived$ to the counter associated to the corresponding action state. 
	
	Let $\set{q_1,\dots, q_n}\subseteq Q_A$ such that, for all $1\leq \ell\leq n$, $i+1\in \nextactionindexset{\rho,i,q_\ell}$. It is the set of action states that are reached by processes in waiting zones in $C_{i+1}$.
	We build the sequence of valuations $v_0,v_1,\dots, v_n$ with $v_0=v'$ and $v_n=v_{i+1}$ such that $(\SS',v_\ell)\abconftrans^*(\SS', v_{\ell+1})$, for all $1\leq \ell < n$. 
	 For all $1\leq \ell\leq n$, let $k_\ell=r_{q_\ell}(i+1)$.
	We prove by induction that, for all $1\leq \ell \leq n$, $v_\ell(\counter_{q_{\ell'},k_{\ell'}} )=v_{i+1}(\counter_{q_{\ell'},k_{\ell'}})$, and $v_\ell(\counter_{q_{\ell'}})=v_{i+1}(\counter_{q_{\ell'}})$ 
	for $1\leq \ell'\leq \ell$, and
	for all other counter $v_\ell(\counter)=v'(\counter)$. 
	
	Observe that for all $1\leq\ell \leq n$, $k_\ell\notin r'_{q_\ell}(\nextactionindexset{\rho,i+1,q_\ell})$.  Indeed, by construction, if $\nextaction{\rho,i+1,e_q}=(q_\ell,j)$, then $r'_{q_\ell}(j)\neq r_{q_\ell}(i+1)$. 
	
	Let $v(\counter_{q_1, k_1})=x$. Then $(\SS', v)(\abconftransup{\delta})^x (\SS', v_1)$ with $\delta(\counter_{q_1})=+1$, $\delta(\counter_{q_1,k_1})=-1$ and $\delta(\counter)=0$ for all
	other $\counter$. Since $k_1\notin r'_{q_\ell}(\nextactionindexset{\rho,i+1,q_1})$, $(q_1,k_1)\notin\lab{\SS'}$ (due to the fact that $(\SS',v_{i+1})$ is in $\abconffrom{\rho}{i+1}$), and $(\SS',\delta, \SS')\in \Delta_e$. Moreover, $v_1(\counter_{q_1})=v'(\counter_{q_1})+v'(\counter_{q_1,k_1})=v_{i+1}(\counter_{q_1})$ and $v_1(\counter_{q_1,k_1})=0=v_{i+1}(\counter_{q_1,k_1})$. 
	
	Let $\ell\geq 1$ and assume that the property is true for $1\leq\ell' < \ell$. We can show in a similar way that for $\delta(\counter_{q_\ell})=+1$, $\delta(\counter_{q_\ell,k_\ell})=-1$ and $\delta(\counter)=0$
	for all other counter $\counter$, we have $(\SS',\delta,\SS')\in\Delta_e$ and $(\SS',v_{\ell-1})(\abconftransup{\delta})^{x'} (\SS', v_\ell)$, with $x'=v(\counter_{q_\ell,k_\ell})$. Hence, $v_{\ell}(\counter_{q_\ell})=
	v_{\ell-1}(\counter_{q_\ell})+v_{\ell-1}(\counter_{q_\ell,k_\ell})$. By induction hypothesis, $v_{\ell}(\counter_{q_\ell})=v'(\counter_{q_\ell})+v'(\counter_{q_\ell,k_\ell})=v_{i+1}(\counter_{q_\ell})$.
	Moreover $v_{\ell}(\counter_{q_\ell,k_\ell})=0=v_{i+1}(\counter_{q_\ell,k_\ell})$.  For all other counter $\counter$, $v_\ell(\counter)=v_{\ell-1}(\counter)$, which proves the property.
	
	Thus, we have proved that $(\SS', v')\abconftrans^* (\SS', v_n)$ with $v_n(\counter_{q_\ell,k_\ell})=v_{i+1}(\counter_{q_\ell,k_\ell})$, $v_n(\counter_{q_\ell})=v_{i+1}(\counter_{q_\ell})$ for all $1\leq \ell\leq n$, 
	and $v_n(\counter)=v'(\counter)$ for all other counter $\counter$, hence $v_n=v_{i+1}$. 
\end{proof}

\section{Proofs of \cref{subsec:lower-bound-target-waitonly}}
\label{app:lower-bound-target-waitonly}

The proof of~\cref{th:synchro-ack-complete} relies on the following lemma. 
\begin{lemma}\label{lemma:lower-bound-target-waitonly:completeness}
	There exist $C_0\in\mathcal{I}$ and $C_f\in \CC$ such that $C_0 \trans^\ast C_f$ in $\PP_\mathcal{V}$ and  $C_f(\aproc) = q_f$ for all $\aproc \in [1, ||C_f||]$ \emph{if and only if} $(\ell_0, \mathbf{0}) \abconftrans^\ast (\ell_f, \mathbf{0})$ in $\mathcal{V}$.
\end{lemma}
\begin{proof}[Proof of \cref{lemma:lower-bound-target-waitonly:completeness}]
	Let the run of the VASS be $(\ell_0, v_0) \abconftrans (\ell_1, v_1) \abconftrans \cdots \abconftrans (\ell_n, v_n)$ with $v_0 = v_n = \mathbf{0}$ and $\ell_n = \ell_f$. We denote $K = \max_{0 \leq i \leq n} |v_i|$ where $|v_i| = \sum_{\counter \in \Counters} v_i(\counter)$. 
	
	We build an execution of the protocol starting from the initial configuration $C_0$ with $K+1$ processes. The first part of the execution goes as follows.
	
	$C_0 \transup{1, (\qinit, !! \$, q_1)} C_1 \transup{2, (\qinit, !! \$, q_1)} \cdots \transup{K, (\qinit, !!\$, q_1)} C_{K} \transup{K+1,(\qinit, !! \textsf{start}, \ell_0)} C_{K+1} $. Observe that $C_{K+1}(j) = \textsf{zero}$ for all $1 \leq j \leq K$ and $C_{K+1}(K+1) = \ell_0$.
	
	We rename $C_{K+1}$ to $C'_{0}$. 
	
	For all $0\leq i \leq n$, let $\mathbb{C}(\ell_i,v_i)=\set{C\in \CC\mid C(K+1)=\ell_i, C([1,K])=\set{\textsf{zero}, \textsf{unit}_\counter\mid\counter\in\Counters},
	\textrm{ and for all $\counter\in\Counters$}, |C^{-1}(\textsf{unit}_\counter)|=v_i(\counter)}$ be the set of protocol configurations corresponding to 
	$(\ell_i,v_i)$. 
	% $C'_i$ be a configuration such that $C'_i(K+1)=s_i$, for all $1\leq j\leq K$, $C'_i(j)\in\set{\textsf{zero}, \textsf{unit}_\counter\mid \counter\in\Counters}$ and for all $\counter \in \Counters$, $|C_i'^{-1}(\textsf{unit}_\counter)| = v_i(\counter)$.
	We show by induction on $i$ that, for all $0\leq i \leq n$, for all $C\in\mathbb{C}(\ell_i,v_i)$, $C'_0\trans^\ast C$. 
	
	%We show that for all $0 \leq i <n$, and $C'_i$ such that $C'_i(K+1) = s_i$ and for all processes $1 \leq j \leq K$, $C'_i(j) \in \set{\textsf{zero}, \textsf{unit}_\counter \mid \counter \in \Counters}$, and furthermore for all $\counter \in \Counters$, $|C_i'^{-1}(\textsf{unit}_\counter)| = v_i(\counter)$, it holds that $C'_0 \trans^\ast C'_i$.
	%We show it by induction.
	For $i = 0$, observe that $\mathbb{C}(\ell_0,v_0)=\set{C'_0}$ hence the property trivially holds.
	Let $ 0 \leq i <n$ and $(\ell_i, \delta_i, \ell_{i+1}) \in \Delta$ such that $(\ell_i, v_i) \abconftransup{} (\ell_{i+1}, v_{i+1})$ and $v_{i+1} = v_i + \delta_i$. 
	Let $\counter_i$ be the unique counter such that $\delta_i(\counter_i)\in\set{1,-1}$. 	 Let now $C\in\mathbb{C}(\ell_{i+1}, v_{i+1})$. 
	%We restricted ourselves to the case where there exists one counter $\counter_i$ such that $\delta_i(\counter_i) \in \set{1, -1}$ and for all other counter $\counter$, $\delta_i(\counter) = 0$.
	
	\textbf{Case $\delta_i(\counter_i) = 1$.} 
	Then $v_{i+1}(\counter_i)=v_{i}(\counter_i)+1>0$, and since $C\in\mathbb{C}(\ell_{i+1}, v_{i+1})$, there exists $1\leq \aproc\leq K$ such that 
	$C(\aproc)=\textsf{unit}_{\counter_i}$. Let then $C_i\in \CC$ such that $C_i(K+1)=\ell_i$, $C_i(\aproc)=\textsf{zero}$ and for all other process $1\leq j\leq K$, $C_i(j)=C(j)$. 
	By definition of the protocol, $(\ell_i, ? \textsf{inc}_{\counter_i}, \ell_{i+1}) \in T$, hence $C_i \transup{\aproc, (\textsf{zero}, !! \textsf{inc}_{\counter_i}, \textsf{unit}_{\counter_i})}  C$.
	Furthermore, it is easy to see that,
	since $|C_i^{-1}(\textsf{unit}_{\counter_i})|=|C^{-1}(\textsf{unit}_{\counter_i})|-1=v_{i+1}(\counter_i)-1=v_i(\counter_i)$, and for all other counter $\counter$,
	$|C_i^{-1}(\textsf{unit}_{\counter})|=|C^{-1}(\textsf{unit}_{\counter})|=v_{i+1}(\counter)=v_i(\counter)$, $C_i\in\mathbb{C}(\ell_i,v_i)$. 
	Hence, by induction hypothesis, $C'_0\trans^\ast C_i\trans C$. 
%	
%	then $|v_{i}| < |v_{i+1}| \leq K$ and so there exists at least one process $\aproc_i$ such that $C'_{i}(\aproc_i) = \textsf{zero}$. Furthermore, by definition of the protocol, $(s_i, ? \textsf{inc}_{\counter_i}, s_{i+1}) \in T$, hence $C'_i \transup{\aproc_i, (\textsf{zero}, \textsf{inc}_{\counter_i}, \textsf{unit}_{\counter_i})}  C'_{i+1}$ where $C'_{i+1}$ is such that for all process $\aproc \neq \aproc_i$, $\aproc \neq K+1$, $C'_{i+1}(\aproc) = C'_i(\aproc)$ and $C'_{i+1}(K+1) = s_{i+1}$. Furthermore, $|C_{i+1}'^{-1}(\textsf{unit}_{\counter_i})| = |C_i'^{-1}(\textsf{unit}_{\counter_i})| + 1= v_i(\counter_i) + 1 =v_{i+1}(\counter_i)$, which concludes the induction proof of that case.
	
	\textbf{Case $\delta_i(\counter_i) = -1$.} Then $v_{i+1}(\counter_i) =v_i(\counter_i)-1<K$ and there exists a process $\aproc$ such that $C(\aproc) =\textsf{zero}$. Let $C_i\in \CC$ such that $C_i(K+1)=\ell_i$, $C_i(\aproc)=\textsf{unit}_{\counter_i}$ and for all other process $1\leq j\leq K$, $C_i(j)=C(j)$. 
By definition of the protocol, $(\ell_i, ? \textsf{dec}_{\counter_i}, \ell_{i+1}) \in T$, hence $C_i \transup{\aproc, (\textsf{unit}_{\counter_i}, !!\textsf{dec}_{\counter_i}, \textsf{zero}_{})}  C$. Furthermore, it is easy to see that,
	since $|C_i^{-1}(\textsf{unit}_{\counter_i})|=|C^{-1}(\textsf{unit}_{\counter_i})|+1=v_{i+1}(\counter_i)+1=v_i(\counter_i)$, and for all other counter $\counter$,
	$|C_i^{-1}(\textsf{unit}_{\counter})|=|C^{-1}(\textsf{unit}_{\counter})|=v_{i+1}(\counter)=v_i(\counter)$, $C_i\in\mathbb{C}(\ell_i,v_i)$. 
	Hence, by induction hypothesis, $C'_0\trans^\ast C_i\trans C$. 

%	
%	 $\textsf{unit}_{\counter_i}$. Furthermore, by definition of the protocol, $(s_i, ? \textsf{dec}_{\counter_i}, s_{i+1}) \in T$, hence $C'_i \transup{\aproc_i, (\textsf{unit}_{\counter_i}, \textsf{dec}_{\counter_i}, \textsf{zero}_{})}  C'_{i+1}$ where $C'_{i+1}$ is such that for all process $\aproc \neq \aproc_i$, $\aproc \neq K+1$, $C'_{i+1}(\aproc) = C'_i(\aproc)$ and $C'_{i+1}(K+1) = s_{i+1}$. Furthermore, $|C_{i+1}'^{-1}(\textsf{unit}_{\counter_i})| = |C_i'^{-1}(\textsf{unit}_{\counter_i})| - 1= v_i(\counter_i) - 1 =v_{i+1}(\counter_i)$, which concludes the induction proof of that case.
%	
	
	Then, since $\mathbb{C}(\ell_n,v_n)=\set{C'_n}$ with $C'_n(K+1)=\ell_n=\ell_f$, and $C'_n(j)=\textsf{zero}$ for all $1\leq j\leq K$, we have that 
	$C'_0\trans^\ast C'_n$.
	%one can find a sequence $C'_0 \trans C'_1 \trans \cdots \trans C'_n$ such that $C'_n(K+1) = s_n = s_f$, and $C'_n(j) = \textsf{zero}$ for all $1 \leq j \leq K$. 
	The execution ends in the following way: 
	$C'_n \transup{1, (\textsf{zero}, !!\textsf{end}, \textsf{z-end})} C'_{n+1} \transup{2, (\textsf{zero}, !!\textsf{end}, \textsf{z-end})} \cdots \transup{K, (\textsf{zero}, !!\textsf{end}, \textsf{z-end})} C'_{n+K} \transup{K+1,(s'_f,  !!\textsf{verif}, q_f)}  C'_{n+K+1}$. One can observe that for all process $\aproc$, $C'_{n+K+1}(\aproc)=q_f$ which concludes the proof.

We now prove %\cref{lemma:lower-bound-target-waitonly:soundness}, 
the other direction, which relies on the following intermediate lemma. 
\begin{lemma}\label{lemma:target-ack-hard-two-indices-broadcast}
	Let $C_0 \trans^\ast C_f$ such that $C_0 \in \mathcal{I}$ and for all $\aproc \in [1, ||C_f||]$, $C_f(\aproc) = q_f$. Then, the message $\textsf{verif}$
	is broadcast exactly once.
\end{lemma}

\begin{proof}
	We let $C_0 \trans C_1 \trans \cdots \trans C_n =C_f$ be the execution that gathers everyone in $q_f$. The broadcast of $\textsf{verif}$ can be done
	only with the transition $(\ell'_f, !!\textsf{verif}, q_f)$.
	
	First, observe that there must exist an index $0 \leq j \leq n$ such that $C_j \transup{\aproc_1, (\ell'_f, !! \textsf{verif}, q_f)} C_{j+1}$; otherwise, $q_f$ is unreachable. Let $j_1$ denote the first such index. 
	%Then as $C_n(\aproc_1) = q_f$, for all $j > j_1$, $C_j(\aproc_1) = q_f$ as otherwise $C_j(\aproc_1) = \textsf{err}$ and $C_n(\aproc_1) = \textsf{err}$.
	If there exists an index $j_2 > j_1$ such that $C_{j_2}\transup{\aproc_2, (\ell'_f, !! \textsf{verif}, q_f)} C_{j_2+1}$, then, by construction of the VASS, $\aproc_2\neq \aproc_1$ (there is now way for $e_1$ to go back to $\ell'_f$), and $C_{j_2+1}(\aproc_1) = \textsf{err}$. Then, $C_n(\aproc_1) =\textsf{err}$ which contradicts the definition of the execution.
%	By definition, $C_n(\aproc_1) = q_f$, hence, there must exists an index $k > j$ such that $C_k \transup{\aproc_2, (p, !!\textsf{election}, p')} C_{k+1}$, with $(p, p') = (\qinit, q_1)$ or $(p ,p') = (p_2, s_0)$. Consider $k$ to be the smallest such index.
%	
%	By construction, $C_{k+1}(\aproc_1) = q_f$. Furthermore, there is no other broadcast of $\textsf{election}$ as otherwise, $\aproc_1$ reaches state $\textsf{err}$ and has no chance to reach $q_f$ again, hence, if there is a third broadcast, $C_{n}(\aproc_1) = \textsf{err}$.
%	
%	Hence, there are exactly two broadcasts. If the second broadcast is done with transition $(\qinit, !!\textsf{election}, q_1)$, then the process that broadcasts $\textsf{election}$ the second (and last) time will stayed blocked on state $q_1$ forever as no other broadcast of $\textsf{election}$ happens, and so $C_n(\aproc_2) = q_1\neq q_f$. Hence, the second broadcast happens with transition $(p_2, !!\textsf{election}, s_0)$.
\end{proof}

%\begin{proof}[Proof of \cref{lemma:lower-bound-target-waitonly:soundness}]
	We let $C_0 \trans C_1 \trans \cdots \trans C_n =C_f$ be the execution that gathers all processes in $q_f$.
	From \cref{lemma:target-ack-hard-two-indices-broadcast}, there exists a unique index $j$ such that
	%	two indices $i_1 < i_2$ such that 
	$C_j \transup{\aproc_1, (\ell'_f, !! \textsf{verif}, q_f)} C_{j+1}$.
	%and $C_{i_2} \transup{\aproc_2, (p, !!\textsf{election}, p')} C_{i_2+1}$, and no other index $j$ such that $C_j \transup{\aproc, (q, !! \textsf{election}, q')} C_{j+1}$ for a process $\aproc$ and states $q, q' \in Q$.
	Hence, there exists a unique process $\aproc_1$ evolving in $\PP_{\mathcal{V}}$; indeed, any process in $\PP_{\mathcal{V}}$ needs to broadcast $\textsf{verif}$
	in order to reach $q_f$, and the only way to broadcast $\textsf{verif}$ is to go through $\PP_{\mathcal{V}}$. So there exists a unique index $j_1 < j$ such that $C_{j_1} \transup{\aproc_1, (\qinit, !!\textsf{start}, \ell_0)} C_{j_1 +1}$. Hence, all other processes receive the message $\textsf{start}$ in $C_{j_1}$; otherwise they will never reach $q_f$. As a consequence, for each process $\aproc \neq \aproc_1$, $C_{j_1 +1}(\aproc) = \textsf{zero}$.

	%	Hence for all process $\aproc$ such that $\aproc \neq \aproc_1$, 
	%	$C_{i_2+1}(\aproc) =  \textsf{zero}$. Indeed, as there is exactly one broadcast of $\textsf{election}$ by $\aproc_1$ before $i_2$, by construction, $C_{i_2}(\aproc) \in \set{\qinit, p_1, p_2, p_3}$ hence, $C_{i_2}(\aproc) \in \set{\qinit, p_1, p_2, p_3, \textsf{zero}}$. 
	%	If $C_{i_2}(\aproc) \in \set{\qinit, p_1, p_2, p_3}$, then $C_f(\aproc) \in \set{\qinit, p_1, p_2, p_3}$ as there is no other broadcast of $\textsf{election}$ which contradicts definition of $C_f$. Hence, $C_{i_2}(\aproc) = \textsf{zero}$. 
	
	%	Observe that the last step of communication is such that $C_{n-1} \transup{\aproc_3, (s'_f, !!\textsf{end}_2, q_f)} C_n$. As $\aproc_2$ is the only process taking the transition $(p_2, !!\textsf{election}, s_0)$, we conclude that $\aproc_3 = \aproc_2$. Hence, $\textsf{end}_2$ is broadcast exactly once, and so
	%	for all process $\aproc$ such that $\aproc \neq \aproc_1$ and $\aproc \neq \aproc_2$, $C_{n-1}(\aproc) = \textsf{z-end}$, as otherwise, $C_n(\aproc) \neq q_f$.
	%	Otherwise, $C_{n-1}(\aproc) \in \set{\textsf{zero}, \textsf{err}', \textsf{unit}_{\counter} \mid \counter \in \Counters}$ and so $C_n(\aproc) \in \set{\textsf{zero}, \textsf{err}', \textsf{unit}_{\counter} \mid \counter \in \Counters}$, and as there is no other broadcast of $\textsf{end}_2$, it holds that $C_f(\aproc) \in \set{\textsf{zero}, \textsf{err}', \textsf{unit}_{\counter} \mid \counter \in \Counters}$ which contradicts the fact that $C_f(\aproc) = q_f$.
	
	%	By construction, there exists a moment 

	We now prove that there exists an index $j_1 < j_2 < j$ such that $C_{j_2}(\aproc_1) = \ell_f$ and for any other process $\aproc \neq \aproc_1$, $C_{j_2}(\aproc) = \textsf{zero}$.
	First note that, unless $\ell_0=\ell_f$, there must be at least one process $\aproc$ such that $C_{j_1+1}(\aproc)=\textsf{zero}$. Otherwise, no process ever broadcasts a message, and $\aproc$ would not
	be able to reach $\ell_f$ and then $\ell'_f$. 
	
	We claim that $j_2$ is the first index such that $C_{j_2} \transup{\aproc_2, (\textsf{zero}, !!\textsf{end}, \textsf{z-end})} C_{j_2+1}$ for some process $\aproc_2$. This transition must occur during
	the execution, otherwise, processes other than $\aproc_1$ would be unable to reach $q_f$ during the unique broadcast of $\textsf{verif}$ in the execution.  First, we prove that $C_{j_2}(\aproc_1) = \ell_f$.
	If it is not the case, it needs to receive some messages in the set $\set{\textsf{inc}_\counter, \textsf{dec}_\counter \mid \counter \in \Counters}$ before broadcasting message $!!\textsf{verif}$. Observe from \cref{lemma:target-ack-hard-two-indices-broadcast}, there is exactly one broadcast of $\textsf{verif}$, and so the sending of one message in $\set{\textsf{inc}_\counter, \textsf{dec}_\counter \mid \counter \in \Counters}$ is received by $\aproc_2$ which goes to state $\textsf{err}'$. As a consequence, it never reaches $q_f$ which contradicts definition of $C_f$.  Hence, $C_{j_2}(\aproc_1) = \ell_f$.
	We prove now that for all other process $\aproc \neq \aproc_1$, $C_{j_2}(\aproc) = \textsf{zero}$. If it is not the case, since $\textsf{end}$ is not broadcast before, there is one process $\aproc'$ such that $C_{j_2}(\aproc') \in \set{\textsf{unit}_\counter \mid \counter \in \Counters}$. In that case, for the same reason as before, it will need to broadcast a message from $\set{\textsf{dec}_{\counter} \mid \counter \in \Counters}$ while $\aproc_2$ is in $\textsf{z-end}$. This would lead $\aproc_2$ to $\textsf{err}'$ and contradicts the definition of $C_f$.

	Hence, $C_{j_2}(\aproc_1) = \ell_f$ and for all other process $\aproc \neq \aproc_1$, $C_{j_2}(\aproc) = \textsf{zero}$.
	
	We rename the configurations visited during this part of the execution to $C_{j_1+1}=C'_0\trans C'_1\trans \cdots \trans C'_k=C_{j_2}$, with $k=j_2-j_1-1$. First observe that since $C'_0(\aproc_1)=\ell_0$,
	and $C'_k(\aproc_1)=\ell_f$, the structure of the VASS implies that $C'_i(\aproc_1)\in\Loc$ for all $0\leq i\leq k$. Similarly, since, for all $\aproc\neq\aproc_1$, $C'_0(\aproc)=C'_k(\aproc)=\textsf{zero}$,
	the structure of the VASS implies that $C'_i(\aproc)\in\set{\textsf{zero}, \textsf{unit}_\counter\mid\counter\in\Counters}$, for all $\aproc\neq\aproc_1$, and for all $0\leq i\leq k$. We build now, by induction on 
	$0\leq i\leq k$, an execution $(\ell_0, v_0)\abconftrans (\ell_1,v_1)\abconftrans \dots \abconftrans (\ell_k, v_k)$ of the VASS such that $\ell_i=C'_i(\aproc_1)$, and $v_i(\counter)=|C_i'^{-1}(\textsf{unit}_\counter)|$ for all $\counter \in \Counters$.
	
	With $v_0=\mathbf{0}_\Counters$, the definition of $C'_0$ allows to conclude immediately for the case $i=0$. Let now $0\leq i < k$ and assume the property true. Let $C'_i \transup{\aproc, t} C'_{i+1}$
	the transition taken in the execution of the protocol. Since, by construction of the protocol, $C'_i(\aproc_1)\in Q_W$, we know that $\aproc\neq\aproc_1$, and then 
	$t\in\set{(\textsf{zero}, !! \textsf{inc}_\counter, \textsf{unit}_\counter), (\textsf{unit}_\counter, !!\textsf{dec}_\counter, \textsf{zero})\mid\counter\in\Counters}$.
	We distinguish two cases:
	
	\begin{itemize}
	\item $t = (\textsf{zero}, !!\textsf{inc}_\counter, \textsf{unit}_\counter)$ for some $\counter \in \Counters$. In that case, by definition of the protocol, $(\ell_i, \delta, \ell_{i+1}) \in \Delta$ with 
	$\delta(\counter) = 1$ and $\delta(\counter') = 0$ for all other counters $\counter'$. Furthermore, $|C'^{-1}_{i+1}(\textsf{unit}_\counter)| = |C'^{-1}_{i}(\textsf{unit}_\counter)| +1 = v_i(\counter) + 1$ by induction
	hypothesis, and for all other counters $\counter'$, $|C'^{-1}_{i+1}(\textsf{unit}_{\counter'})| = |C'^{-1}_{i}(\textsf{unit}_{\counter'})|=v_i(\counter')$, by induction hypothesis. Hence, $(\ell_i, v_i) \abconftrans (\ell_{i+1}, v_{i+1})$ with $v_{i+1}(\counter) = |C_{i+1}'^{-1}(\textsf{unit}_\counter)|$ for all $\counter \in \Counters$.
	
	\item $t = (\textsf{unit}_\counter, !!\textsf{dec}_\counter, \textsf{zero})$ for some $\counter \in \Counters$. In that case, by definition of the protocol, $(\ell_i, \delta, \ell_{i+1}) \in \Delta$ with $\delta(\counter) = -1$ and $\delta(\counter') = 0$ for all other counters $\counter'$. Furthermore, $|C'^{-1}_{i+1}(\textsf{unit}_\counter)| = |C'^{-1}_{i}(\textsf{unit}_\counter)| -1 = v_i(\counter) - 1 \geq 0$, by induction hypothesis.
	For all other counters $\counter'$, $|C'^{-1}_{i+1}(\textsf{unit}_{\counter'})| = |C'^{-1}_{i}(\textsf{unit}_{\counter'})|=v_i(\counter')$, by induction hypothesis. Hence, $(\ell_i, v_i) \abconftrans (\ell_{i+1}, v_{i+1})$ where $v_{i+1}(\counter) = |C_{i+1}'^{-1}(\textsf{unit}_\counter)|$ for all $\counter \in \Counters$.
	\end{itemize}
	
	Hence, $(\ell_0, v_0) \abconftrans^\ast (\ell_k, v_k)$ with $\ell_k = C_k'(\aproc_1) = C_{j_2}(\aproc_1) = \ell_f$ and $v_k(\counter ) = |C_k'^{-1}(\textsf{unit}_\counter)| = 0$ for all counters $\counter$, which concludes the proof.

\end{proof}

\section{Proof of \cref{lemma:synchro-action:mutu-reach:reduction}}

We prove that if there are runs $(s_0, \mathbf{0})\abconftrans^\ast(s_f, \mathbf{0})$ and 
$(s_f, \mathbf{0})\abconftrans^\ast(s_0, \mathbf{0})$ in the VASS $\mathcal{V}_\PP$, there is an execution in $\PP$ that leads all the processes in $q_f$. 
%We provide here some insight on the proof of \cref{lemma:synchro-action:mutu-reach:reduction}, more precisely on the right-to-left sense: building an execution of $\PP$ from the two runs of the VASS. 

The idea is that from a run of the VASS going from $(s_0, \mathbf{0})$ to $(s_f, \mathbf{0})$, one can build a run from $(s_0, \mathbf{0})$ to $(s_f, \mathbf{0})$ \emph{that never takes the transition $(s_f, \mathbf{0}, s_0)$}. Hence this new run will necessarily be of the form:
$(s_0, \mathbf{0}) \abconftrans^\ast (\emptyset, v_{init}) \abconftrans^\ast (\emptyset, v_f) \abconftrans (s'_f, v'_f) \abconftrans (s_f, v_f) \abconftrans^\ast (s_f, \mathbf{0})$ with no visit of $(s_f, \mathbf{0}, s_0)$. Here, $v_{init}(\counter) = 0$ for all counters $\counter \neq \counter_{\qinit}$ and $v_f(\counter) = 0$ for all counters $\counter \neq \counter_{q_f}$. An easy adaptation of the proof of~\cref{lemma:soundness} allows then to build an execution of the protocol from $
C_0 \in \mathcal{I}$ to $C_f$ such that for all $\aproc \in [1, ||C_f||]$, $C_f(\aproc) = q_f$.

We now focus on transforming the run of the $\mathcal{V}_\PP$.  From the VASS construction, we can decompose the run into: 
$(s_0, v_0) \abconftrans^\ast (s_f, v_1)\abconftrans (s_0, v_1)\abconftrans^\ast (s_f, v_2)\dots
(s_0, v_{n-1}) \abconftrans^\ast (s_f,v_n)$ with $v_0 = v_n = \mathbf{0}$.
Moreover, for all $0 \leq i < n$ there exist vectors $v'_i$, $v''_i$ and $v'''_i$ such that:
$(s_0, v_i) \abconftrans^\ast (s_0, v'_i) \abconftrans (\emptyset, v'_i) \abconftrans^\ast (\emptyset, v''_{i}) \abconftrans (s'_f, v'''_i) \abconftrans (s_f, v''_i) \abconftrans^\ast (s_f,v_{i+1})$, with (i) $v'_i(\counter_{\qinit}) = v_i(\counter_{\qinit}) + k$ for some $k \in \nat$ and for all other counters $\counter$, $v'_i(\counter) = v_i(\counter)$, and (ii) $v''_i(\counter_{q_f}) = v_{i+1}(\counter_{q_f}) + j$ for some $j \in \nat$ and for all other counters $\counter$, $v''_i(\counter) = v_{i+1}(\counter)$.

The key observation is to note that for all $s , s' \in \Loc$, for all counter valuations $v_1, \Delta_1, v_2$ %such that $v_1 \preceq v_1'$ (where $\preceq$ denotes the component-wise order), if $(s, v_1) \abconftrans^\ast (s',v_2)$ then \nas{j'ai remplacé $v_1$ par $v'_1$}$(s, v'_1) \abconftrans^\ast (s', v_2 + (v_1' - v_1))$. This allows us to reorder the increments and decrements in the run.
such that $\mathbf{0}\preceq\Delta_1$, with $\preceq$ the component-wise order, if $(s, v_1) \abconftrans^\ast (s',v_2)$ then $(s, v_1+\Delta_1) \abconftrans^\ast (s', v_2 + \Delta_1)$.
By reordering, we construct a run where all increments of $\counter_{\qinit}$ occur at the beginning and all decrements of $\counter_{q_f}$ occur at the end. 

For all $0\leq i<n$, let $\Delta_i$ be defined by:
\begin{itemize}
\item $\Delta_i(\counter_{\qinit})=\sum_{i<j< n} v'_j(\counter_{\qinit})-v_j(\counter_{\qinit})$
\item $\Delta_i(\counter_{q_f})=\sum_{0\leq j<i} v''_j(\counter_{q_f})-v_{j+1}(\counter_{q_f})$
\item $\Delta_i(\counter)=0$ for all other $\counter\in\Counters$.
\end{itemize}

Observe first that $\mathbf{0}\preceq \Delta_i$, and that, for all $0\leq i< n$, $v''_i+\Delta_i=v'_{i+1}+\Delta_{i+1}$. Then, from the different parts of the run 
$(\emptyset, v'_i)\abconftrans^\ast (\emptyset, v''_i)$, we deduce that $(\emptyset, v'_i+\Delta_i)\abconftrans^\ast (\emptyset, v''_i+\Delta_i)$, and we can hence build
an execution $(s_0, \mathbf{0})\abconftrans^\ast (s_0, v'_0+\Delta_0)\abconftrans (\emptyset, v'_0+\Delta_0)\abconftrans^\ast (\emptyset, v'_1+\Delta_1)\dots\abconftrans^\ast
(\emptyset, v''_{n-1}+\Delta_{n-1})$. 

Since in the original run $(\emptyset, v'_{n-1}) \abconftrans^\ast (\emptyset, v''_{n-1}) \abconftrans (s'_f, v'''_{n-1}) \abconftrans (s_f, v''_{n-1}) \abconftrans^\ast (s_f,v_n)$, with $v_n=\mathbf{0}$, we deduce that $v''_{n-1}(\counter_{\qinit})=\mathbf{0}$, $v''_{n-1}(\counter_{q_f})>0$, and $v''_{n-1}(\counter)=0$ for all
other $\counter\in\Counters$. Moreover, $\Delta_{n-1}(\counter)=0$ for all $\counter\neq \counter_{q_f}$, hence can continue the execution 
$(\emptyset, v''_{n-1}+\Delta_{n-1})\abconftrans (s'_f, v'''_{n-1}+\Delta_{n-1})\abconftrans (s_f, v''_{n-1}+\Delta_{n-1})\abconftrans^\ast (s_f, \mathbf{0})$. 

%\nastext{ancienne version}
%This results in a shortened run:\nas{je ne comprends pas ce nouveau run, il me semble qu'il prend encore la transition $(s_f,\mathbf{0}, s_0)$}
%$(s_0, \mathbf{0}) \abconftrans^\ast (s_0, w_0) \abconftrans (\emptyset, w_0) \abconftrans^\ast (\emptyset, w_1) \abconftrans (s'_f, w'_1) \abconftrans (s_f, w_1) \abconftrans (s_0, w_1) \abconftrans (\emptyset, w_1) \abconftrans^\ast (\emptyset, w_2) \abconftrans \cdots \abconftrans (s_f, w_n) \abconftrans^\ast (s_f, \mathbf{0})$ for a sequence of vectors $w_0, w_1, w'_1, w_2, w'_2, \dots,w_n$ such that:
%\begin{align*}
%	w_i(\counter) &= v_i(\counter)  \textrm{ for all } \counter \nin \set{\counter_{\qinit}, \counter_{q_f}}\\
%	w_i(\counter_{\qinit}) &= v_i(\counter_{\qinit}) + \sum_{i \leq j \leq n} (v'_i(\counter_{\qinit}) - v_i(\counter_{\qinit}))\\
%	w_i(\counter_{q_f}) &= v''_i(\counter_{q_f}) + \sum_{0 \leq j < i} (v''_i(\counter_{q_f}) - v_{i+1}(\counter_{q_f})) \textrm{ for }i >0\\
%	w_0(\counter_{q_f}) &= 0
%\end{align*}
%And $w'_i(\counter_{q_f}) = w_i(\counter_{q_f}) - 1$ and $w'_i(\counter) = w_i(\counter)$ for all other counters $\counter$.
%\nastext{fin ancienne version}
This run avoids the transition $(s_f, \mathbf{0}, s_0)$.

Observe that if $q_f \in Q_W$, then this transformation is not possible, as run segments between $s_0$ and $s_f$ would be of the form $(s_0, v) \abconftrans^\ast (\emptyset, v_1) \abconftrans^\ast (\set{S_f}, v_2) \abconftrans (s'_f, v_3) \abconftrans (s_f, v_3)$. As a result, we can not shorten a run by gluing a configuration with location $\set{S_f}$ and one with location $\emptyset$.

\section{Proofs of \cref{sec:pspace-hard}}

We prove here the soundness and correctness of the reduction presented in \cref{sec:pspace-hard}.
Denote $\mathcal{L}(\mathcal{A})$ the language of automaton $\mathcal{A}$, i.e. $\mathcal{L}(\mathcal{A}) = \set{w\mid \Delta^\ast(q^0, w) = q^f}$ where $q^0$ is its initial state and $q^f$ its final state.
\begin{lemma}
	If there exists a word $w \in \bigcap_{1\leq i \leq n} \mathcal{L}(\mathcal{A}_i)$, then there exists an execution $\rho \in \Lambda_\omega$ and $\aproc \in [1, \nbagents{\rho}]$ such that for all $i \in \nat$, there exists $j > i$ such that $\rho_j \transup{\aproc, (q_{n+2}, !! \smiley, q_1)} \rho_{j+1}$.
\end{lemma}
\begin{proof}
	We denote $w = a_1\dots a_k$.
	Let $C_0$ be the initial configuration with $n+2$ processes. We form the following prefix of execution:
	\[
		C_0 \transup{1, (\qinit, !!1, q_1^0)} C_1 \transup{2, (\qinit, !!2, q_2^0)} \cdots \transup{n, (\qinit, !!n, q_n^0)} C_n \transup{n+1, (\qinit, !!\#, q_1)} C_{n+1}
	\]
	$C_{n+1}$ is such that: $C_{n+1}(i) = q_i^0$ for all $i \in [1,n]$, and $C_{n+1}(n+1) = q_1$ and $C_{n+1}(n+2) = \qinit$.
	Next, consider the following sequence of configurations:
	\begin{align*}
		C_{n+1} &\transup{n+2, (\qinit, !!a_1, \qinit)} C_{n+2} \transup{n+2, (\qinit, !!a_2, \qinit)} \cdots \transup{n+2, (\qinit, !!a_k, \qinit)} C_{n+k+1}\\
		& \transup{n+2, (\qinit, !!\$, \qinit)} C_{n+k+2} \transup{1, (q_1^{\$}, !!\text{end}_1, q_1^e)} \cdots \transup{n, (q_n^{\$}, !!\text{end}_n, q_n^e)} C_{2n+k+2}\\
		&\transup{n+1, (q_{n+2}, !!\smiley, q_1)} C_{2n+k+3} 
	\end{align*}
As $a_1 \dots a_k = w \in \bigcap_{1\leq i \leq n} \mathcal{L}(\mathcal{A}_i)$, it holds that: $C_{n+k+1}(i) = q_i^f$ for all $i \in [1 , n]$. Hence, $C_{n+k +2}(i) = q_i^\$ $ for all $i \in [1, n]$. Furthermore, $C_{n+k+2}(n+1) = q_2$. We conclude that $C_{2n+k+3} = C_{n+1}$. Hence, we can form an infinite execution $\rho = C_0 C_1 \dots C_{n+1} \dots C_{2n+k+3} \dots C_{3n+k+5} \dots$ such that for all $i \in \nat$, $C_{n+1 +i(n+k+2)} = C_{n+1}$ and $C_{n +i(n+k+2)} \transup{n+1, (q_{n+2}, !! \smiley, q_1)} C_{n+1 + i(n+k+2)}$, which concludes the proof.

\end{proof}

\begin{lemma}
	If there exists an execution $\rho \in \Lambda_\omega$ and $\aproc \in [1, \nbagents{\rho}]$ such that for all $i \in \nat$, there exists $j > i$ such that $\rho_j \transup{\aproc, (q_{n+2}, !! \smiley, q_1)} \rho_{j+1}$, then there exists $w \in \bigcap_{1\leq i \leq n} \mathcal{L}(\mathcal{A}_i)$.
\end{lemma}
\begin{proof}
	Denote $\rho$ such an execution.
	We start by proving the following lemma.
	\begin{lemma}
		There exists $i_0 \in \nat$, from which no new processes reach state $\frownie$, in other words, for all process $\aproc$, if $\rho_{i_0}(\aproc) \neq \frownie$ then $\rho_i(\aproc) \neq \frownie$ for all $i > i_0$.
	\end{lemma}
	\begin{proof}
		We reason by contradiction. Otherwise, for all $i_0 \in \nat$, either all processes are on $\frownie$, or there exists an process not in $\frownie$ who reach it later. As $\frownie$ is a sink state and there is a finite number of processes, eventually all processes are on $\frownie$, leading to a deadlock configuration, which contradicts the fact that $\rho \in \Lambda_\omega$.
	\end{proof}
	We now denote $\aproc_0$ a process taking infinitely often transition $(q_{n+2}, !! \smiley, q_1)$.
	We easily see that for all $i \in \nat$, $\rho_i(\aproc_0) \neq \frownie$, as otherwise, $\rho_j(\aproc_0) \neq \frownie$ for all $j >i$.
	Consider now $i_1, i_2, i_3 \in \nat$ three indexes such that $i_0 < i _1 <i_2 <i_3 $.
	and $\rho_{i_j} \transup{\aproc_0, (q_{n+2}, !!\smiley, q_1)} \rho_{i_j +1}$ for $j \in \set{1,2,3}$.
	% and $\aproc_0$ does not take the transition $(q_{n+2}, !!\smiley, q_1)$ between $\rho_{i_1 +1}$ and $\rho_{i_2}$ nor between $\rho_{i_2 +1}$ and $\rho_{i_3}$. 
%	We suppose wlog that $i_1$ is not the first index from which 
	By construction of the protocol,
	%there exists $i_4 \in [i_1+1, i_2-1]$ such that $\rho_{i_4}(\aproc_0) = q_1$ and $\rho_{i_4 +1 }(\aproc_0) = q_2$.
	there exist $\aproc_1, \dots, \aproc_n$ which respectively broadcast messages $\text{end}_1, \dots, \text{end}_n$ which are received by $\aproc_0$ between $i_1$ and $i_2$.
	Observe that from the protocol's construction, once process $\aproc_0$ is on state $q_2$, the only sequence of broadcast messages that does not put $\aproc_0$ (before it reaches $q_{n+2}$) in $\frownie$ is $\text{end}_1 \cdot \text{end}_2 \cdots \text{end}_n  $.
	Meaning that there exists $i_1 < i_4 < i_2$ such that $C_{i_4}(\aproc_0) = q_2$ and:
	\[
	C_{i_4} \transup{\aproc_1, (q_1^\$, !!\text{end}_1, q_1^e)} \transup{\aproc_2, (q_2^\$, !!\text{end}_2, q_2^e)} \cdots \transup{\aproc_n, (q_n^\$, !!\text{end}_1, q_n^e)} C_{i_4 +n}
	\]
	as any other behaviour would put $\aproc_0$ in $\frownie$. Observe that for all $i \in [1,n]$, $C_{i_4 +n}(\aproc_i) = q_i^e$.
	Note that $i_4 + n = i_2$ as if $C_{i_4 + n} \transup{(q, !!m, q')} C_{i_4 + n+ 1}$ with $m\neq \smiley$, then $C_{i_4 +n}(\aproc_n) = \frownie$. Hence:
	\begin{align*}
		&C_{i_2}(\aproc_i) = q_i^0 \textrm{ for all } i \in [1,n]\\
		&C_{i_2}(\aproc_0) = q_1
	\end{align*}
	
%	 Hence, there exists $i_3 \leq i_2$ such that $\rho_{i_3}(\aproc_0) = q_{n+2}$, and $\rho_{i_3}(\aproc_i) = q_i^e$. The next broadcast of $\smiley$ leads $\aproc_1, \dots, \aproc_n$ to respectively, $q_1^0, \dots, q_n^0$.

	Denote $i_5$ the next index from which $\$$ is broadcast (it exists as, in particular, $\aproc_0$ needs to take the cycle with $(q_2, ?\text{end}_1, q_3)$).
	
	As $i_3 > i_0$, no process $\aproc_i$ reaches $\frownie$, hence $\rho_{i_5}(\aproc_i) = q_i^f$ as otherwise, $\rho_{i_5 + 1}(\aproc_i) = \frownie$ for all $i\in [1,n]$. 
	
	Observe that between $i_3 +1$ and $i_5$, the only broadcasts are of the form $(\qinit, !!a, \qinit)$ with $a \in \Sigma$ as any other broadcast would put $\aproc_0$ in $\frownie$. 
	
	Hence, the sequence of letters $a \in \Sigma$ broadcast between $i_3+1$ and $i_5$ is received by all processes $\aproc_1, \dots, \aproc_n$ (as automata are complete), and thanks to the protocol's construction, it forms a word accepted by all automata.
\end{proof}

\end{document}